\documentclass[conference]{IEEEtran}
\onecolumn
\usepackage{latexsym}
\usepackage{cite}
\usepackage{amssymb}
\usepackage{bm}
\usepackage[cmex10]{amsmath}
\usepackage{amsthm}
\usepackage[dvips]{graphics}
\usepackage{graphicx}
\usepackage{psfrag}
\usepackage{setspace}
\usepackage{array}
 \allowdisplaybreaks
\usepackage[OT1]{fontenc}
\usepackage[top= 1in, left= 1in, right= 1in, bottom= 1 in]{geometry}
\usepackage{multicol}
\usepackage{algorithm}
\usepackage{algpseudocode}
\usepackage{fancybox}
\usepackage{pdfsync}
\usepackage[small]{caption}

\theoremstyle{definition}

\newtheorem{definition}{Definition}

\newtheorem{theorem}{Theorem}
\newtheorem{lemma}{Lemma}

\newtheorem{example}{Example}

\newcommand{\cx}{{\cal X}}
\newcommand{\cy}{{\cal Y}}
\newcommand{\cs}{{\cal S}}
\newcommand{\iid}{\textrm{i.i.d.\ }}

\newcommand{\T}{{\cal T}_{\epsilon}^{(n)}}
\newcommand{\tS}{{V}}

\newcommand{\E}[1]{{\mathbb E}\Big[{#1}\Big]}
\newcommand{\D}[1]{{\mathbb D}\big({#1}\big)}
\newcommand{\blambda}{{\bm \lambda}}
\newcommand{\bmu}{{\bm \mu}}
\newcommand{\balpha}{{\bm \alpha}}
\newcommand{\bq}{{\bf q}}
\newcommand{\bQ}{{\bf Q}}
\newcommand{\bx}{{\bf x}}
\newcommand{\bp}{{\bf p}}

\begin{document}

\title{\mbox{}\\ Channel Coding and Source Coding With Increased Partial Side Information}
\author{Avihay Shirazi, Uria Basher and Haim Permuter}
\let\thefootnote\relax\footnotetext{Avihay Shirazi, Uria Basher and Haim Permuter are with the Department of Electrical and Computer Engineering at the Ben Gurion University of the Negev, Beer Sheva, Israel. Emails: avihays@bgu.ac.il, basher@bgu.ac.il, haimp@bgu.ac.il
The material in this paper was presented in part at the Allerton Conference on Communication,
Control, and Computing, September  2010}
%
\pagestyle{plain}
\setcounter{page}{1}
\pagenumbering{arabic}
\maketitle
\begin{abstract}
 Let $(S_{1,i},S_{2,i})\sim {\rm i.i.d}\  p(s_1,s_2)$, $i=1,2,\dots$ be a memoryless, correlated partial side information sequence. In this work we study channel coding and source coding problems where the partial side information $(S_1, S_2)$ is available at the encoder and the decoder, respectively, and, additionally, either the encoder's or the decoder's side information is increased by a limited-rate description of the other's partial side information. We derive six special cases of channel coding and source coding problems and we characterize the capacity and the rate-distortion functions for the different cases. We present a duality between the channel capacity and the rate-distortion cases we study.
In order to find numerical solutions for our channel capacity and rate-distortion problems, we use the Blahut-Arimoto algorithm and convex optimization tools. As a byproduct of our work, we found
a tight lower bound on the Wyner-Ziv solution by formulating its Lagrange dual as a geometric program.
Previous results in the literature provide a geometric programming formulation that is only a lower bound, but not necessarily tight.
Finally, we provide several examples corresponding to the channel capacity and the rate-distortion cases we presented.
\end{abstract}
\begin{center} \small
  {\bf Index Terms}
\end{center}
{\small Blahut-Arimoto algorithm, channel capacity, channel coding, convex optimization, duality, Gelfand-Pinsker channel coding, geometric programming, partial side information, rate-distortion, source coding, Wyner-Ziv source coding.}
\section{Introduction}
In this paper we investigate point-to-point channel models and rate-distortion problem models where both users have different and correlated partial side information and where, in addition, a rate-limited description of one of the user's side information is delivered to the other user. We then show the duality between the channel models and the rate-distortion models we investigate. In the process of investigating the rate-distortion problems, we found a tight lower bound on the rate-distortion of the Wyner-Ziv\cite{1055508} problem. We show here that it is possible to write the Lagrange dual of the Wyner-Ziv rate-distortion function as a geometric program. Then, we show that the optimal solution of this geometric program is the correct solution of the Wyner-Ziv problem.

For the convenience of the reader, we refer to the state information as the side information, to the partial side information that is available to the encoder as the encoder's side information (ESI) and to the partial side information that is available to the decoder as the decoder's side information (DSI). To the rate-limited description of the other user's side information we refer as the increase in the side information. For example, if the decoder is informed with its DSI and, in addition, with a rate-limited description of the ESI, then we would say that the decoder is informed with {\it increased} DSI.

To make the motivation for this paper clear, let us look at a simple example, as depicted in Figure~\ref{fig:Partial_SI}. Two remote users, User $1$ - the encoder and User $2$ - the decoder, want to communicate between them over a channel that is being interrupted by two interrupters, Interrupter $1$ and Interrupter $2$. We allow the interruptions $S_1$ and $S_2$ generated by the interrupters to be correlated, i.e., $(S_1,S_2) \sim p(s_1,s_2)$. Assume that Interrupter $1$ is located in close proximity to User $1$ and can fully describe its future interruption, $S_1$, to User $1$ and that Interrupter $2$ is located in close proximity to User $2$ and can also fully describe its future interruption, $S_2$, to user $2$. In addition, assume that Interrupter $1$ can increase the side information of User $2$ with rate-limited information about its interruption. In these circumstances, we  pose the question; what is the capacity of the channel between User $1$ and User $2$? We extensively discuss the answer to this question in the forthcoming sections.

\begin{figure}[t]
  \centering
  \small
  \psfrag{e}{Encoder}
  \psfrag{d}{Decoder}
  \psfrag{i}{Interrupter $1$}
  \psfrag{j}{Interrupter $2$}
  \psfrag{r}{Rate}
  \psfrag{c}{Channel}
  \psfrag{p}{$p(y|x,s_1,s_2)$}
  \psfrag{w}{$W$}
  \psfrag{u}{$\hat W$}
  \psfrag{s}{$\begin{array}{c}\mbox{ESI} \\ S_1 \end{array}$}
  \psfrag{z}{$\begin{array}{c}\mbox{DSI} \\ S_2 \end{array}$}
  \includegraphics[scale = 0.9]{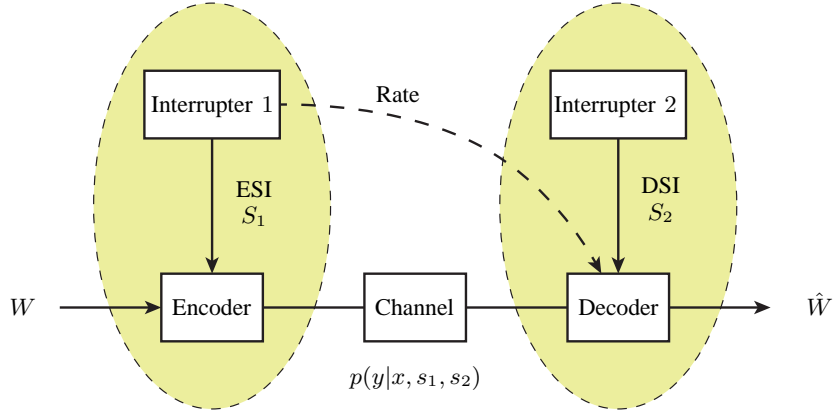}
  \caption{{\it Increased partial side information example}. The encoder wants to send a message to the decoder over an interrupted channel in the presence of side information. The encoder is provided with the ESI and the decoder is provided with increased DSI. i.e., the decoder is informed with a rate-limited description of the ESI in addition to the DSI.}
  \label{fig:Partial_SI}
\end{figure}

\subsection{Channel capacity in the presence of state information}
\label{sec:chann-capac-pres}
The three problems of channel capacity in the presence of state information that we adress in this paper are presented in Figure~\ref{fig:Capacity}. We make the assumption that the encoder is informed with partial state information, the ESI ($S_1$), and the decoder is informed with different, but correlated, partial state information, which is the DSI ($S_2$). The channel capacity problem cases are:
\begin{itemize}
\item Case 1: The decoder is provided with increased DSI; i.e., in addition to the DSI, the decoder is also informed with a rate-limited description of the ESI.
\item Case 2: The encoder is informed with increased ESI.
\item Case 2$_C$: Similar to Case 2, with the exception that the ESI is known to the encoder in a causal manner. Notice that the rate-limited description of the DSI is still known to the encoder noncausally.
\end{itemize}

We will subsequently provide the capacity of Case 1 and Case 2$_C$ and caracterize the lower and the upper bounds on Case 2, which differ only by a Markon relation.
The results for the first case under discussion, Case 1, can be concluded from Steinberg's problem~\cite{4608994}. In \cite{4608994}, Steinberg introduced and solved the case in which the encoder is fully informed with the ESI and the decoder is informed with a rate-limited description of the ESI. Therefore, the innovation in Case 1 is that the decoder is also informed with the DSI. The solution for this problem can be derived by considering the DSI to be a part of the channel's output in Steinberg's solution. In the proof of the converse in his paper, Steinberg uses a new technique that involves using the Csisz\'ar sum twice in order to get to a single-letter bound on the rate. We shall use this technique to present a duality in the converse of the Gelfand-Pinsker~\cite{citeulike:437050} and the Wyner-Ziv~\cite{1055508} problems, which, by themselves, constitute the basis for most of the results in this paper. In~\cite{1055508}, Wyner and Ziv present the rate-distortion function for data compression problems with side information at the decoder. We make use of their coding scheme in the achievability proof of the lower bound of Case 2 for describing the ESI with a limited rate at the decoder.  In~\cite{citeulike:437050}, Gelfand and Pinsker present the capacity for a channel with noncausal CSI at the encoder. We use their coding scheme in the achievability proof of Case1 and the lower bound of Case 2 for transmitting information over a channel where the ESI is the state information at the encoder. Therefore, we combine in our problems the Gelfand-Pinsker and the Wyner-Ziv problems. Another related paper is~\cite{1662378}, in which Shannon presented the capacity of a channel with causal CSI at the transmitter. We make use of Shannon's result in the achievability proof of Case 2$_C$ for communicating over a channel with causal ESI at the encoder. We also use Shannon's strategies~\cite{1662378}, for developing an iterative algorithm to calculate the capacity of the cases we present in this paper.

Some related papers that can be found in the literature are mentioned herein. Heegard and El Gamal~\cite{DBLP:journals/tit/HeegardG83} presented a model of a state-dependent channel, where the transmitter is informed with the CSI at a rate limited to $R_e$ and the receiver is informed with the CSI at a rate limited to $R_d$. This result relates to Case 1, Case 2 and Case 2$_C$ since we consider the rate-limited description of the ESI or the DSI as  side information known at both the encoder and the decoder. Cover and Chiang~\cite{Cover:2006:DCC:2263239.2266915} extended the Gelfand-Pinsker problem and the Wyner-Ziv problem to the case where both the encoder and the decoder are provided with different, but correlated, partial side information. They also showed a duality between the two cases, which is a topic that will be discussed later in this paper. Rozenzweig, Steinberg and Shamai~\cite{1424319}  and Cemal and Steinberg~\cite{4385766} studied channels with partial state information at the transmitter.
A detailed subject review on channel coding with state information was given by Keshet, Steinberg and Merhav in~\cite{1454716}.

In addition to these three cases, we also present a more general case, where the encoder is informed with increased ESI and the decoder is informed with increased DSI; i.e., there is a rate-limited description of the ESI at the decoder and there is a rate-limited description of the DSI at the encoder. We provide an achievability scheme that bounds the capacity for this case from below, however, this bound does not coincide with the capacity and, therefore, this problem remains open.
\subsection{Rate-distortion with side information}
\label{sec:rate-distortion-with}
In this paper we adress three problems of rate-distortion with side information, as presented in Figure~\ref{fig:Source_Coding}. In common with the channel capacity problems, we assume that the encoder is informed with the ESI ($S_1$) and the decoder is informed with the DSI ($S_2$), where the source, $X$, the ESI and the DSI are correlated. The rate-distortion problem cases we investigate in this paper are:
 \begin{itemize}
\item Case 1: The decoder is provided with increased DSI.
\item Case 1$_C$: Similar to Case 1, with the exception that the ESI is known to the encoder in a causal manner. The rate-limited description of the ESI is still known to the decoder noncausally.
\item Case 2: The encoder is informed with increased ESI.
\end{itemize}
Case 2 is a special case of Kaspi's~\cite{DBLP:journals/tit/Kaspi85} two-way source coding for $K=1$. In~\cite{DBLP:journals/tit/Kaspi85}, Kaspi introduced a model of multistage communication between two users, where each user may transmit up to $K$ messages to the other user, dependent on the source and the previous received messages. For Case 2, we can consider sending the rate-limited description of the DSI as the first transmission and then, sending a function of the source, the ESI and the rate-limited description of the DSI as the second transmission. This fits into Kaspi's problem for $K=1$ and thus Kaspi's theorem also applies to Case 2. Kaspi's problem was later extended by Permuter, Steinberg and Weissman~\cite{5466529} to the case where a common rate-limited side information message is being conveyed to both users. Another strongly related paper is Wyner and Ziv's paper~\cite{1055508}. In the achievability of Case 1 we use the Wyner-Ziv coding scheme twice; once for describing the ESI at the decoder where the DSI is the side information and once for the main source and the ESI where the DSI is the side information. The rate-limited description of the ESI is the side information provided to both the encoder and the decoder. In~\cite{Cover:2006:DCC:2263239.2266915} there is an extension to the Wyner-Ziv problem to the case where both the encoder and the decoder are provided with correlated partial side information.  Weissman and El Gamal~\cite[Section 2]{DBLP:journals/tit/WeissmanG06} and Weissman and Merhav~\cite{DBLP:journals/tit/WeissmanM05} presented source coding with causal side information at the decoder, which relates to Case 1$_C$.

%
As with the channel capacity, we present a bound on the general case of rate-distortion with two-sided increased partial side information. In  this problem setup the encoder is informed with a rate-limited description of the DSI in addition to the ESI and the decoder is informed with a rate-limited description of the ESI in addition to the DSI. We present an achievability scheme that bounds the optimal rate from above, however, this bound does not coincide with the optimal rate and, therefore, this problem remains open.

\subsection{Duality}
\label{sec:duality}
Within the scope of this work we point out a duality relation between the channel capacity and the rate-distortion cases we discuss. The operational duality between channel coding and source coding was first mentioned by Shannon \cite{Shannon59}.  In~\cite{DBLP:journals/tit/PradhanCR03}, Pradhan, Chou and Ramchandran studied the functional duality between some cases of channel coding and source coding, including the duality between the Gelfand-Pinsker problem and the Wyner-Ziv problem. This duality was also described by Cover and Chiang in~\cite{Cover:2006:DCC:2263239.2266915}, where they provided a transformation that makes  duality between channel coding and source coding with two-sided state information apparent.
Zamir, Shamai and Erez~\cite{Zamir:2006:NLC:2263239.2266893} and Su, Eggers and Girod~\cite{911306} utilized the duality between channel coding and source coding with side information to develop coding schemes for the dual problems.

In our paper we show that the channel capacity cases and the rate-distortion cases we discuss are operational duals in a way that strongly relates to the Wyner-Ziv and Gelfand-Pinsker duality. We also provide a transformation scheme that shows this duality in a clear way. Moreover, we show a duality relation between Kaspi's problem and Steinberg's~\cite{4608994} problem by showing a duality relation between Case 2 source coding and Case 1 channel coding. Also, we show duality in the converse parts of the Gelfand-Pinsker and the Wyner-Ziv problems. We show that both converse parts can be proven in a perfectly dual way by using the Csisz\'ar sum twice.

\subsection{Computational algorithms}
\label{sec:geom-progr-1}
Calculating channel capacity and rate-distortion problems, in general, and the Gelfand-Pinsker and the Wyner-Ziv problems, in particular, is not straightforward.
Blahut~\cite{Blahut72computationof} and Arimoto~\cite{Arimoto1972a} suggested an iterative algorithm (to be referred to as the B-A algorithm) for numerically computing the channel capacity and the rate-distortion problems. Willems~\cite{Willems1983} and Dupuis, Yu and Willems~\cite{1365218} presented iterative algorithms based on the B-A algorithm for computing the Gelfand-Pinsker and the Wyner-Ziv functions. We use principles from Willems' algorithms to develop an algorithm to numerically calculate the capacity for the cases we presented. More B-A based iterative algorithms for computing channel capacity and rate-distortion with side information can be found in~\cite{DBLP:journals/tit/ChengSX05} and in~\cite{5205855}. A B-A based algorithm for maximizing the directed-information can be found in~\cite{DBLP:journals/tit/NaissP13}.

Another approach for solving the Wyner-Ziv rate-distortion problem is the geometric programming approach. This approach was presented by Chiang and Boyd in their paper~\cite{Chiang04geometricprogramming}, in which they described methods, based on convex optimization and geometric programming, to calculate the channel capacity of the Gelfand-Pinsker channel and to calculate a lower bound on the rate-distortion of the Wyner-Ziv problem. Chiang and Boyd considered the Lagrange-dual of the Wyner-Ziv problem and they formulated a geometric program that constitutes a lower bound on the rate-distortion. However, their lower bound is not tight because they implicitly used the assumption that the derivative of the Lagrangian is zero for each value of the side information individually, while the original expression is only restricted to zero when averaging over the side information. During our present work, we found a tight lower bound on the rate-distortion of the Wyner-Ziv problem. The tight bound is obtained by considering a primal variable in the dual problem. A similar trick has been used recently by Naiss and Permuter~\cite{DBLP:journals/tit/NaissP13a} for transforming the rate-distortion with feed-forward problem into a geometric program.

\subsection{Organization of the paper and main contributions}
\label{sec:organ-paper-main}

To summarize, the main contributions of this paper are 1) we give single-letter characterizations of the capacity and the rate-distortion functions of new channel and source coding problems with increased partial side information, 2) we show a duality relationship between the channel capacity cases and the rate-distortion cases that we discuss, 3) we provide a tight lower bound on the Wyner-Ziv solution using convex optimization and geometric programming tools, 4) we provide a B-A based algorithm to solve the channel capacity problems we describe, 5) we show a duality between the Gelfand-Pinsker capacity converse and the Wyner-Ziv rate-distortion converse.

The reminder of this paper is organized as follows. In Section~\ref{sec:Setting} we introduce some notations for this paper and provide the settings of three channel coding and three source coding cases with increased partial side information. In Section~\ref{sec:main-results} we present the main results for coding with increased partial side information; we provide the capacity and the rate-distortion for the cases we introduced in Section~\ref{sec:Setting} and we point out the duality between the cases we examined. Section~\ref{sec:geom-progr} contains the main results for the geometric programming; we formulate a geometric program that is a tight lower bound on the Wyner-Ziv solution. Section~\ref{sec:examples} contains illuminating examples for the cases discussed in the paper. In Section~\ref{sec:algorithm} we describe the B-A based algorithm we used in order to solve the capacity examples. We conclude the paper in Section~\ref{sec:open_problems} and we highlight two open problems; channel capacity and rate-distortion with two-sided rate-limited partial side information. Appendix~\ref{sec:dual-conv-gelf} contains the duality derivation for the converse proofs of the Gelfand-Pinsker and the Wyner-Ziv problems and Appendices~\ref{sec:proof-CC} through~\ref{sec:proofs-algorithm} contain the proofs for our theorems and lemmas.

%
\begin{figure*}[t]
  \begin{minipage}[b]{0.48\linewidth}
    \centering
    \psfrag{N}{\bf{Noncausal:}}
    \psfrag{U}{\bf{Causal:}}
    \psfrag{P}{\bf Open Problem}
    \footnotesize
    \psfrag{l}{
      \begin{tabular}[c]{c}
        {Case 1}\\
        {$C_1(R')$}
      \end{tabular} }
    \psfrag{m}{
      \begin{tabular}[c]{c}
        {Case 2}\\
        {$C_2(R')$}
      \end{tabular} }
    \psfrag{p}{
      \begin{tabular}[c]{c}
        {Case 12}\\
        {$C_{12}(R'_1,R'_2)$}
        \end{tabular} }
    \psfrag{o}{
      \begin{tabular}[c]{c}
        {Case 2$_C$}\\
        {$C_{2,C}(R')$}
      \end{tabular} }
    \scriptsize
    \psfrag{A}{\ \ Encoder}
    \psfrag{B}{\ \ Channel}
    \psfrag{D}{\ \ Decoder}
    \psfrag{W}{$W$}
    \psfrag{w}{$\hat{W}$}
    \psfrag{X}{$X^n$}
    \psfrag{x}{$X_i$}
    \psfrag{Y}{$Y^n$}
    \psfrag{y}{$Y_i$}
    \psfrag{T}{$S_1^n$}
    \psfrag{t}{$S_{1,i}$}
    \psfrag{R}{$S_2^n$}
    \psfrag{r}{$S_{2,i}$}
    \psfrag{J}{$R'$}
    \psfrag{K}{$R'$}
    \psfrag{j}{$R'_1$}
    \psfrag{k}{$R'_2$}
    \includegraphics[scale = 0.4]{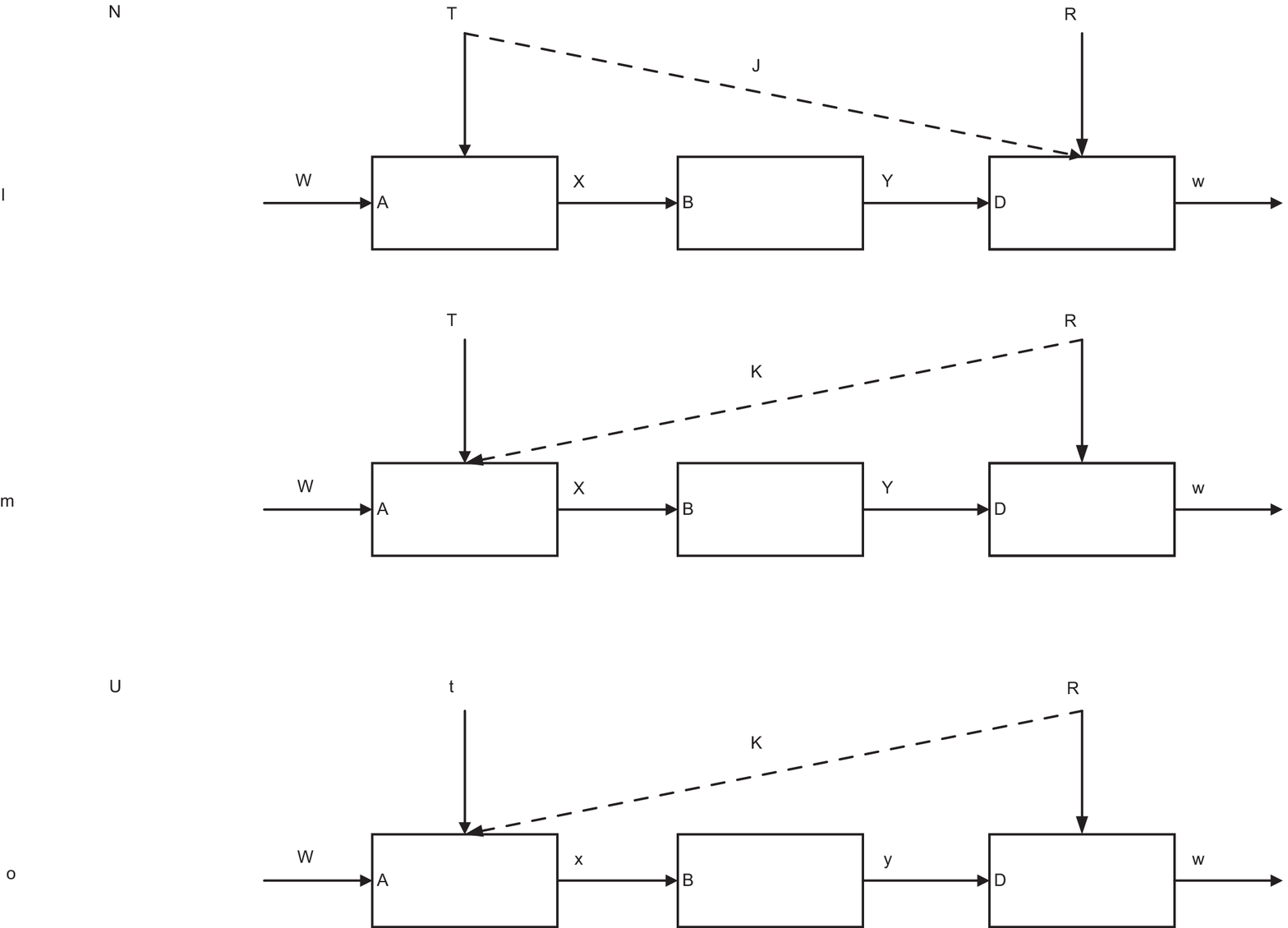}
    \caption{Channel coding with state information. Case 1: Rate-limited ESI at  the decoder. Case 2: Rate-limited DSI at the encoder. Case 2$_C$: Causal ESI and rate-limited DSI at the encoder.}
      \label{fig:Capacity}
  \end{minipage}
  \hspace{0.5cm}
  \begin{minipage}[b]{0.47\linewidth}
    \centering
    \psfrag{N}{\bf{Noncausal:}}
    \psfrag{U}{\bf{Causal:}}
    \psfrag{P}{\bf Open Problem}
    \footnotesize
    \psfrag{l}{
      \begin{tabular}[c]{c}
        {\ \ Case 2}\\
        {$R_2(R',D)$}
      \end{tabular} }
    \psfrag{m}{
      \begin{tabular}[c]{c}
        {\ \ Case 1}\\
        {$R_1(R',D)$}
      \end{tabular} }
    \psfrag{o}{
      \begin{tabular}[c]{c}
        {Case 1$_C$}\\
        {$R_{1,C} (R',D)$}
      \end{tabular} }
    \psfrag{p}{
      \begin{tabular}[c]{c}
        {Case 12}\\
        {$R_{12}(R'_1, R'_2)$}
      \end{tabular} }
    \scriptsize
    \psfrag{A}{\ \ Encoder}
    \psfrag{B}{\ \ Decoder}
    \psfrag{X}{$X^n$}
    \psfrag{x}{$\hat{X}^n$}
    \psfrag{s}{$\hat{X}_i$}
    \psfrag{E}{$S_1^n$}
    \psfrag{e}{$S_{1,i}$}
    \psfrag{R}{$S_2^n$}
    \psfrag{r}{$S_{2,i}$}
    \psfrag{J}{$R'$}
    \psfrag{K}{$R'$}
    \psfrag{j}{$R'_1$}
    \psfrag{k}{$R'_2$}
    \includegraphics[scale = 0.4]{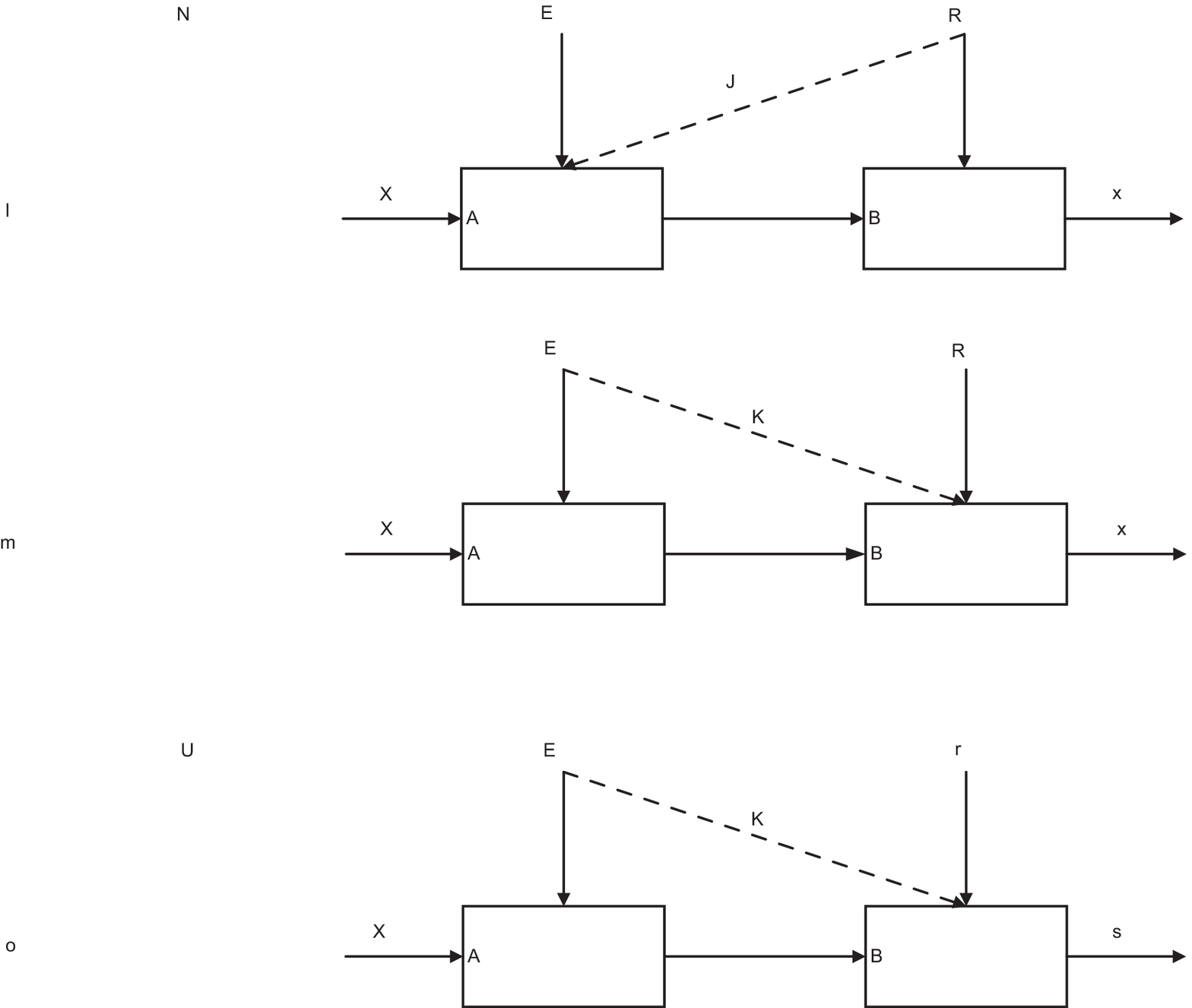}
    \caption{Source coding with side information.  Case 2: Rate-limited DSI at the encoder. Case 1: Rate-limited ESI at the decoder. Case 1$_C$: Causal DSI and rate-limited ESI at the decoder. The cases are presented in this order to allow each source coding case to be paralel to the dual channel coding case.
    \label{fig:Source_Coding}}
  \end{minipage}
\end{figure*}
%

\section{Problem Setting and Definitions} \label{sec:Setting}
In this section we describe and formally define three cases of channel coding problems and three cases of source coding problems. All six cases are presented in Figures~\ref{fig:Capacity} and~\ref{fig:Source_Coding}.

{\flushleft {\bf Notations.}} We use subscripts and superscripts to denote vectors in the following ways: $x^j = (x_1,\dots, x_j)$ and $x_i^j = (x_i,\dots, x_j)$ for $i \leq j$. Moreover, we use the lower case $x$ to denote sample value, the upper case $X$ to denote a random variable, the calligraphic letter $\cal X$ to denote the alphabet of $X$, $|\cal X|$ to denote the cardinality of the alphabet of $X$ and $p(x)$ to denote the probability $\Pr \{X=x\}$. We use the notation $\T(X)$ to denote the strongly typical set of the random variable $X$, as defined in \cite[Chapter 11]{cove_thom_91}.
\subsection{Definitions and problem formulation - channel coding with state information}
\label{sec:Def:CC}
\begin{definition}  
  A {\it discrete channel} is defined by the set $\{\cx, \cs_1, \cs_2, p(s_1,s_2), p(y|x,s_{1},s_2), \cy\}$. The channel's input sequence, $\{X_i \in \cx, i=1,2,\dots\}$, the ESI  sequence, $\{S_{1,i} \in \cs_1, i = 1,2,\dots\}$, the DSI sequence, $\{ S_{2,i} \in \cs_2, i=1,2,\dots\}$, and the channel's output  sequence, $\{Y_i \in \cy, i=1,2,\dots\}$, are discrete random variables drawn from the finite alphabets $\cx, \cs_1, \cs_2, \cy$, respectively. Denote the message and the message space as $W \in  \{1,2,\dots, 2^{nR}\}$ and let $\hat W$ be the reconstruction of the message $W$. The random variables $(S_{1,i},  S_{2,i})$ are i.i.d. $\sim p(s_1, s_2)$ and the channel is memoryless, i.e., at time $i$, the output, $Y_i$, has a conditional distribution of
  \begin{eqnarray}
    p(y_i|x^i, s_1^i, s_2^i, y^{i-1}) = p(y_i|x_i, s_{1,i}, s_{2,i}).
  \end{eqnarray}
\end{definition}

In the remainder of the paper, unless specifically mentioned otherwise, we refer to the ESI and the DSI as if they are known to the encoder and the decoder, respectively, in a noncausal manner. Also, as noted before, we use the term {\it increased} side information to indicate that the user's side information also includes a rate-limited description of the other user's partial side information. For example, when the decoder is informed with the DSI and with a rate-limited description of the ESI we would say that the decoder is informed with {\it increased} DSI.

{\flushleft {\bf Problem Formulation. }} For the channel $p(y|x,s_1,s_2)$, consider the following channel coding problem cases:
\begin{itemize}
\item{\bf Case 1}: The encoder is informed with ESI and the decoder is informed with increased DSI.
\item {\bf Case 2}: The encoder is informed with increased ESI and  the decoder is informed with DSI.
\item {\bf Case 2}$_C$: The encoder is informed with increased causal ESI ($S_1^i$ at time $i$) and the decoder is informed with DSI. This case is the same as Case 2, except for the causal ESI.
\end{itemize}
All cases are presented in Figure \ref{fig:Capacity}.

%
\begin{definition}
  A $(n,2^{nR}, 2^{nR'_{j}})$ {\it code}, $\{j \in 1,2\}$, for a channel with increased partial side information, as illustrated in Figure~\ref{fig:Capacity}, consists of two encoders and one decoder. The encoders are $f$ and $f_v$, where $f$ is the encoder for the channel's input and $f_v$ is the encoder for the side information, and the decoder is $g$, as described for each case:\\
  {{\it Case 1}: Two encoders
    \begin{align*}
      &f_v:\quad  \mathcal{S}_1^n \mapsto \{1,2,\dots, 2^{nR'_{1}}\},\\
      &f:\quad \{1,2,\dots, 2^{nR}\} \times \mathcal{S}_1^n \times \{1,2,\dots,2^{nR'_{1}}\} \mapsto \mathcal{X}^n,
    \end{align*}
    and a decoder
    \begin{align}
      &g:\quad \mathcal{Y}^n \times \mathcal{S}_2^n \times \{1,2,\dots, 2^{nR'_{1}}\} \mapsto \{1,2, \dots, 2^{nR}\}.
    \end{align}}
  {{\it Case 2}: Two encoders
    \begin{align*}
      &f_v:\quad \mathcal{S}_2^n  \mapsto \{1,2,\dots, 2^{nR'_{2}}\},\\
      &f:\quad \{1,2, \dots, 2^{nR}\} \times \mathcal{S}_1^n \times \{1,2,\dots, 2^{nR'_{2}}\}  \rightarrow \mathcal{X}^n,
    \end{align*}
    and a decoder
    \begin{align}
      &g:\quad  \mathcal{Y}^n \times \mathcal{S}_2^n \times  \{1,2,\dots, 2^{nR'_{2}}\} \mapsto \{1,2,\dots, 2^{nR}\}.
    \end{align}}
  {{\it Case 2$_C$}: Two encoders
    \begin{align*}
      &f_v:\quad \mathcal{S}_2^n     \mapsto \{1,2,\dots, 2^{nR'_{2}}\},\\
      &f_i: \quad \{1,2,\dots, 2^{nR}\}     \times \mathcal{S}_1^i \times \{1,2,\dots,2^{nR'_{2}}\} \mapsto \mathcal{X}_i,
    \end{align*}
    and a decoder
    \begin{align}
      &g:\quad  \mathcal{Y}^n \times \mathcal{S}_2^n \times \{1,2,\dots,2^{nR'_{2}}\} \mapsto \{1,2,\dots,2^{nR}\}.
    \end{align}}
  %
  %
  The {\it average probability of error}, $P_e^{(n)}$, for a $(2^{nR}, 2^{nR'_{j}}, n)$ code is defined as
  \begin{eqnarray} P_e^{(n)} = \frac{1}{2^{nR}} \sum_{w=1}^{2^{nR}} \Pr\Big\{
    \hat{W} \neq W\ | W=w\Big\},
  \end{eqnarray}
  where the index $W$ is chosen according to a uniform distribution over the set $\{1,2,\dots,2^{nR}\}$.
  %
  A rate pair $(R, R')$ is said to be {\it achievable} if   there exists a sequence of $(2^{nR}, 2^{nR'}, n)$ codes such that the average probability of error $P_e^{(n)} \rightarrow 0$ as $n \rightarrow \infty$.
\end{definition}
\begin{definition}
  The {\it capacity} of the channel, $C(R')$, is the supremum of all $R$ such that the rate pair $(R,R')$ is achievable.
\end{definition}
\subsection{Definitions and problem formulation - source coding with side information}
\label{sec:def:SC}
Throughout this article we use the common definitions of rate-distortion as presented in \cite{cove_thom_91}.
\begin{definition} \label{def:SC:problem_setting}
  The source sequence $\{X_i \in {\cal X}, i= 1,2,\dots\}$, the ESI sequence $\{S_{1,i} \in {\cal S}_1, i= 1,2,\dots\}$ and the DSI sequence $\{S_{2,i} \in  {\cal S}_2, i= 1,2,\dots\}$ are discrete random variables drawn from the finite alphabets ${\cal X}, {\cal S}_1$ and ${\cal S}_2$ respectively. The random variables $(X_i, S_{1,i}, S_{2,i})$ are i.i.d $\sim p(x,s_1,s_2)$. Let $\hat{{\cal X}}$ be the reconstruction alphabet and $d_x: {\cal X} \times \hat{{\cal X}}  \mapsto [0,\infty)$ be the distortion measure. The distortion between sequences is defined in the usual way:
  \begin{eqnarray}
    d(x^n,\hat{x}^n) = \frac{1}{n} \sum_{i=1}^n d(x_i, \hat{x}_i).
  \end{eqnarray}
\end{definition}
{\flushleft  {\bf Problem Formulation. }} For the source, $X$, the ESI, $S_1$, and the DSI, $S_2$, consider the following source coding problem cases:
\begin{itemize}
\item {\bf Case 1}: The encoder is informed with ESI and the decoder is informed with increased DSI.
\item {\bf Case 2}: The encoder is informed with increased ESI and the decoder is informed with DSI.
\item {\bf Case 1}$_C$: The encoder is informed with ESI and the decoder is informed with increased causal DSI ($S_2^i$ at time $i$). This case is the same as Case 1, except for the causal DSI.
\end{itemize}
All cases are presented in Figure \ref{fig:Source_Coding}.
%
\begin{definition}
  A $(n, 2^{nR}, 2^{nR'_{j}}, D)$ {\it code}, $\{j \in 1,2\}$, for the source $X$ with increased partial side information, as illustrated in Figure \ref{fig:Source_Coding}, consists of two encoders, one decoder and a distortion constraint. The encoders are $f$ and $f_v$, where $f$ is the encoder for the source and $f_v$ is the encoder for the side information, and the decoder is $g$, as described for each case:\\*
  {{\it Case 1}: Two encoders
    \begin{align*}
      &f_v:\quad \mathcal{S}_1^n \mapsto \{1,2,\dots,2^{nR'_{1}}\},\\
      &f:\quad \mathcal{X}^n \times \mathcal{S}_1^n \times \{1,2,\dots,2^{nR'_{1}}\} \mapsto \{1,2,\dots,2^{nR}\},
    \end{align*}
    and a decoder
    \begin{align}
      &g:\quad \{1,2,\dots,2^{nR}\} \times \mathcal{S}_2^n \times \{1,2,\dots,2^{nR'_{1}}\} \mapsto \hat{\mathcal{X}}^n.
    \end{align}}
  {{\it Case 2}: Two encoders
    \begin{align*}
      &f_v:\quad \mathcal{S}_2^n  \mapsto \{1,2,\dots,2^{nR'_{2}}\},\\
      &f:\quad \mathcal{X}^n \times \mathcal{S}_1^n \times \{1,2,\dots,2^{nR'_{2}}\}  \mapsto \{1,2,\dots,2^{nR}\},
    \end{align*}
    and a decoder
    \begin{align}
      &g:\quad  \{1,2,\dots, 2^{nR}\} \times \mathcal{S}_2^n \times  \{1,2,\dots,2^{nR'_{2}}\} \mapsto \hat{\mathcal{X}}^n.
    \end{align}
  {\it Case 1$_C$}: Two encoders
    \begin{align*}
      &f_v:\quad \mathcal{S}_1^n     \mapsto \{1,2,\dots,2^{nR'_{1}}\},\\
      &f:\quad \mathcal{X}^n     \times \mathcal{S}_1^n \times \{1,2,\dots,2^{nR'_{1}}\} \mapsto \{1,2,\dots,2^{nR}\},
    \end{align*}
    and a decoder
    \begin{align}
      &g_i:\quad     \{1,2,\dots,2^{nR}\} \times \mathcal{S}_{2}^i \times \{1,2,\dots,2^{nR'_{1}}\} \mapsto \hat{\mathcal{X}}_i.
    \end{align}}
  The distortion constraint for all three cases is:
  \begin{eqnarray}
    \mathbb{E}\Big[\frac{1}{n} \sum_{i=1}^nd(X_i,\hat{X}_i)\Big] \leq D.
  \end{eqnarray}
  For a given distortion, $D$, and for any $\epsilon > 0$, the rate pair $(R,R')$ is said to be {\it achievable} if there exists a  $(n,2^{nR},2^{nR'},D+\epsilon)$ code for the rate-distortion problem.
\end{definition}
%
\begin{definition}
  For a given $R'$ and distortion $D$, the {\it operational rate} $R^*(R',D)$ is the infimum of all $R$, such that the rate pair $(R,R')$ is achievable.
\end{definition}
\section{Coding with Increased Partial Side Information - Main Results}
\label{sec:main-results}
In this section we present the main results of this paper. We will first present the results for the channel coding cases, then the main results for the source coding cases and, finally, we will present the duality between them.
\subsection{Channel coding with side information}
\label{sec:CC:main-results}
For a channel with two-sided state information as presented in Figure \ref{fig:Capacity}, where $(S_{1,i}, S_{2,i}) \sim p(s_1,s_2)$, the {\it capacity} is as follows
\begin{theorem}[The capacity for the cases in Figure \ref{fig:Capacity}]\label{theorem:CC} For the memoryless channel $p(y|x,s_1,s_2)$, where $S_1$ is the ESI and $S_2$ is the DSI and the side information $(S_{1,i}, S_{2,i}) \sim p(s_1,s_2)$,  the channel capacity is\\*
\\*
  {\it Case\ 1:} The encoder is informed with ESI and the decoder is informed with increased DSI,
  {\small
    \begin{align}
      C_1^* =& \max_{\substack{
          p(v_1|s_1) p(u|s_1,v_1) p(x|u,s_1,v_1)\\
          {\rm s.t.}\ \ R' \geq I(\tS_1;S_1)-I(V_1;Y,S_2)}}
      I(U;Y,S_2|\tS_1) -I(U;S_1|\tS_1).
    \end{align}}
  {\it Case\ 2:} The encoder is informed with increased ESI and the decoder is informed with DSI;
  \\
  Lower bounded by
  {\small
    \begin{align} \label{eq:CC_case2}
      C_2^{lb*}& = \max_{
        \substack{
          p(v_2|s_2) p(u|s_1,v_2) p(x|u,s_1,v_2)\\
          {\rm s.t.}\ \ R' \geq I(V_2;S_2|S_1)
        }}
      I(U;Y,S_2|V_2)-I(U;S_1|V_2).
    \end{align}}
  Upper bounded by
  {\small
    \begin{align}
      C_2^{ub1*}& = \max_{
        \substack{
          p(v_2|s_1,s_2) p(u|s_1,v_2) p(x|u,s_1,v_2)\\
          {\rm s.t.}\ \ R' \geq I(V_2;S_2|S_1)
        }}
      I(U;Y,S_2|V_2)-I(U;S_1|V_2)
    \end{align}}
  and by
  {\small
    \begin{align}
      C_2^{ub2*}& = \max_{
        \substack{
          p(v_2|s_2) p(u|s_1,s_2,v_2) p(x|u,s_1,v_2)\\
          {\rm s.t.}\ \ R' \geq I(V_2;S_2|S_1)
        }}
      I(U;Y,S_2|V_2)-I(U;S_1|V_2).
    \end{align}}
  {\it Case\ 2$_C$:} The encoder is informed with increased causal ESI ($S_1^i$ at time $i$) and the decoder is informed with DSI,
  {\small
    \begin{align}
      C_{2C}^*& = \max_{
        \substack{
          p(v_2|s_2) p(u|v_2) p(x|u,s_1,v_2)\\
          R' \geq I(V_2; S_2)
        }}
      I(U;Y,S_2|V_2).
    \end{align}}
  For case $j$,  $j\in\{1,2\}$, some joint distribution, $p(s_1,s_2,v_j,u,x,y)$, and $(U,V_j)$ being some auxiliary random variables with bounded cardinality.\\*

Section \ref{sec:proof-CC} contains the proof.
\end{theorem}
\begin{lemma}\label{lemma:CC} For all three channel coding cases described in this section and for $j\in\{1,2\}$, the following statements hold
  \begin{itemize}
  \item[$(i)$] The function $C_j(R')$ is a concave function of $R'$.
  \item[$(ii)$] It is enough to take $X$ to be a deterministic function of $(U,S_1,V_j)$ to evaluate $C_j$.
  \item[$(iii)$] The auxiliary alphabets $\mathcal{U}$ and $\mathcal{V}_j$ satisfy
    \begin{align*}
      \begin{array}[c]{l c}
        \mbox{for Case 1:} & |\mathcal{V}_1| \leq |\mathcal{X}| |\mathcal{S}_1| |\mathcal{S}_2| + 1 \mbox{\quad and } \\
        & |\mathcal{U}| \leq |\mathcal{X}| |\mathcal{S}_1| |\mathcal{S}_2| \big(|\mathcal{X}| |\mathcal{S}_1| |\mathcal{S}_2| +1\big),\\
        \hline
        \mbox{for Case 2:} & |\mathcal{V}_2| \leq  |\mathcal{S}_1| |\mathcal{S}_2| + 1 \mbox{\quad and} \\
        & |\mathcal{U}| \leq |\mathcal{X}| |\mathcal{S}_1| |\mathcal{S}_2| \big(|\mathcal{S}_1| |\mathcal{S}_2| +1\big),\\
        \hline
        \mbox{for Case 2$_C$:} & |\mathcal{V}_2| \leq |\mathcal{S}_2| + 1 \mbox{\quad and}\\
        & |\mathcal{U}| \leq |\mathcal{X}| |\mathcal{S}_2| \big(|\mathcal{S}_2| +1\big).
      \end{array}
    \end{align*}
  \end{itemize}
  Appendix~\ref{sec:proof-lemma-1} contains the proof for the above lemma.

{\it Remark:}
We assume that the lower bound of Case 2 is tight, namely, $C_2 = C_2^{lb}$. This claim is hard to corroborate; we have not, as yet, derived a converse proof that maintains both Markov relations $V_2 - S_2 - S_2$ and $U - (S_1,V_2) - S_2$ and that bounds any achievable rate from above simultaneously.
\end{lemma}
\subsection{Source coding with side information}
\label{sec:SC:Main_results}
For the problem of source coding with side information as presented in Figure \ref{fig:Source_Coding}, the {\it rate-distortion} function is as follows:
\begin{theorem}[The rate-distortion function for the cases in Figure \ref{fig:Source_Coding}]\label{theorem:SC} For a bounded distortion measure $d(x,\hat{x})$, a source, $X$, and side information, $S_1,S_2$, where $(X_i, S_{1,i},S_{2,i}) \sim p(x,s_1,s_2)$, the rate-distortion function is\\*
  \\*
  {\it Case 1:} The encoder is informed with ESI and the decoder is informed with increased DSI,
  {\small
    \begin{align}
      R_1^*(D)& = \min_{
      \substack{
          p(v_1|s_1) p(u|x, s_1, v_1) p(\hat{x}|u,s_2, v_1)\\
          {\rm s.t.}\ \ R' \geq I(V_1 ; S_1 | S_2)
        }}
      I(U;X,S_1|V_1) - I(U;S_2|V_1).
    \end{align}}
  %
  {\it Case 1$_C$:} The encoder is informed with ESI and the decoder is informed with increased causal DSI ($S_2^i$ at time $i$),
  {\small
    \begin{align}
      R_{1C}^*(D)& = \min_{
        \substack{
          p(v_1|s_1) p(u|x,s_1,v_1) p(\hat{x}|u,s_2, v_1)\\
          {\rm s.t.} \ \ R' \geq I(V_1;S_1)
        }}
      I(U;X,S_1|V_1).
    \end{align}}
  {\it Case 2:} The encoder is informed with increased ESI and the decoder is informed with DSI,
  {\small
    \begin{align}
      R_2^*(D) &= \min_{
      \substack{
          p(v_2|s_2) p(u|x, s_1, v_2) p(\hat{x}|u,s_2, v_2)\\
          {\rm s.t.}\ \ \ R' \geq I(V_2 ; S_2 ) -I(V_2; X, S_1)
        }}
      I(U ; X, S_1| V_2) - I(U ; S_2|V_2).
    \end{align}}
  For case $j$, $j\in \{1,2\}$, some joint distribution, $p(x,s_1,s_2,v_j,u,\hat{x})$, where $\mathbb{E}\Big[\frac{1}{n}\sum_{i=1}^nd(X_i,\hat{X}_i)\Big] \leq D$ and $(U,V_j)$ being some auxiliary random variables with bounded cardinality.

Section \ref{proof:CS} contains the proof.
\end{theorem}
\begin{lemma} \label{lemma:SC}
  For all cases of rate-distortion problems in this section and for $j \in \{1,2\}$, the following statements hold.
  \begin{itemize}
  \item[(i)] The function $R_j(R',D)$ is a convex function of $R'$ and $D$.
  \item[(ii)] It is enough to take $\hat{X}$ to be a deterministic function of $(U,S_2,V_j)$ to evaluate $R_j$.
  \item[(iii)]  The auxiliary alphabets $\mathcal{U}$ and $\mathcal{V}_j$ satisfy
    \begin{align*}
      \begin{array}[c]{l c}
        \mbox{for Case 1:} & |\mathcal{V}_1| \leq |\mathcal{S}_1| |\mathcal{S}_2| + 1 \mbox{\quad and } \\
        & |\mathcal{U}| \leq |\mathcal{X}| |\mathcal{S}_1| |\mathcal{S}_2| \big(|\mathcal{S}_1| |\mathcal{S}_2| +1\big),\\
        \hline
        \mbox{for Case 1$_C$:} & |\mathcal{V}_1| \leq  |\mathcal{S}_1| + 1 \mbox{\quad and} \\
        & |\mathcal{U}| \leq |\mathcal{X}| |\mathcal{S}_1| \big(|\mathcal{S}_1| + 1\big),\\
        \hline
        \mbox{for Case 2:} & |\mathcal{V}_2| \leq |\mathcal{X}| |\mathcal{S}_1| |\mathcal{S}_2| + 1 \mbox{\quad and}\\
        & |\mathcal{U}| \leq |\mathcal{X}| |\mathcal{S}_1| |\mathcal{S}_2| \big(|\mathcal{X}| |\mathcal{S}_1| |\mathcal{S}_2| +1\big).
      \end{array}
    \end{align*}
  \end{itemize}
  Appendix~\ref{sec:proof-lemma-1} contains the proof for the above lemma.

\end{lemma}

\subsection{Main results - duality}
\label{sec:main-results-duality}
We now investigate the duality between the channel coding and the source coding for the cases in Figures \ref{fig:Capacity} and \ref{fig:Source_Coding}. The following transformation makes the duality between the  channel coding cases 1, 2, 2$_C$ and the source coding cases 2, 1, 1$_C$, respectively, evident. The left column corresponds to channel coding and the right column to source coding.
For cases $j$ and $\bar{j}$, where $j,\bar{j} \in~\{1,2\}$ and $\bar{j}~\neq~j$, consider the transformation:
\begin{align}
  {\rm channel\ coding} &\longleftrightarrow {\rm source \ coding}\\
  C &\longleftrightarrow R(D)\\
  {\rm maximization} &\longleftrightarrow {\rm minimization}\\
  C_j &\longleftrightarrow R_{\bar{j}}(D)\\
  X & \longleftrightarrow \hat{X}\\
  Y & \longleftrightarrow X\\
  S_j & \longleftrightarrow S_{\bar{j}}\\
  V_j & \longleftrightarrow V_{\bar{j}}\\
  U & \longleftrightarrow U\\
  R' &  \longleftrightarrow R'.
\end{align}
This transformation is an extension of the transformation provided in~\cite{Cover:2006:DCC:2263239.2266915} and in~\cite{DBLP:journals/tit/PradhanCR03}. Note that while the channel capacity formula in Case $j$ and the rate-distortion function in Case $\bar{j}$  are dual to one another in the sense of maximization-minimization, the corresponding rates $R'$ are not dual to each other in this sense; i.e., one would expect to see an opposite inequality ($\geq\ \leftrightarrow \ \leq$) for dual cases, where we have an inequality that is in the same direction ($\leq\ \leftrightarrow \ \leq$) in the $R'$ formulas. The duality in the side information rates, $R'$,  is then in the sense that the arguments in the formulas for the dual $R'$ are dual. This exception is due to the fact that while the Gelfand-Pinsker and the Wyner-Ziv problems for the main channel or the main rate-distortion problems are dual, the Wyner-Ziv problem for the side information stays the same; the only difference is the input and the output.

\section{Geometric Programming}
\label{sec:geom-progr}

In this section, we provide a method to evaluate the Wyner-Ziv rate, using the Lagrange dual function and geometric programming. Before presenting the main results on this subject, let us provide the definitions and notations that we will use throughout this section and throughout the proof of the forthcoming main results.

\subsection{Definitions and preliminaries - convex optimization and Lagrange duality}
\label{sec:defen-conv-optim}
Most of the notations and the definitions that we use in this section are taken from~\cite{Boyd:2004:CO:993483}. We denote the variable $x$ with dimension greater than $1$ as $\bx$ and we use $\bx \succeq 0$ to denote that $x_i \geq 0$ for all $i = 1,2,\dots, \dim(\bx)$.

Consider the following  optimization problem:
\begin{align}
  \label{eq:51}
  \begin{array}[l]{l l}
  \mbox{minimize} & f_0({\bf x})\\
  \mbox{subject to} & f_i({\bf x}) \leq 0,\quad i = 1,2,\dots, m,\\
  & h_j({\bf x}) = 0, \quad j = 1,2,\dots, p,
\end{array}
\end{align}
with the variable ${\bf x} \in \mathbb{R}^n$.
We refer to $f_0$ as the {\it objective function} of the optimization problem and to $f_i$ and $h_j$ as the {\it constraint functions}. We let $\mathcal{D}$ denote the domain of ${\bf x}$; this is the set of all points for which the objective and the constraint functions are defined. We denote the optimal minimizer of $f_0(\bx)$ in $\mathcal{D}$ as $\bx^*$. If the objective function, $f_0({\bf x})$, and the inequality constraint functions, $f_i({\bf x}),\ i=1,2,\dots,m$, are all convex in ${\bf x}$ and the equality constraint functions, $h_j({\bf x}),\ j=1,2,\dots,p$, are affine in ${\bf x}$, then the problem is said to be a {\it convex optimization problem}. The {\it Lagrangian} associated with problem~(\ref{eq:51}) is
\begin{align}
  \label{eq:52}
  L(\bx,\blambda, \bmu) = f_0({\bf x}) + \sum_{i=1}^m \lambda_i f_i(\bx) + \sum_{j=1}^p \mu_j h_j(\bx),
\end{align}
where ${\bf x} \in \mathcal{D},\ \blambda \in \mathbb{R}^m$ and $\bmu \in \mathbb{R}^p$. The {\it Lagrange dual function}, as defined in~\cite[Capter 5.1.2]{Boyd:2004:CO:993483}, is
  \begin{align}
    \label{eq:53}
    g(\blambda, \bmu) = \inf_{\bx \in \mathcal{D}} L(\bx,\blambda, \bmu).
  \end{align}
Following from \cite[Chapter 5.1.3]{Boyd:2004:CO:993483}, for any $\blambda$  where $\lambda_i \geq 0$ for $i=1,2,\dots,m$, the Lagrange dual function yields a lower bound on the optimal value, $f_0(\bx^*)$.
  The {\it Lagrange dual problem}~\cite[Chapter 5.2]{Boyd:2004:CO:993483} associated with (\ref{eq:51}) is
  \begin{align}
    \label{eq:54}
    \begin{array}[l]{l l}
      \mbox{maximize} & g(\blambda,\bmu)\\
      \mbox{subject to} & \lambda_i \geq 0,\quad i = 1,2,\dots, m.
    \end{array}
  \end{align}
In this context, we refer to the original problem (\ref{eq:51}) as the {\it primal problem}.
The {\it strong duality} property is associated with the case where the solution for the dual problem and the solution for the primal problem coincide. Following from~\cite[Chapter 5.2.3]{Boyd:2004:CO:993483}, if the primal problem is convex and Slater's condition~\cite[Chapter 5.2.3]{Boyd:2004:CO:993483} holds, then strong duality holds.

A special family of optimization problems that we are interested in is the family of geometric programs. This type of optimization problems is defined in~\cite[Chapter 4.5]{Boyd:2004:CO:993483} and is summarized here. Define {\it monomial} as the function
\begin{align}
  \label{eq:61}
  f(\bx) = c x_1^{a_1} x_2^{a_2} \dots x_n^{a_n},
\end{align}
were $c > 0$ and $a_i \in \mathbb{R}$. A sum of monomials, i.e., a function of the form
\begin{align}
  \label{eq:62}
  f(\bx) = \sum_{k=1}^K c_k x_1^{a_{1k}} x_2^{a_{2k}} \dots x_n^{a_{nk}},
\end{align}
where $c_k > 0$, is called a {\it posynomial}. An optimization problem of the form
\begin{align}
  \label{eq:63}
  \begin{array}[l]{l l}
  \mbox{minimize} & f_0({\bf x})\\
  \mbox{subject to} & f_i({\bf x}) \leq 1,\quad i = 1,2,\dots, m,\\
  & h_j({\bf x}) = 1, \quad j = 1,2,\dots, p,
\end{array}
\end{align}
where $f_0, \dots, f_{m}$ are posynomials, $h_1,\dots,h_p$ are monomials and $\bx \succeq 0$ is called a {\it geometric program}. Geometric programs, as mentioned in~\cite[Chapter 4.5]{Boyd:2004:CO:993483}, are not convex problems. However, these problems can be transformed into convex optimization problems by taking $\log(\cdot)$ on both the objective and the constraint functions.

\subsection{Problem Setting and Main Results}
\label{sec:main-results-GP}

Let us consider the classic Wyner-Ziv problem as illustrated in Figure~\ref{fig:WZ}. Assume correlated random variables $(X,S) \sim \iid p(x,s)$ with finite alphabets $\mathcal{X},\mathcal{S}$, respectively. Let $\big\{(X_i, S_{i})\big\}_{i=1}^n$ be a sequence of $n$ independent drawings of $(X,S)$. Let the sequence $X^n$ be the source sequence and let $S^n$ be the side information sequence available at the decoder. We wish to describe the source, $X$, at rate $R$ bits per symbol and to reconstruct $\hat X$ at the decoder with a distortion smaller than or equal to $D$, i.e., when encoding $X$ in blocks of length $n$, we desire that $ \mathbb{E}\Big[\frac{1}{n} \sum_{i=1}^nd(X_i,\hat{X}_i)\Big] \leq D$.

\begin{figure}[h!]
  \centering \small
  \psfrag{A}{\ \ Encoder}
  \psfrag{B}{\ \ Decoder}
  \psfrag{X}{$X$}
  \psfrag{x}{$\hat X$}
  \psfrag{R}{$S$}
  \includegraphics[scale = 0.5]{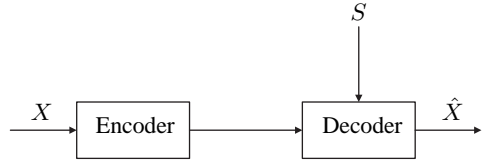}
  \caption{{\it The Wyner-Ziv problem}.}
  \label{fig:WZ}
\end{figure}

The rate-distortion function with side information at the decoder \cite{1055508} is
\begin{align}
  \label{eq:17}
  R(D) =  \min_{p(u|x) p(\hat x|u,s)}  I(U;X|S)
\end{align}
for some joint distribution $p(x,s,u,\hat x)$ such that $ \mathbb{E}\Big[d(X,\hat{X})\Big] \leq D$, i.e., $\sum_{x,s,u,\hat x} p(x,s) p(u|x) p(\hat x|u,s) d(x, \hat x) \leq D$. According to~\cite{Willems1983}, we can write the expression of the rate-distortion function as
\begin{align}
  \label{eq:18}
  R(D) = \min_{q(t|x)} I(T;X|S)
\end{align}
for some joint distribution $p(x,s,t) = p(x,s) q(t|x)$, where $\mathcal{T}$ is the set of all mappings
\begin{align}
  \label{eq:19}
  t:\quad \mathcal{S} \mapsto \hat{\mathcal{X}},
\end{align}
and the distortion constraint
\begin{align}
  \label{eq:20}
  \sum_{x,s,t} p(x,s) q(t|x) d\big(x, t(s)\big) \leq D
\end{align}
is maintained. We denote the set of $q(t|x)$'s for all $x \in \mathcal{X}$ and $t \in \mathcal{T}$ as  $\bq \in \mathbb{R}^{|\mathcal{T}||\mathcal{X}|}$ and we note that $I(T;X|S)$ is a convex function of $\bq$ and that the rate-distortion function, $R(D)$, is its optimal value.

Combining (\ref{eq:18}) and (\ref{eq:20}), we get that the Wyner-Ziv problem is the following problem
\begin{align}
  \label{eq:21}
  \begin{array}[l]{l l}
  \mbox{minimize} & \sum_{x,s,t} p(x,s) q(t|x) \log \frac{q(t|x)}{Q(t|s)}\\
  \mbox{subject to} & \sum_{t} q(t|x) = 1 \quad \forall x,\\
  & \sum_{x,s,t} p(x,s) q(t|x) d\big(x, t(s)\big) \leq D,\\
  \mbox{ } & q(t|x) \geq 0 \quad \forall x,t,
\end{array}
\end{align}
where the variables of the optimization are $\bq$ and the constant parameters are the source distribution, $p(x,s)$, the distortion measure, $d\big(x,t(s)\big)$, and the distortion constraint, $D$, for all $x \in \mathcal{X}$, $s \in \mathcal{S}$ and $t\in \mathcal{T}$. The marginal distribution $Q(t|s)$ is defined by
\begin{align}
 Q(t|s) = \frac{\sum_x p(x,s) q(t|x)}{\sum_x p(x,s)},
\end{align}
 We define the set of $Q(t|s)$'s for all $s \in \mathcal{S}$ and $t \in \mathcal{T}$ as $\bQ \in \mathbb{R}^{|\mathcal{T}||\mathcal{S}|}$.

The main result of this section is brought in the following theorem.
\begin{theorem}
\label{theoremGP}
  The Lagrange dual of the Wyner-Ziv rate-distortion problem is the following geometric program (in convex form):
  \begin{align}
    \label{eq:55}
    \begin{array}[l]{l l}
    \mbox{maximize} &  \sum_x p(x) \alpha_x - \gamma D\\
    \mbox{subject to} &  \alpha_x + \sum_s p(s|x) \bigg[ \log p(x|s) - \gamma d\big( x,t(s) \big) - y_{x,s,t}\bigg] \leq 0\quad \forall x,t,\\
    & \log \left( \sum_x \exp\big\{y_{x,s,t}\big\} \right) \leq 0 \quad \forall s,t,\\
    & \gamma \geq 0,
  \end{array}
\end{align}
where the optimization variables are ${\bm \alpha} \in \mathbb{R}^{|\mathcal{X}|}, \gamma \in \mathbb{R}_+$ and ${\bf y} \in \mathbb{R}^{|\mathcal{X}| |\mathcal{S}| |\mathcal{T}|}$, and the constant parameters are the source distribution $p(x,s)$, the distortion measure $d\big(x,t(s)\big)$ and the distortion constraint, $D$. Furthermore, if Slater's condition~\cite[Chapter 5.2.3]{Boyd:2004:CO:993483} holds, then strong duality holds and the solution for the optimization problem in~(\ref{eq:55}) is a tight lower bound on the Wyner-Ziv solution, (\ref{eq:21}), and $R(D)$ is its optimal value.
\end{theorem}
\proof  The proof for Theorem~\ref{theoremGP}  is given in Appendix~\ref{sec:proof-GP}.

\section{Examples}
\label{sec:examples}

In this section we provide examples for Case~2 of the channel coding theorem and for Case~1 of the source coding theorem. The numerical iterative algorithm, which we used to numerically calculate the lower bound, $C_2^{lb}$, is provided in the next section.
\begin{example}[{\it Case 2 channel coding for a binary channel}]\label{ex:CC}
  Consider the binary channel illustrated in Figure \ref{fig:ch_top}. The alphabet of the input, the output and the two states is binary $\mathcal{X} = \mathcal{Y} = \mathcal{S}_1 = \mathcal{S}_2 = \{0,1\}$ with  $(S_1,S_2) \sim {\bf P_{S_1 S_2}}$ being a joint PMF matrix. The channel is dependent on the states $S_1$ and $S_2$, where the encoder is fully informed with $S_1$ and with $S_2$ with a rate limited to $ R'$ and the decoder is fully informed with $S_2$. The dependence of the channel on the states is illustrated in Figure \ref{fig:ch_top}. If  $(S_1 = 1, S_2 = 0)$ then the channel is the {\it Z channel} with transition probability $\epsilon$, if $(S_1 = 1, S_2 = 1)$ then the channel has no error, if $(S_1=0,S_2=0)$ then the channel is the {\it X-channel} and if $(S_1=0,S_2=1)$ then the channel is the {\it S-channel} with transition probability of $\epsilon$. The side information's joint pmf is
\begin{align*}
  {\bf P_{S_1S_2}} =  \begin{pmatrix} 0.1 & 0.4 \\ 0.4 & 0.1     \end{pmatrix}.
\end{align*}
The expressions for the lower bound on the capacity $C^{lb}_2( R')$ and for $ R'$ are brought in Case~2 of Theorem~\ref{theorem:CC}.
  \begin{figure}[h!] \vspace{-1cm}
    \centering \small
    \psfrag{e}{Encoder}
    \psfrag{c}{Channel}
    \psfrag{d}{Decoder} \footnotesize
    \psfrag{M}{$M$}
    \psfrag{x}{$X^n$}
    \psfrag{y}{$Y^n$}
    \psfrag{m}{$\hat M$}
    \psfrag{s}{$S_1^n$}
    \psfrag{v}{$S_2^n$}
    \psfrag{R}{$ R'$}
    \psfrag{z}{$0$}
    \psfrag{o}{$1$}
    \psfrag{A}{ $(1,0)$}
    \psfrag{B}{ $(1,1)$}
    \psfrag{C}{ $(0,0)$}
    \psfrag{D}{ $(0,1)$}
    \psfrag{E}{$(S_1, S_2)$} \small
    \psfrag{F}{The}
    \psfrag{G}{Channel}
    \includegraphics[scale = 0.6]{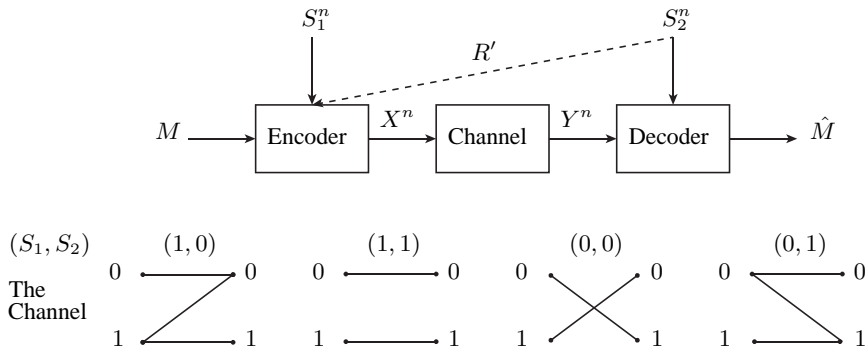}
    \caption{Example 1 {\it Channel coding Case 2} - channel topology.} \label{fig:ch_top}
  \end{figure}

In Figure~\ref{fig:example1} we provide the graph from of the computation of the lower bound on the capacity for the binary channel we are testing. In the graph, we present the lower bound, $C^{lb}_2( R')$, as a function of $ R'$. We also provide the Cover \& Chiang~\cite{Cover:2006:DCC:2263239.2266915} capacity (where $ R' = 0$) and the Gelfand \& Pinsker~\cite{citeulike:437050} capacity (where $ R' = 0$ and the decoder is not informed with $S_2$).\\

\begin{figure}[h!]
  \centering \small 
  \psfrag{a}{$C^{lb}_2( R')$}
  \psfrag{b}{C-C rate}
  \psfrag{c}{G-P rate}
  \psfrag{r}{$ R'$ [{\it bits}]}
  \psfrag{d}{Rate [{\it bits}]}
  \psfrag{h}{$H(S_2|S_1)$}
  \includegraphics[scale=0.8]{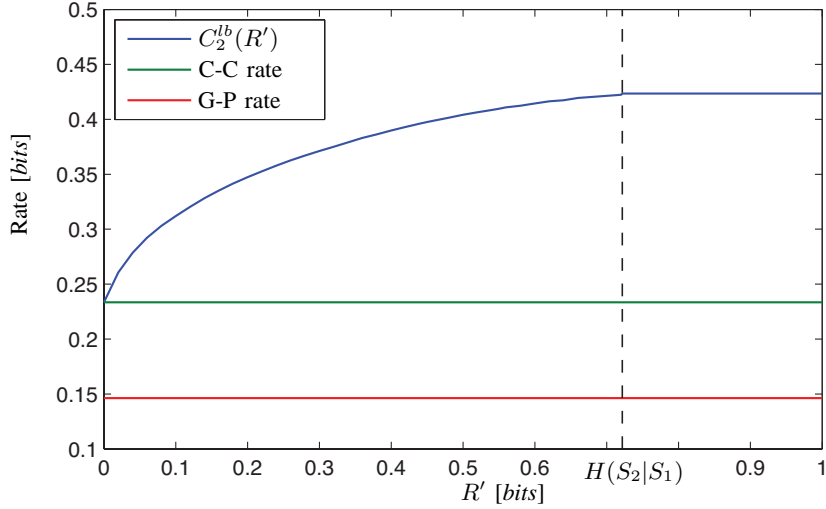}
  \caption[]{Example 1. {\it Channel coding Case 2}  for the channel depicted in Figure~\ref{fig:ch_top}, where the side information is distributed $S_1 \sim \mbox{Bernoulli}(0.5)$, and $\Pr \{S_2 \neq S_1\} = 0.8$. ${\it C^{lb}_2(R')}$ is the lower bound on the capacity of this channel, {\it C-C rate} is the Cover-Chiang rate ($R'=0$) and {\it G-P rate} is the Gelfand-Pinsker rate ($R'=0$ and the decoder has no side information available at all). Notice that at the encoder the maximal uncertainty about $S_2$ is $H(S_2|S_1) = 0.7219$ {\it bit}. Therefore, for any $ R' \geq 0.7219$ $C_2^{lb}$ reaches its maximal value.}
  \label{fig:example1}
\end{figure}

{\it Discussion}:
\begin{enumerate}
\item
  The algorithm that we used to calculate $C_2^{lb}(R')$ and $R'$ combines a grid-search and a Blahut-Arimoto-like algorithms. We first construct a grid of probabilities of the random variable $V_2$ given $S_2$, namely, $w(v_2|s_2)$. Then, for every probability $w(v_2|s_2)$ such that $I(V_2;S_2|S_1)$ is close enough to $R'$ we calculate the maximum of $I(U;Y,S_2|V_2) - I(U;S_1|V_2)$ using the iterative algorithm described in the next section. We then choose the maximum over those maximums and declare it to be $C_2^{lb}$. By taking a fine grid of the probabilities $w(v_2|s_2)$ the operation's result can be arbitrarily close to $C_2^{lb}$.
\item
  For a given joint PMF matrix ${\bf P_{S_1S_2}}$, we can see that $C^{lb}_2( R')$ is non-decreasing in $ R'$. Furthermore, since the expression $I(V_2;S_2|S_1)$ is bounded by $R_{\max} = \max_{p(v_2|s_2)} I(V_2;S_2|S_1) = H(S_2|S_1)$, allowing $R'$ to be greater than $R_{\max}$ cannot improve $C_2^{lb}$ any more. i.e., $C_2^{lb}( R' = R_{\max}) = C_2^{lb}( R' >R_{\max})$. Therefore, it is enough to allow $ R'=R_{\max}$ to achieve $C_2^{lb}$, as if the encoder is fully informed with $S_2$.
\item
Although $C_2^{lb}$ is a lower bound on the capacity, it can be significantly greater than the Cover-Chiang and the Gelfand-Pinsker rates for some channel models, as can be seen in this example. Moreover, we can actually state that $C_2^{lb}$ is always greater than or equal to the Gelfand-Pinsker and the Cover-Chiang rates. This is due to the fact that when $R' = 0$, $C_2^{lb}$ coincides with the Cover-Chiang rate, which, in its turn, is always greater than or equal to the Gelfand-Pinsker rate; since $C_2^{lb}$ is also non-decreasing in $R'$, it is obvious that our assertion holds.
\end{enumerate}
\end{example}

\begin{example}[{\it Source coding Case 1 for a  binary-symmetric source and Hamming distortion}] \label{ex:RD1}
  Consider the source $X = S_1 \oplus S_2$, where $S_1, S_2 \sim \iid \mbox{Bernoulli} (0.5)$, and consider the problem setting depicted in Case~1 of the source coding problems. It is sufficient for the decoder to reconstruct $S_1$ with distortion $\mathbb{E} \big[d(S_1, \hat S_1) \big] \leq D$ in order to reconstruct $X$ with the same distortion. Furthermore, the two rate-distortion problem settings illustrated in Figure~\ref{fig:equivalent_RD} are equivalent.
  \begin{figure}[h!]
    \centering \small
    \psfrag{a}{Setting $1$}
    \psfrag{b}{Setting $2$}
    \psfrag{d}{Dec}
    \psfrag{e}{Enc} \footnotesize
    \psfrag{s}{$S_1$}
    \psfrag{v}{$S_2$}
    \psfrag{c}{$\hat S_1$}
    \psfrag{x}{$X$}
    \psfrag{r}{$R+ R'$}
    \psfrag{R}{$R$}
    \psfrag{t}{$ R'$}
    \psfrag{y}{$\hat X$}
    \includegraphics[scale = 0.6]{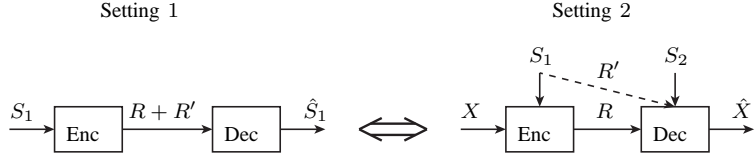}
    \caption{The equivalent rate-distortion problem for Case~1 for the source $X = S_1 \oplus S_2$ where $S_1, S_2 \sim \iid \mbox{Bernoulli} (0.5)$.}
    \label{fig:equivalent_RD}
  \end{figure}

  For every achievable rate in Setting 1, $\E{ d(S_1, \hat S_1)} \leq D$. Denote $\hat X \triangleq \hat S_1 \oplus S_2$, then, $d(S_1,\hat S_1) = S_1 \oplus \hat S_1 = (S_1 \oplus S_2) \oplus (\hat S_1 \oplus S_2) = X \oplus \hat X = d(X, \hat X)$ and, therefore, $\E{d(S_1, \hat S_1)} \leq D$ in Setting 1 $\Rightarrow\ \E{d(X,\hat X)} \leq D$ in Setting 2. In the same way, for Setting 2, denote $\hat S_1 \triangleq \hat X \oplus S_2$. Then, $d(X, \hat X) = X \oplus \hat X = S_1 \oplus \hat S_1$ and, therefore, $\E{d(X,\hat X)} \leq D$ in Setting 2 $\Rightarrow \ \E{d(S_1,\hat S_1)} \leq D$ in Setting 1. Hence, we can conclude that the two settings are equivalent and, for any given $0 \leq D$ and  $0 \leq  R'$, the rate-distortion function is
\begin{align}
  \label{eq:116}
  R(D) = \left\{
  \begin{array}{l r}
    1-H(D) -  R' & 1-H(D)-  R' \geq 0\\
    0& 1-H(D) -  R' < 0
  \end{array} \right..
\end{align}
In Figure~\ref{fig:Example2_RD} we present the plot resulting for this example. It is easy to verify that the Wyner \& Ziv rate and the Cover \& Chiang rate for this setting are $R_{WZ}(D) = R_{CC}(D) = \max \big\{1-H(D), 0\big\}$.
\begin{figure}[h!]
  \centering \footnotesize
  \psfrag{a}{$ R' = 0$}
  \psfrag{b}{$ R' = 0.1$}
  \psfrag{c}{$ R' = 0.2$}
  \psfrag{d}{$ R' = 0.3$}
  \psfrag{e}{$ R' = 0.4$}
  \psfrag{f}{$0.1$ bit}
  \psfrag{D}{$D$}
  \psfrag{R}{$R$ [{\it bits}]}
  \includegraphics[scale = 0.6]{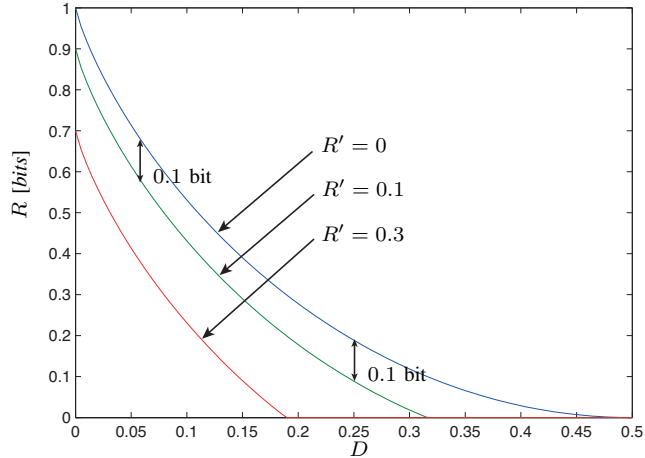}
  \caption{Example 2. {\it Source coding Case~1 for binary-symmetric source and Hamming distortion}. The source is given by $X = S_1 \oplus S_2$, where $S_1, S_2 \sim \mbox{Bernoulli} (0.5)$. The graph shows the rate-distortion function for different values of $R'$.}
  \label{fig:Example2_RD}
\end{figure}
\end{example}

\begin{example}[{\it Geometric programming and the Wyner-Ziv problem}]
  Consider the traditional Wyner-Ziv~\cite{1055508} problem where the source, $X$, and the side information, $S$, are distributed according to $X \sim {\rm Bernoulli}(0.5)$ and $\Pr\{S \neq X\} = 0.3$. We calculated the rate-distortion function, $R(D) = \min_{p(u|x) p(\hat x|u,s)} I(U;X|S)$ s.t. $\E{d(X,\hat X) \leq D}$, by using three different methods: first by using \cite[Theorem II]{1055508}, second by using \cite[Proposition 3]{Chiang04geometricprogramming} and third by using the geometric programming solution we introduced in Theorem~\ref{theoremGP}. The plot resulting from this computation is brought in Figure~\ref{fig:ExampleGeoProgWZ}.
  \begin{figure}[h!]
    \centering \small
    \psfrag{a}{Wyner-Ziv rate}
    \psfrag{b}{Chiang-Boyd lower bound}
    \psfrag{c}{Dual geometric program}
    \psfrag{d}{$D$}
    \psfrag{r}{Rate [{\it bits}]}
    \includegraphics[scale = 0.7]{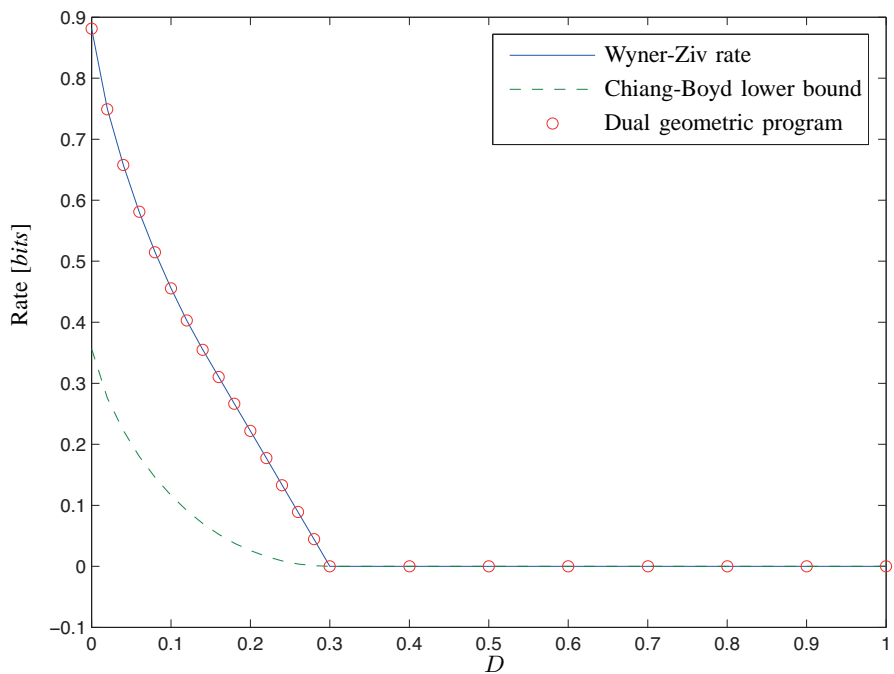}
    \caption{Example 3. {\it Geometric programming and Wyner-Ziv.} The source and the side information distribute $X \sim {\rm Bernoulli} (0.5)$ and $\Pr\{S \neq X\} = 0.3$.}
    \label{fig:ExampleGeoProgWZ}
  \end{figure}

  It can be seen in the figure that the geometric program, which was calculated according to Theorem~\ref{theoremGP}, is tight to the Wyner-Ziv rate.
\end{example}

%
\begin{example}[{\it Geometric programming and source coding Case 1}] 
  \begin{figure}[h!]
    \centering
    \psfrag{s}{$S_1$}
    \psfrag{S}{$S_2$}
    \psfrag{X}{$X$}
    \psfrag{Z}{$Z_0$}
    \psfrag{z}{$Z_1$}
    \includegraphics[scale = 0.8]{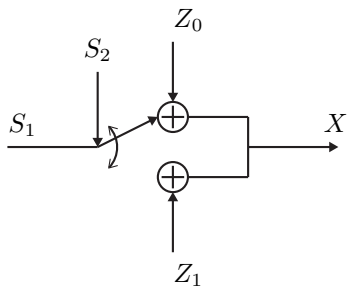}
    \caption{Example 4. {\it Source coding Case 1} with binary symmetric source generation, as given in~(\ref{eq:29})}
    \label{fig:binary_source}
  \end{figure}

  Again, consider a rate-distortion problem as outlined in Case~1 with a binary-symmetric source and Hamming distortion. The source, $X$, is the output of the system illustrated in Figure~\ref{fig:binary_source}, $S_1, S_2 \sim \iid \mbox{Bernoulli} (0.5)$, $S_2$ is controlling a switch, $Z_0 \sim \mbox{Bernoulli}\left(0.3 \right)$ and $Z_1 \sim \mbox{Bernoulli} \left(0.001 \right)$. The output of this system can be expressed as
  \begin{align}
    \label{eq:29}
    X = \left\{
      \begin{array}[l]{l r}
        S_1 \oplus Z_0, & S_2 = 0\\
        S_1 \oplus Z_1, & S_2 = 1
      \end{array}\right..
  \end{align}
 This source coding problem was introduced by Cheng, Stankovic and Xiong~\cite{DBLP:journals/tit/ChengSX05} for the case where the users are not allowed to share with each other their partial side information ($R'=0$). The rate-distortion expression for this problem is $ R_1(D) = \min  I(U;X,S_1|V_1) - I(U;S_2|V_1)$, where the minimization is over all $p(v_1|s_1) p(u|x, s_1, v_1)p(\hat{x}|u,s_2, v_1)$  s.t.  $R' \geq I(V_1;S_1|S_2)$ and that $\mathbb{E}\Big[\frac{1}{n}\sum_{i=1}^nd(X_i,\hat{X}_i)\Big] \leq D$. We solve this example by using the geometric programming expression we developed in Theorem~\ref{theoremGP}. The algorithm we developed in order to solve this problem uses some of the main principles we used in the algorithm that we developed for Example~\ref{ex:CC}  (Algorithm~1) and that is detailed in Section~\ref{sec:algorithm}. For this reason, we now bring a summary of the algorithm for this example.

First, as claimed in Section~\ref{sec:geom-progr}, it is possible to write the expression for the rate-distortion as $R(D) = \min I(T;X,S_1|V_1) - I(T;S_2|V_1)$ where the minimization is over all $w(v_1|s_1) q(t|x,s_1,v_1)$ s.t.  $R' \geq I(V_1 ; S_1|S_2)$ and that $\mathbb{E}\Big[\frac{1}{n}\sum_{i=1}^nd(X_i,T(S_2,V_1))\Big] \leq D$. The variable $T$ is the mapping $T: \mathcal{S}_2 \times \mathcal{V}_1 \to \hat{\mathcal{X}}$. It can be verified that for every fixed probability, $w(v_1|s_1)$, the function $I(T;X,S_1|V_1) - I(T;S_2|V_1)$ is a convex function of $q(t|x,s_1,v_1)$. Now, we construct a fine grid of probabilities $w(v_1|s_1)$, and we keep those $w(v_1|s_1)$ for which $R' \geq I(V_1;S_1|S_2) \geq R'-\epsilon$ in the array $\mathcal{W}^*$. At this point, for every $w(v_1|s_1) \in \mathcal{W}^*$ that we kept, we let $R_w(D)$ be the solution for the following geometric program
\begin{align}
  \label{eq:111}
  \begin{array}[l]{l l}
    \mbox{maximize} & \sum_{x,s_1,v_1} \alpha_{x,s_1,v_1} p(x,s_1,v_1) - \gamma D \\
    \mbox{subject to} & \alpha_{x,s_1,v_1} +  \sum_{s_2} p(s_2|x, s_1) \Big[ \log p(x,s_1|s_2,v_1) - \gamma d\big( x,t(s_2,v_1) \big) - y_{x,s_1,s_2,v_1,t} \Big]  \leq 0,\ \ \forall x,s_1,v_1,t,\\
    & \log \Big( \sum_{x,s_1} \exp \big\{ y_{x,s_1,s_2,v_1,,t}\big\} \Big) \leq 0, \quad \forall s_2,v_1,t,\\
    & \gamma \geq 0,
  \end{array}
\end{align}
where the variables of the maximization are $\balpha \in \mathbb{R}^{|\mathcal{X}| |\mathcal{S}_1| |\mathcal{V}_1|}, \gamma \in \mathbb{R}$ and ${\bf y} \in \mathbb{R}^{|\mathcal{X}| |\mathcal{S}_1| |\mathcal{S}_2| |\mathcal{V}_1| |\mathcal{T}|}$. It can be verified that this geometric program is a generalization of the geometric program we developed in Theorem~\ref{theoremGP} and that it corresponds to the problem of minimizing $I(T;X,S_1|V_1) - I(T;S_2|V_1)$ over $q(t|x,s_1,v_1)$ s.t. $\mathbb{E}\Big[\frac{1}{n}\sum_{i=1}^nd(X_i,\hat{X}_i)\Big] \leq D$ (for a fixed probability $w(v_1|s_1)$). Therefore, all we are left to do now is to declare
\begin{align}
  R(D) = \min_{w(v_1|s_1) \in \mathcal{W}^*} R_w(D).
\end{align}
This concludes the summary of the algorithm for solving this example.

The numeric result of the calculation of this rate-distortion function is brought in Figure~\ref{fig:RD_ex2}.
  \begin{figure}[h!]
    \centering \small 
    \psfrag{r}{Rate [{\it bits}]}
    \psfrag{D}{$D$}
    \psfrag{a}{$ R' = 0$}
    \psfrag{b}{$ R' = 0.25$}
    \psfrag{c}{$ R' = 0.5$}
    \psfrag{d}{$ R' = 0.75$}
    \psfrag{e}{$ R' = 1$}
    \includegraphics[scale = 0.8]{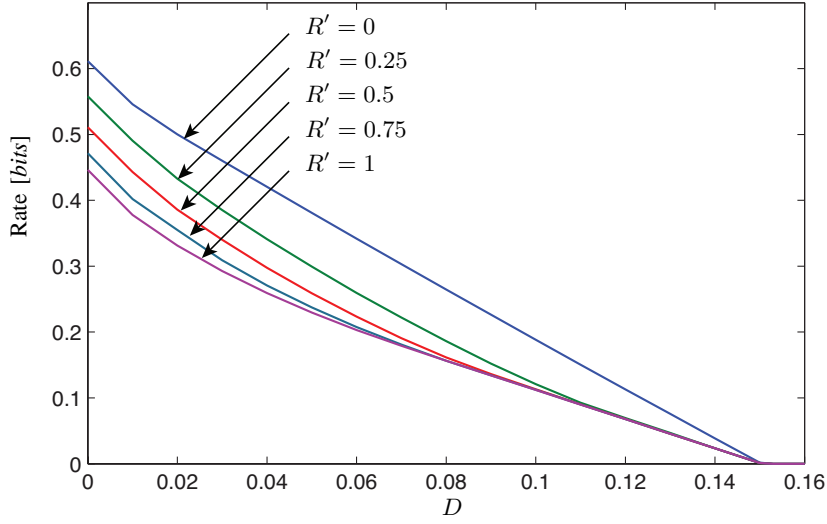}
    \caption{Example 4. {\it Geometric programming and source coding Case 1.} The source $X$ is depicted in Figure~\ref{fig:binary_source} and the distortion is the Hamming distortion.}
    \label{fig:RD_ex2}
  \end{figure}
\end{example}

\section{Semi-Iterative Algorithm}
\label{sec:algorithm}

In this section we provide algorithms that numerically calculate the lower bound on the capacity of Case 2 of the channel coding problems. The calculation of the Gelfand-Pinsker and the Wyner-Ziv problems has been addressed in many papers in the past, including~\cite{DBLP:journals/tit/HeegardG83},~\cite{Willems1983},~\cite{1365218} and~\cite{DBLP:journals/tit/ChengSX05}. All these algorithms are based on Arimoto's~\cite{Arimoto1972a} and Blahut's~\cite{Blahut72computationof} algorithms and on the fact that the Wyner-Ziv and the Gelfand-Pinsker problems can be presented as convex optimization problems. On the contrary, our problems are not convex in all of their optimization variables and, therefore, cannot be presented as convex optimization problems. In order to solve our problems we devised a different approach which combines a grid-search and a Blauhut-Arimoto-like algorithm. In this section, we provide the mathematical justification for those two algorithms. Other algorithms to numerically compute the channel capacity or the rate-distortion of the rest of the cases presented in this paper can be derived using the principles that we describe in this section.

\subsection{An algorithm for computing the lower bound on the capacity of Case 2}
\label{sec:alg_channel}
\begin{figure}[h!]
  \centering
  \small
  \psfrag{A}{\ \ Encoder}
    \psfrag{B}{\ \ Channel}
    \psfrag{D}{\ \ Decoder}
    \psfrag{W}{$W$}
    \psfrag{w}{$\hat{W}$}
    \psfrag{X}{$X^n$}
    \psfrag{x}{$X_i$}
    \psfrag{Y}{$Y^n$}
    \psfrag{y}{$Y_i$}
    \psfrag{T}{$S_1^n$}
    \psfrag{t}{$S_{1,i}$}
    \psfrag{R}{$S_2^n$}
    \psfrag{r}{$S_{2,i}$}
    \psfrag{J}{$R'$}
    \psfrag{K}{$R'$}
    \includegraphics[scale = 0.5]{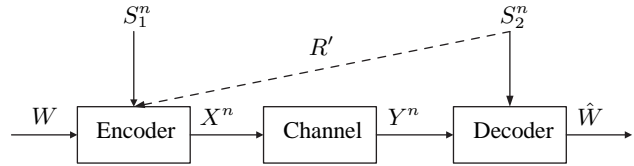}
  \caption{Channel coding: Case 2. $C^{lb}_2 = \max I(U;Y,S_2|V_2) - I(U;S_1|V_2)$, where the maximization is over all PMFs $w(v_2|s_2) p(u|s_1,v_2)p(x|s_1,v_2,u)$ such that $R' \geq I(V_2;S_2|S_1)$.}
  \label{fig:alg_channel}
\end{figure}
Consider the channel in Figure~\ref{fig:alg_channel} described by  $p(y|x,s_1,s_2)$ and consider the joint PMF $p(s_1,s_2)$. The capacity of this channel is lower bounded by $\max I(U;Y,S_2|V_2)-I(U;S_1|V_2)$, where the maximization is over all PMFs $p(s_1,s_2) w(v_2|s_2) p(u|s_1,v_2) p(x|s_1,v_2,u) p\big(y|x, s_1,s_2\big)$ such that $R' \geq I(V_2;S_2|S_1)$. Notice that the lower bound expression is not concave in $w(v_2|s_2)$, which is the main difficulty with the computation of it. We first present an outline of the semi-iterative algorithm we developed, then we present the mathematical background and justification for the algorithm and, finally, we present the detailed algorithm.

For any fixed PMF $w(v_2|s_2)$ denote
\begin{align}
  R_w & \triangleq I(V_2;S_2|S_1),\\
  C^{lb}_{2,w} & \triangleq \max_{p(u|s_1,v_2)p(x|u,s_1,v_2)} I(U;Y,S_2|V_2) - I(U;S_1|V_2). \label{eq:103}
\end{align}

Then, the lower bound on the capacity , $C^{lb}_2( R')$, can be expressed as
\begin{align}
  \label{eq:102}
  C^{lb}_2( R') = \max_{ \substack{w(v_2|s_2) \\  {\rm s.t.}\ R' \geq R_w}} \max_{p(u|s_1,v_2) p(x|u,s_1,v_2)} [I(U;Y,S_2|V_2) - I(U;S_1|V_2)] \triangleq \max_{ \substack{w(v_2|s_2) \\  {\rm s.t.}\ R' \geq R_w}} C^{lb}_{2,w}.
\end{align}
The outline of the algorithm is as follows: for any given rate $R' \leq H(S_2|S_1)$, $\epsilon>0$ and $\delta > 0$,
\begin{enumerate}
\item Establish a fine and uniformly spaced grid of legal PMFs, $w(v_2|s_2)$, and denote the set of all of those PMFs as $\mathcal{W}$.
\item Establish the set  $\mathcal{W}^* := \Big\{ w(v_2|s_2) \ |\  w(v_2|s_2) \in \mathcal{W}$ and $R'-\epsilon \leq R_w \leq R'\Big\}$. This set is the set of all PMFs $w(v_2|s_2)$ such that $R_w$ is $\epsilon$-close to $R'$ from below. If $\mathcal{W}^*$ is empty, go back to step 1 and make the grid finer. Otherwise, continue.
\item For every $w(v_2|s_2) \in \mathcal{W}^*$, perform a Blahut-Arimoto-like optimization to find $C^{lb}_{2,w}$ with accuracy of $\delta$.
\item Declare $C_2^{lb}(R') = \max_{w(v_2|s_2) \in \mathcal{W}^*} C^{lb(\epsilon,\delta,\mathcal{W})}_2(R')$.
\end{enumerate}
{\it Remarks}:
(a) We considered only those $R'$s such that $R' \leq H(S_2|S_1)$ since $H(S_2|S_1)$ is the maximal value that $I(V_2;S_2|S_1)$ takes. The interpretation of this is that if the encoder is informed with $S_1$, we cannot increase its side information about $S_2$ in more than $H(S_2|S_1)$. Therefore, for any $H(S_2|S_1) \leq R'$, we can limit $R'$ to be equal to $H(S_2|S_1)$ in order to compute the capacity.
(b) Since $C^{lb}_{2,w}(R')$ is continuous in $w(v_2|s_2)$ and bounded (for example, by $I(X;Y|S_1,S_2)$ from above and by $I(X;Y)$ from below), $C_2^{(\epsilon,\delta,\mathcal{W})}(R')$ can be arbitrarily close to $C_2^{lb}(R')$ for $\epsilon \to 0,\ \delta \to 0$ and $|\mathcal{W}| \to \infty$.

{\bf Mathematical background and justification}

Here we focus on finding the lower bound on the capacity of the channel for a fixed distribution $w(v_2|s_2)$, i.e., finding $C^{lb}_{2,w}$. Note that the mutual information expression $I(U;Y,S_2|V_2) - I(U;S_1|V_2)$ is concave in $p(u|s_1,v_2)$ and convex in $p(x|u,s_1,v_2)$. Therefore, a standard convex maximization technique is not applicable for this problem. However, according to Dupuis, Yu and Willems~\cite{1365218}, we can write the expression for the lower bound as $C^{lb}_{2,w} = \max_{q(t|s_1,v_2)} I(T;Y,S_2|V_2) - I(T;S_1|V_2)$, where $q(t|s_1,v_2)$ is a probability distribution over  the set of all possible strategies $t: \mathcal{S}_1 \times \mathcal{V}_2 \to \mathcal{X}$, the input symbol $X$ is selected using $x=t(s_1,v_2)$ and $p(y|x,s_1,s_2) = p(y|x,s_1,s_2,v_2) = p\big(y|t(s_1,v_2),s_1,s_2,v_2\big)$.
Now, since $I(T;Y,S_2|V_2) - I(T;S_1|V_2)$ is concave in $q(t|s_1,v_2)$, we can use convex optimization methods to derive $C_{2,w}^{lb}$.

Denote the PMF
\begin{align}
  \label{eq:p}
  p(s_1,s_2,v_2,t,y) & \triangleq p(s_1,s_2) w(v_2|s_2) q(t|s_1,v_2) p(y|t,s_1,s_2,v_2),
\end{align}
 and denote also
\begin{align}
  \label{eq:104}
  J_w(q,Q) &\triangleq \sum_{s_1,s_2,v_2,t,y} p(s_1,s_2,v_2,t,y) \log \frac{Q(t|y,s_2,v_2)}{q(t|s_1,v_2)},\\
  Q^*(t|y,s_2,v_2) &\triangleq \frac{\sum_{s_1} p(s_1,s_2,v_2,t,y)}{\sum_{s_1,t'} p(s_1,s_2,v_2,t',y)}.\label{eq:Q}
\end{align}
Notice that $Q^*(t|y,s_2,v_2)$ is a marginal distribution of $p(s_1,s_2,v_2,t,y)$ and that $J_w(q,Q^*) = I(T;Y,S_2|V_2)-I(T;S_1|V_2)$ for the joint PMF $p(s_1,s_2,v_2,t,y)$.

The following lemma is the key for the iterative algorithm.
\begin{lemma}
  \begin{align}
    \label{eq:105}
    C_{2,w}^{lb} = \sup_{q'(t|s_1,v_2)} \max_{Q'(t|y,s_2,v_2)} J_w(q',Q').
  \end{align}
\end{lemma}
The proof for this is brought by Yeung in~\cite{Yeung:2008:ITN:1457455}. In addition, Yeung shows that the two-step alternating optimization procedure converges monotonically to the global optimum if the optimization function is concave. Hence, if we show that $J_w(q,Q)$ is concave, we can maximize it using an alternating maximization algorithm over $q$ and $Q$.
\begin{lemma}
\label{concavityOfJ}
  The function $J_w(q,Q)$ is concave in $q$ and $Q$ simultaneously.
\end{lemma}
We can now proceed to calculate the steps in the iterative algorithm.
 \begin{lemma}
   For a fixed $q,~ J_w(q,Q)$ is maximized for $Q=Q^*$.
 \end{lemma}
 \begin{proof}
   The above follows from the fact that $Q^*$ is a marginal distribution of $p(s_1,s_2) w(v_2|s_2) q(t|s_1,v_2)\linebreak[4] p(y|t,s_1,s_2,v_2)$ and the property of the K-L divergence $D(Q^*\| Q') \geq 0$.
 \end{proof}
 \begin{lemma}
\label{maximizing_q}
   For a fixed $Q,~J_w(q,Q)$ is maximized for $q=q^*$, where $q^*$ is defined by
   \begin{align}
     \label{eq:q}
     q^*(t|s_1,v_2) = \frac{\prod_{s_2,y}Q(t|y,s_2,v_2)^{p(s_2|s_1,v_2) p(y|t,s_1,s_2,v_2)}}{\sum_{t'} \prod_{s_2,y}Q(t|y,s_2,v_2)^{p(s_2|s_1,v_2) p(y|t',s_1,s_2,v_2)}},
   \end{align}
and
\begin{align}
  \label{eq:107}
  p(s_2|s_1,v_2) = \frac{p(s_1,s_2) w(v_2|s_2)}{\sum_{s'_2}p(s_1,s'_2) w(v_2|s'_2)}.
\end{align}
 \end{lemma}
Define $U_w(q)$ in the following way
\begin{align}
  \label{eq:111}
  U_w(q) = \sum_{s_1,v_2} p(s_1,v_2) \max_t \sum_{s_2,y} p(s_2|s_1,v_2) p(y|t,s_1,s_2,v_2) \log \frac{Q^*(t|y,s_2,v_2)}{q(t|s_1,v_2)},
\end{align}
where $Q^*$ is given in $($\ref{eq:Q}$)$, $p(s_1,v_2)$ and $p(s_2|s_1,v_2)$ are marginal distributions of the joint PMF $p(s_1,s_2,v_2,t,y) = p(s_1,s_2) w(v_2|s_2) q(t|s_1,v_2) p(y|t,s_1,s_2,v_2)$.
The following lemma will help us to define a termination condition for the algorithm.
\begin{lemma}
\label{U_w}
  For every $q(t|s_1,v_2)$ the function $U_w(q)$ is an upper bound on $C^{lb}_{w,2}$ and converges to $C_{2,w}^{lb}$ for a large enough number of iterations.
\end{lemma}

\subsection{Semi-iterative algorithm}
\label{sec:semi-iter-algor}
The the algorithm for finding $C^{lb}_2( R')$ is brought in Algorithm~1. Notice that the result of this algorithm, $C_{2}^{(\epsilon,\delta,\mathcal{W})}(R')$, can be arbitrarily close to $C_2^{lb}(R')$ for $\epsilon \to 0,\ \delta\to 0$ and $|\mathcal{W}| \to \infty$.

\begin{algorithm}
\label{algorithm}
  \caption{Numerically calculating $C^{lb}_2(R')$}
{\fontsize{10}{18}\selectfont

  \begin{algorithmic}[1]
    \State Chose $\epsilon>0$, $\delta>0$
    \State Set $R' \gets \min\{R', H(S_2|S_1)\}$  \Comment {the amount of information needed for the encoder to know $S_2$ given $S_1$}
    \State Set $C \gets -\infty$
    \State Establish a fine and uniformly spaced grid of legal PMFs $w(v_2|s_2)$ and name it $\mathcal{W}$
\ForAll {$w$ {\bf in} $\mathcal{W}$}
  \State Compute $R_w$ using \begin{align*} R_w = I(V_2;S_2) - I(V_2;S_1)  \end{align*}

  \If { $R'-\epsilon \le R_{w} \le R'$}
    \State Set $Q(t|y,s_2,v_2)$ to be a uniform distribution over $\{1,2,\dots, |\mathcal{T}|\}$, where $\mathcal{T}$ is the alphabet of $t$.
    \Statex \qquad \quad i.e., $Q(t|y,s_2,v_2) = \frac{1}{|\mathcal{T}|},\ \ \forall t,y,s_2,v_2$

    \Repeat
      \State Set $q(t|s_1,v_2) \gets q^*(t|s_1,v_2)$ using
      \begin{align*}
        q^*(t|s_1,v_2) = \frac{\prod_{s_2,y} Q(t|y,s_2,v_2)^{p(s_2|s_1,v_2) p(y|t,s_1,s_2,v_2)}}{ \sum_{t'} \prod_{s_2,y} Q(t'|y,s_2,v_2)^{p(s_2|s_1,v_2) p(y|t',s_1,s_2,v_2)}}
      \end{align*}
      \State Set $(Q(t|y,s_2,v_2) \gets Q^*(t|y,s_2,v_2)$ using
      \begin{align*}
        Q^*(t|y,s_2,v_2) = \frac{\sum_{s_1} p(s_1,s_2,v_2,t,y)}{\sum_{s_1,t'} p(s_1,s_2,v_2,t',y)}
      \end{align*}
      \State Compute $J_w(q,Q)$ using
      \begin{align*}
        J_w(q,Q) = \sum_{s_1,s_2,v_2,t,y} p(s_1,s_2,v_2,t,y) \log \frac{Q(t|y,s_2,v_2)}{q(t|s_1,v_2)}
      \end{align*}
      \State Compute $U_w(q)$ using
      \begin{align*}
        U_w(q) = \sum_{s_1,v_2} p(s_1,v_2) \max_t \sum_{s_2,y} p(s_2|s_1,v_2) p(y|t,s_1,s_2,v_2) \log \frac{Q^*(t|y,s_2,v_2)}{q(t|s_1,v_2)}
      \end{align*}

    \Until {$U_w(q) - J(q,Q) < \delta$}

    \If {$C \le J_w(q,Q)$}
      \State Set $C \gets J_w(q,Q)$
    \EndIf
  \EndIf

\EndFor

\If {$C < 0$}  \Comment {there is no PMF $w(v_2|s_2) \in \mathcal{W}$ such that $R_{w}$ is $\epsilon$-close to $R'$ from below}
  \State {\bf go to} line 4 and make the grid finer
\EndIf
\State {\bf Declare} $C_2^{lb(\epsilon,\delta,\mathcal{W})}(R') = C$
  \end{algorithmic}}
\end{algorithm}

\section{Open Problems}
\label{sec:open_problems}
In this section we discuss the generalization of the channel capacity and the rate-distortion problems that we presented in Section \ref{sec:main-results}. We now consider the cases where the encoder and the decoder are informed with both a rate-limited description of the ESI and a rate-limited description of the DSI simultaneously, as illustrated in Figure~\ref{fig:channel12}. Although proofs for the converses are not provided in this paper and are considered as open problems, we do provide achievability schemes for both problems.

\subsection{A lower bound on the capacity of a channel with two-sided increased partial side information}
\label{sec:channel_open_problem}
\begin{figure}[h!]
  \centering
  \small
  \psfrag{A}{\ Encoder}
  \psfrag{B}{\ \ Channel}
  \psfrag{D}{\ Decoder}
  \psfrag{W}{$W$}
  \psfrag{w}{$\hat{W}$}
  \psfrag{X}{$X^n$}
  \psfrag{x}{$X_i$}
  \psfrag{Y}{$Y^n$}
  \psfrag{y}{$Y_i$}
  \psfrag{T}{$S_1^n$}
  \psfrag{t}{$S_{1,i}$}
  \psfrag{R}{$S_2^n$}
  \psfrag{r}{$S_{2,i}$}
  \psfrag{J}{$R'_{1}$}
  \psfrag{K}{$R'_{2}$}

  \includegraphics[scale = 0.5]{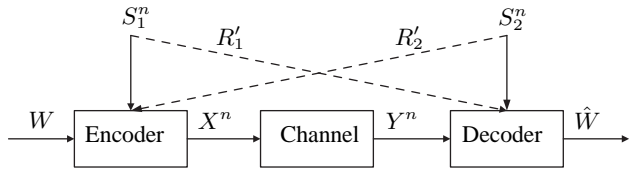}
  \caption{A lower bound on the capacity of a channel with two-sided increased partial side information: $C_{12} \geq \max I(U;Y,S_2|V_1,V_2) - I(U;S_1|V_1,V_2)$, where the maximization is over all PMFs $p(v_1|s_1) p(v_2|s_2) p(u|s_1,v_1,v_2) p(x|u,s_1,v_1,v_2)$ such that $R'_{1} \geq I(V_1;S_1) - I(V_1; Y,S_2,V_2)$ and $R'_{2} \geq I(V_2;S_2) - I(V_2;S_1,V_1)$.}
  \label{fig:channel12}
\end{figure}
Consider the channel illustrated in Figure~\ref{fig:channel12}, where $(S_{1,i},S_{2,i})~\iid \sim p(s_1,s_2)$. The encoder is informed with the ESI $(S_1^n)$ and rate-limited DSI and the decoder is informed with the DSI $(S_2^n)$ and rate-limited ESI. An $(n, 2^{nR}, 2^{nR'_{1}}, 2^{nR'_{2}})$ code for the discussed channel consists of three encoding maps:
\begin{align*}
  &f_{v1}: \quad {\cal S}_1^n \mapsto \{1,2,\dots, 2^{nR'_{1}}\},\\
  &f_{v2}: \quad {\cal S}_2^n \mapsto \{1,2,\dots, 2^{nR'_{2}}\},\\
  &f: \quad \{1,2,\dots,2^{nR}\} \times {\cal S}_1^n \times \{1,2,\dots,2^{nR'_{2}}\} \mapsto {\cal X}^n,
\end{align*}
and a decoding map:
\begin{align*}
  &g: {\cal Y}^n \times {\cal S}_2^n \times \{1,2,\dots,2^{nR'_{1}}\} \mapsto \{1,2,\dots,2^{nR}\}.
\end{align*}

{\it Fact 1:} The channel capacity, $C^*_{12}$, of this channel coding setup is bounded from below as follows:
\begin{align}
  C_{12}^* &\geq \max_{
    \substack{ p(v_1|s_1) p(v_2|s_2) p(u|s_1,v_1,v_2) p(x|u,s_1,v_1,v_2)\\
      \mbox{s.t. \ } R'_{1} \geq I(V_1;S_1) - I(V_1; Y,S_2,V_2)\\
      \quad R'_{2} \geq I(V_2;S_2) - I(V_2;S_1)}}
  I(U;Y,S_2|V_1,V_2) - I(U;S_1|V_1,V_2),
\end{align}
for some joint distribution $p(s_1,s_2,v_1,v_2,u,x,y)$ and $U,V_1$ and $V_2$ are some auxiliary random variables.

The proof for the achievability follows closely the proofs given in Appendix~\ref{sec:proof-CC} and, therefore, we only provide the outline of the achievability. The main steps of the achievability scheme are outlined in the following.

{\it Sketch of proof of Achievability for Fact~1:}
 (a) {\it The ESI encoder wants to describe $S_1^n$ to the decoder with rate of $R'_1$}. We generate $2^{n ( I(V_1;S_1) + \epsilon )}$ sequences $V_1^n\ \iid \sim p(v_1)$ and randomly distribute them into $2^{n \big( I(V_1; S_1) - I(V_1 ; Y,S_2,V_2) + 2\epsilon\big) }$ bins; each bin contains $2^{n( I(V_1;Y,S_2,V_2) - \epsilon)}$ codewords. The ESI encoder is given the sequence $s_1^n$ and first looks for a sequence $v_1^n$ that is jointly typical with $s_1^n$. If there is such a codeword, the ESI encoder sends the index of the bin that contains $v_1^n$ to the decoder. The decoder, given $y^n,s_2^n,v_2^n$, looks for a unique codeword in the received bin that is jointly typical with $y^n,s_2^n,v_2^n$. Since there are more than $2^{nI(V_1;S_1)}$ sequences $V_1^n$, the ESI encoder is assured with high probability to find a sequence $v_1^n$ such that $(v_1^n,s_1^n) \in \T(V_1,S_1)$. Since, in addition, there are less than $2^{nI(V_1;Y,S_2,V_2)}$ codewords in the bin, the decoder is assured to find a unique sequence $v_1^n$ in the bin such that $(v_1^n,y^n,s_2^n,v_2^n) \in \T(V_1,Y,S_2,V_2)$ with high probability. Therefore, the constraint on the shared ESI is maintained if $R'_1 > I(V_1;S_1) - I(V_1;Y,S_2,V_2)$.

(b) {\it The DSI encoder wants to describe $S_2^n$ to the channel's encoder with a rate of $R'_2$}. We generate $2^{n(I(V_2;S_2) + \epsilon)}$ sequences $V_2^n \sim \iid p(v_2)$ and randomly distribute them into $2^{n \big( I(V_2;S_2) - I(V_2;S_1,V_1) +2\epsilon \big) }$ bins; each bin contains $2^{n(I(V_2;S_1,V_1) - \epsilon)}$ codewords. The DSI encoder, given $s_2^n$, first looks for a sequence $v_2^n$ that is jointly typical with $s_2^n$. If there is such a codeword, the DSI encoder sends the index of the bin where $v_2^n$ is located to the channel's encoder. The channel's encoder, given $s_1^n,v_1^n$, looks for a unique sequence $v_2^n$ in the received bin that is jointly typical with $s_1^n,v_1^n$. Since there are more than $2^{nI(V_2;S_2)}$ sequences $V_2^n$, the DSI encoder is assured with high probability to find such a sequence $v_2^n$ such that $(v_2^n,s_2^n) \in \T(V_2,S_2)$. In its turn, the channel's encoder is also assured with high probability to find the unique sequence $v_2^n$ in its received bin such that $(v_2^n,s_1^n,v_1^n) \in \T(V_2,S_1,V_1)$, since there are less than $2^{nI(V_2;S_1,V_1)}$ codewords $V_2^n$ in the bin. Therefore, the constraint of the shared DSI is maintained if $R'_2 > I(V_2;S_2) - I(V_2;S_1,V_1)$.

(c) {\it The encoder wants to send the message $W$ to the decoder}. For each $v_1^n,v_2^n$ we generate $2^{n(I(U;Y,S_2|V_1,V_2) - \epsilon)}$ sequences $U^n$ using the PMF $ p(u^n|v_1^n,v_2^n) = \prod_{i=1}^n p(u_i|v_{1,i},v_{2,i})$ and randomly distribute them into $2^{n \big( I(U;Y,S_2,|V_1,V_2) - I(U;S_1|V_1,V_2) - 2\epsilon \big)}$ bins; each bin contains $2^{n(I(U;S_1|V_1,V_2) + \epsilon)}$ codewords. The encoder, given $s_1^n,v_1^n,v_2^n$ and the message $W$, looks in the bin number $W$ for a sequence $u^n$ that is jointly typical with $s_1^n,v_1^n,v_2^n$ and sends $x_i = f(u_i,s_{1,i},v_{1,i},v_{2,i})$ over the channel at time $i$. The decoder receives $y^n,s_2^n,v_1^n,v_2^n$ and first looks for a unique sequence $u^n$ that is jointly typical with $y^n,s_2^n,v_1^n,v_2^n$. Upon finding the desired sequence $u^n$, the decoder declares $\hat W$ to be the index of the bin that contains $u^n$. Having less than $2^{nI(U;Y,S_2|V_1,V_2)}$ sequences $U^n$ assures with high probability that decoder will identify a unique sequence $u^n$ such that $(u^n,y^n,s_2^n,v_1^n,v_2^n) \in \T(U,Y,S_2,v_1,v_2)$. This is also valid because the Markov relation $(U,V_1,V_2) - (X,S_1,S_2) - Y$ implies that $(u^n,v_1^n,v_2^n,x^n,s_1^n,s_2^n,y^n) \in \T(U,V_1,V_2,X,S_1,S_2,Y)$. In addition, since in each of the encoder's bins there are more than $2^{nI(U;S_1|V_1,V_2)}$ codewords $U^n$, the encoder is assured with high probability to find a sequence $u^n$ in the bin indexed $W$ such that $(u^n,s_1^n,v_1^n,v_2^n) \in \T(U,S_1,v_1,v_2)$. We can conclude that if  $R < I(U;Y,S_2|V_1,V_2) - I(U;S_1|V_1,V_2)$ is maintained, then a reliable communication over the channel is achievable; namely, it is possible to find a sequence of codes such the $\Pr\{\hat W \neq W\}$ goes to zero as the block length goes to infinity. This concludes the sketch of the achievability.

\subsection{An upper bound on the rate-distortion with two-sided increased partial side information}
\label{sec:channel_open_problem}
\begin{figure}[h!]
  \centering
  \small
  \psfrag{A}{\ Encoder}
  \psfrag{B}{\ Decoder}
  \psfrag{X}{$X^n$}
  \psfrag{x}{$\hat{X}^n$}
  \psfrag{S}{$\hat{X}^n$}
  \psfrag{E}{$S_1^n$}
  \psfrag{e}{$S_{1,i}$}
  \psfrag{R}{$S_2^n$}
  \psfrag{r}{$S_{2,i}$}
  \psfrag{J}{$R'_{2}$}
  \psfrag{K}{$R'_{1}$}

  \includegraphics[scale = 0.5]{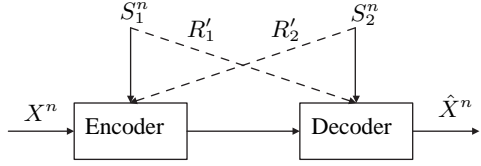}
  \caption{An upper bound on the rate-distortion with two-sided increased partial side information: $R_{12}(D) \leq \min I(U;X,S_1|V_1,V_2) - I(U;S_2|V_1,V_2)$, where the minimization is over all PMFs $p(v_1|s_1) p(v_2|s_2) p(u|x,s_1,v_1,v_2) p(\hat x|u,s_2,v_1,v_2)$ such that $R'_{1} \geq I(V_1;S_1) - I(V_1; S_2)$, $R'_{2} \geq I(V_2;S_2) - I(V_2;X,S_1,V_1)$ and $\mathbb{E}\Big[\frac{1}{n}\sum_{i=1}^nd(X,\hat X)\Big]\leq D$.}
  \label{fig:RD12}
\end{figure}
Consider the rate-distortion problem illustrated in Figure~\ref{fig:RD12}, where the source $X$ and the side information $S_1,S_2$ are distributed $(X_i,S_{1,i},S_{2,i}) \sim \iid p(x,s_1,s_2)$. The encoder is informed with the ESI $(S_1^n)$ and rate-limited DSI and the decoder is informed with the DSI $(S_2^n)$ and rate-limited ESI. An $(n, 2^{nR}, 2^{nR'_{1}}, 2^{nR'_{2}}, D)$ code for the discussed rate-distortion problem consists of three encoding maps:
\begin{align*}
  &f_{v1}: \quad {\cal S}_1^n \mapsto \{1,2,\dots, 2^{nR'_{1}}\},\\
  &f_{v2}: \quad {\cal S}_2^n \mapsto \{1,2,\dots, 2^{nR'_{2}}\},\\
  &f: \quad {\cal X}^n \times {\cal S}_1^n \times \{1,2,\dots,2^{nR'_{2}}\} \mapsto \{1,2,\dots,2^{nR}\},
\end{align*}
and a decoding map:
\begin{align*}
  &g: \{1,2,\dots,2^{nR}\} \times {\cal S}_2^n \times \{1,2,\dots,2^{nR'_{1}}\} \mapsto \hat{{\cal X}}^n.
\end{align*}

{\it Fact 2:} For a given distortion, $D$, and a given distortion measure, $d(X,\hat X) :\ \ \mathcal{X} \times \hat{\mathcal{X}} \mapsto \mathbb{R}^+$, the rate-distortion function $R^*_{12}(D)$ of this setup is bounded from above as follows:
\begin{align}
  \label{eq:R12}
  R_{12}^*(D) &\leq \min_{
    \substack{ p(v_1|s_1) p(v_2|s_2) p(u|x,s_1,v_1,v_2) p(\hat x|u,s_2,v_1,v_2)\\
      \mbox{s.t. \ } R'_{1} \geq I(V_1;S_1) - I(V_1; S_2,V_2)\\
      \quad R'_{2} \geq I(V_2;S_2) - I(V_2;X,S_1,V_1)}}
  I(U;X,S_1|V_1,V_2) - I(U;S_2|V_1,V_2),
\end{align}
for some joint distribution $p(x,s_1,s_2,v_1,v_2,u,\hat x)$ where $\mathbb{E}\Big[\frac{1}{n}\sum_{i=1}^nd(X_i,\hat X_i)\Big]\leq D$ and  $U,V_1$ and $V_2$ are some auxiliary random variables.

The achievability proof is outlined in the following. The steps of the proof resemble the steps of the achievability proof for Fact~1.

{\it Sketch of proof of Achievability for Fact~2:}
(a) {\it The ESI encoder wants to describe $S_1^n$ to the decoder with a rate of $R'_1$}. We generate $2^{n ( I(V_1;S_1) + \epsilon )}$ sequences $V_1^n\ \iid \sim p(v_1)$ and randomly distribute them into $2^{n \big( I(V_1; S_1) - I(V_1 ; S_2,V_2) + 2\epsilon\big) }$ bins; each bin contains $2^{n( I(V_1;S_2,V_2) - \epsilon)}$ codewords. The ESI encoder is given the sequence $s_1^n$ and first looks for a sequence $v_1^n$ that is jointly typical with $s_1^n$. If there is such a codeword, the ESI encoder sends the index of the bin that contains $v_1^n$ to the decoder. The decoder, given $s_2^n,v_2^n$, looks for a unique codeword in the received bin that is jointly typical with $s_2^n,v_2^n$. Since there are more than $2^{nI(V_1;S_1)}$ sequences $V_1^n$, the ESI encoder is assured with high probability to find a sequence $v_1^n$ such that $(v_1^n,s_1^n) \in \T(V_1,S_1)$. Since, in addition, there are less than $2^{nI(V_1;S_2,V_2)}$ codewords in the bin, the decoder is assured with high probability to find a unique sequence $v_1^n$ in the bin such that $(v_1^ns_2^n,v_2^n) \in \T(V_1,S_2,V_2)$. Therefore, the constraint on the rate of the shared ESI is maintained  if $R'_1 > I(V_1;S_1) - I(V_1;S_2,V_2)$.

(b) {\it The DSI encoder wants to describe $S_2^n$ to the source encoder with a rate of $R'_2$}. We generate $2^{n(I(V_2;S_2) + \epsilon)}$ sequences $V_2^n \sim \iid p(v_2)$ and randomly distribute them into $2^{n \big( I(V_2;S_2) - I(V_2;X,S_1,V_1) +2\epsilon \big) }$ bins; each bin contains $2^{n(I(V_2;X,S_1,V_1) - \epsilon)}$ codewords. The DSI encoder, given $s_2^n$, first looks for a sequence $v_2^n$ that is jointly typical with $s_2^n$. If there is such a codeword, the DSI encoder sends the index of the bin where $v_2^n$ is located to the source encoder. The source encoder, given $x^n,s_1^n,v_1^n$, looks for a unique sequence $v_2^n$ in the received bin that is jointly typical with $x^n,s_1^n,v_1^n$. Since there are more than $2^{nI(V_2;S_2)}$ sequences $V_2^n$, the DSI encoder is assured with high probability to find a sequence $v_2^n$ such that $(v_2^n,s_2^n) \in \T(V_2,S_2)$. At the same time, the source encoder is assured with high probability to find the unique sequence $v_2^n$ in its received bin such that $(v_2^n,x^n,s_1^n,v_1^n) \in \T(V_2,X,S_1,V_1)$, since there are less than $2^{nI(V_2;X,S_1,V_1)}$ codewords $V_2^n$ in the bin. Therefore, the constraint on the rate of the shared DSI is maintained if $R'_2 > I(V_2;S_2) - I(V_2;X,S_1,V_1)$.

(c) {\it The source encoder wants to describe the source $X$ to the decoder with distortion smaller than or equal to $D$; that is $\E{d(X,\hat X)}\leq D$}. For each $v_1^n,v_2^n$ we generate $2^{n(I(U;X,S_1|V_1,V_2) + \epsilon)}$ sequences $U^n$ using the PMF $p(u^n|v_1^n,v_2^n) = \prod_{i=1}^n p(u_i|v_{1,i},v_{2,i})$ and randomly distribute them into $2^{n \big( I(U;X,S_1,|V_1,V_2) - I(U;S_2|V_1,V_2) + 2\epsilon \big)}$ bins; each bin contains $2^{n(I(U;S_2|V_1,V_2) - \epsilon)}$ codewords. The source encoder, given $x^n,s_1^n,v_1^n,v_2^n$, looks for a sequence $u^n$ that is jointly typical with $x^n,s_1^n,v_1^n,v_2^n$ and sends the index of the bin that contains $u^n$ to the decoder. The decoder, given $s_2^n,v_1^n,v_2^n$, looks for a unique sequence $u^n$ in the received bin that is jointly typical with $s_2^n,v_1^n,v_2^n$. Upon finding the desired sequence $u^n$, the decoder declares $\hat x_i = g(u_i,s_{2,i},v_{1,i},v_{2,i})$ for $i \in \{1,2,\dots,n\}$ to be the reconstruction of the source $x^n$. Having more than $2^{nI(U;X,S_1|V_1,V_2)}$ sequences $U^n$ assures the encoder with high probability to find a sequence $u^n$ such that $(u^n,x^n,s_1^n,v_1^n,v_2^n) \in \T(U,X,S_1,v_1,v_2)$. Since, in addition, each one of the bins contains there are less than $2^{nI(U;S_2|V_1,V_2)}$ codewords $U^n$, the decoder is assured with high probability to find a unique sequence $u^n$ in the bin such that $(u^n,s_2^n,v_1^n,v_2^n) \in \T(U,S_2,v_1,v_2)$. Therefore, and since the Markov chain $(X,S_1) - (U,S_2,V_1,V_2) - \hat X$ is satisfied, we can conclude that a rate of $R > I(U;X,S_1|V_1,V_2) - I(U;S_2|V_1,V_2)$ allows the decoder to produce $\hat x^n$ that satisfies the distortion constraint with high probability; i.e., that $d(x^n,\hat x^n) \leq D$ with high probability. This concludes the sketch of the proof of the achievability.

\appendices
\label{appendix}

\section{Duality of the Converse of the Gelfand-Pinsker Theorem and the Wyner-Ziv Theorem}
\label{sec:dual-conv-gelf}
In this appendix we provide proofs of the converse of the Gelfand-Pinsker capacity and the converse of the Wyner-Ziv rate in a dual way.

\begin{align}
  \label{eq:5}
  \begin{array}[c]{r |r l | r l} \small
  &  &\mbox{Channel capacity} & &\mbox{Rate-distortion}\\
    \hline
  1& nR  &= H(W) & nR &= H(T)\\
  \hline
  2& & \stackrel{(a)}{\leq} I(W;Y^n) - I(W;S^n) +n\epsilon_n & & \stackrel{(a)}{\geq} I(T;X^n) - I(T;S^n)\\
  \hline
  3& & = \sum_{i=1}^n \Big[ I(W;Y_i|Y^{i-1})  & & = \sum_{i=1}^n \Big[ I(T;X_i|X^{i-1}) \\
  &  & \quad - I(W;S_i|S_{i+1}^n) \Big] +n\epsilon_n & & \quad - I(T;S_i|S_{i+1}^n) \Big] \\
  \hline
  4& & = \sum_{i=1}^n \Big[ I(W,S_{i+1}^n;Y_i|Y^{i-1}) & &  = \sum_{i=1}^n \Big[ I(T,S_{i+1}^n;X_i|X^{i-1}) \\
  & & \quad - I(W,Y^{i-1};S_i|S_{i+1}^n) \Big]  +\Delta - \Delta^* +n\epsilon_n & & \quad - I(W,X^{i-1};S_i|S_{i+1}^n) \Big] +\Delta - \Delta^*\\
  \hline
  5& & \stackrel{(b)}{\leq} \sum_{i=1}^n \Big[ I(W,,Y^{i-1},S_{i+1}^n;Y_i)  & & \stackrel{(b)}{\geq} \sum_{i=1}^n \Big[ I(T,,X^{i-1},S_{i+1}^n;X_i) \\
  & & \quad - I(W,Y^{i-1},S_{i+1}^n;S_i) \Big] +n\epsilon_n & & \quad - I(T,X^{i-1},S_{i+1}^n;S_i) \Big] \\
  \hline
  6& & = \sum_{i=1}^n \Big[ I(U_i;Y_i) - I(U_i;S_i) \Big] + n\epsilon_n, & & = \sum_{i=1}^n \Big[ I(U_i;X_i) - I(U_i;S_i) \Big],
  \end{array}
\end{align}
where
\begin{align}
  \label{eq:7}
  \begin{array}[c]{l l | l l}
  \Delta & = \sum_{i=1}^n I(Y^{i-1};S_i|W,S_{i+1}^n), & \Delta & = \sum_{i=1}^n I(X^{i-1};S_i|T,S_{i+1}^n),\\
  \Delta^* &= \sum_{i=1}^n I(S_{i+1}^n;Y_i|W,Y^{i-1}),& \Delta^* &= \sum_{i=1}^n I(S_{i+1}^n;X_i|T,X^{i-1}),\\
  (a) & \mbox{follows from Fano's inequality} &  (a) & \mbox{follows from Fano's inequality} \\
  & \mbox{and from that fact that } W \mbox{ is} & & \mbox{and from the fact that } T \mbox{ is} \\
  & \mbox{independent of } S^n, & & \mbox{independent of } S^n,\\
  (b) & \mbox{follows from the fact that } S_i \mbox{ is} \quad & (b) & \mbox{follows from the fact that } S_i \mbox{ is} \\
  & \mbox{independent of } S_{i+1}^n. & & \mbox{independent of } S_{i+1}^n \mbox{ and that } X_i \\
  & & &\mbox{is independent of } X^{i-1}.
  \end{array}
\end{align}
By substituting the output $Y$ and the input $X$ in the channel capacity theorem with the input $X$ and the output $\hat X$ in the rate-distortion theorem, respectively, we can observe duality in the converse proofs of the two theorems.

\section{Proof of Theorem \ref{theorem:CC}}
\label{sec:proof-CC}
In this section we provide the proofs for Theorem~\ref{theorem:CC}, Cases 2 and 2$_C$. The results for Case~1, where the encoder is informed with ESI and the decoder is informed with increased DSI, can be derived directly from \cite[Theorem VII]{4608994}. In~\cite{4608994}, Steinberg considered the case where the encoder is fully informed with the ESI and the decoder is informed with a rate-limited description of the ESI. Therefore, by considering the DSI, $S_2^n$, to be a part of the channel's output, we can apply Steinberg's result on the channel depicted in Case 1. For this reason, the proof for this case is omitted.

\subsection{Proof of Theorem \ref{theorem:CC}, Case 2}
\label{proof:CC:case2}

\begin{figure}[h!]
  \centering
  \small
  \psfrag{A}{\ \ Encoder}
    \psfrag{B}{\ \ Channel}
    \psfrag{D}{\ \ Decoder}
    \psfrag{W}{$W$}
    \psfrag{w}{$\hat{W}$}
    \psfrag{X}{$X^n$}
    \psfrag{x}{$X_i$}
    \psfrag{Y}{$Y^n$}
    \psfrag{y}{$Y_i$}
    \psfrag{T}{$S_1^n$}
    \psfrag{t}{$S_{1,i}$}
    \psfrag{R}{$S_2^n$}
    \psfrag{r}{$S_{2,i}$}
    \psfrag{J}{$R'$}
    \psfrag{K}{$R'$}
    \includegraphics[scale = 0.5]{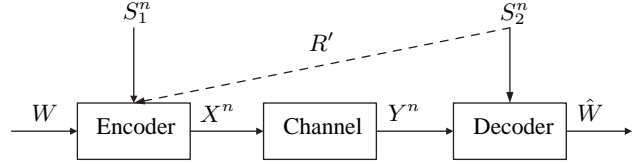}
  \caption{Channel capacity: Case 2. Lower bound: $C^{lb}_2 = \max I(U;Y,S_2|V_2) - I(U;S_1|V_2)$, where the maximization is over all joint PMFs $p(s_1,s_2,v_2,u,x,y)$ that maintain the Markov relations $U - (S_1,V_2) - S_2$ and $V_2 - S_2 - S_1$ and the constraint $R' \geq I(V_2;S_2|S_1)$. Upper bounds: $C^{ub1}_2$ is the result of the same expressions as for the lower bound, except that the maximization is taken over all PMFs that maintain the Markov chain $U-(S_1,V_2) - S_2$, and $C_2^{ub2}$ is the result of the same expressions as for the lower bound, except that this time the maximization is taken over all PMFs that maintain $V_2-S_2-S_1$.}
  \label{fig:channel2}
\end{figure}
The proof of the lower bound, $C_2^{lb}$, is performed in the following way: for the description of the DSI, $S_2$, at a rate $R'$ we use a Wyner-Ziv coding scheme where the source is $S_2$ and the side information is $S_1$. Then, for the channel coding, we use a Gelfand-Pinsker coding scheme where the state information at the encoder is $S_1$, $S_2$ is a part of the channel's output and the rate-limited description of $S_2$ is side information at both the encoder and the decoder. Notice that $I(U;Y,S_2|V_2) - I(U;S_1,|V_2) = I(U;Y,S_2,V_2) - I(U;S_1,V_2)$ and that, since the Markov chain $V_2 - S_2 - S_1$ holds, we can also write $R' \geq I(V_2;S_2) - I(V_2;S_1)$. We make use of these expressions in the following proof.\\
\mbox{\quad}{\bf Achievability:} ({\it Channel capacity Case 2 - Lower bound}). Given $(S_{1,i},S_{2,i})\sim$ i.i.d. $p(s_1,s_2)$ and the memoryless channel $p(y|x,s_1,s_2)$, fix $ p(s_1, s_2, v_2, u, x, y) = p(s_1,s_2) p(v_2|s_2) p(u|s_1,v_2) p(x|u,s_1,v_2) p(y|x,s_1,s_2)$, where $x = f(u,s_1,v_2)$ (i.e., $p(x|u,s_1,v_2)$ can get the values $0$ or $1$).\\*

{\it Codebook generation and random binning}
\begin{enumerate}
\item Generate a codebook $\mathcal{C}_v$ of $2^{n(I(V_2;S_2))+2\epsilon}$ sequences $V_2^n$ independently using $\iid \sim p(v_2)$. Label them $v_2^n(k)$, where $k \in \big\{1,2,\dots, 2^{n(I(V_2;S_2) + 2\epsilon)}\big\}$, and randomly assign each sequence $v_2^n(k)$ a bin number $b_v\big(v_2^n(k)\big)$ in the set $\big\{1,2,\dots, 2^{nR'}\big\}$.
\item Generate a codebook $\mathcal{C}_u$ of $2^{n(I(U;Y,S_2,V_2) - 2\epsilon)}$ sequences $U^n$ independently using $\iid \sim p(u)$. Label them $u^n(l)$, $l \in \big\{1,2,\dots, 2^{n(I(U;Y,S_2,V_2) - 2\epsilon)}\big\}$, and randomly assign each sequence a bin number $b_u\big(u^n(l)\big)$ in the set $\big\{1,2,\dots, 2^{nR}\big\}$.
\end{enumerate}
Reveal the codebooks and the content of the bins to all encoders and decoders.\\*
{\it Encoding}
\begin{enumerate}
\item {\it State Encoder}: Given the sequence $S_2^n$, search the codebook $\mathcal{C}_v$ and identify an index $k$ such that $\big(v_2^n(k), S_2^n\big) \in \T(V_2,S_2)$. If such a $k$ is found, stop searching and send the bin number $j = b_v\big(v_2^n(k)\big)$. If no such $k$ is found, declare an error.
\item {\it Encoder}: Given the message $W$, the sequence $S_1^n$ and the index $j$, search the codebook $\mathcal{C}_v$ and identify an index $k$ such that $\big(v_2^n(k),S_1\big) \in \T(V_2,S_1)$. If no such $k$ is found or there is more than one such index, declare an error. If a unique $k$, as defined, is found, search the codebook $\mathcal{C}_u$ and identify an index $l$ such that $\big(u^n(l),S_1^n,v_2^n(k)\big) \in \T(U,S_1,V_2)$ and $b_u\big(u^n(l)\big) = W$. If a unique $l$, as defined, is found, transmit $x_i = f\big(u_i(l), S_{1,i}, v_{2,i}(k)\big)$, $i = 1,2,\dots,n$. Otherwise, if there is no such $l$ or there is more than one, declare an error.
\end{enumerate}
{\it Decoding}\\
Given the sequences $Y^n,S_2^n$ and the index $k$, search the codebook $\mathcal{C}_u$ and identify an index $l$ such that $\big(u^n(l), Y^n,S_2^n,v_2^n(k)\big) \in \T(U,Y,S_2,V_2)$. If a unique $l$, as defined, is found, declare the message $\hat W$ to be the bin index where $u^n(l)$ is located, i.e., $\hat W = b_u\big(u^n(l)\big)$. Otherwise, if no such $l$ is found or there is more than one, declare an error.
\\*
{\it Analysis of the probability of error}\\
Without loss of generality, let us assume that the message $W=1$ was sent and the indexes that correspond with the given $W=1, S_1^n, S_2^n$ are $(k=1, l=1$ and $j=1)$; i.e., $v_2^n(1)$ corresponds with $S_2^n$, $b_v\big(v_2^n(1)\big)=1$,  $u^n(1)$ is chosen according to $\big(W=1, S_1^n, v_2^n(1)\big)$ and $b_u\big(u^n(1)\big)=1$.\\*
Define the following events:
\begin{align*}
  E_1 &:= \Big\{\forall v_2^n(k) \in \mathcal{C}_v,\ \big(v_2^n(k), S_2^n\big) \notin \T(V_2,S_2)\Big\}\\
  E_2 &:= \Big\{\big(v_2^n(1),S_1^n\big) \notin \T(V_2,S_1) \Big\}\\
  E_3 &:= \Big\{\exists k' \neq 1 \mbox{ such that } b_v\big(v_2^n(k')\big) = 1 \mbox{ and } \big(v_2^n(k'), S_1^n\big) \in \T(V_2,S_1) \Big\}\\
  E_4 &:= \Big\{\forall u^n(l) \in \mathcal{C}_u \mbox{ such that } b_u\big(u^n(l)\big) = 1,\ \big(u^n(l), S_1^n, v_2^n(1)\big) \notin \T(U,S_1,V_2) \Big\}\\
  E_5 &:= \Big\{\big(u^n(1), Y^n, S_2^n, v_2^n(1)\big) \notin \T(U,Y,S_2,V_2)\Big\}\\
  E_6 &:= \Big\{\exists l' \neq 1 \mbox{ such that } \big(u^n(l'), Y^n, S_2^n, v_2^n(1)\big) \in \T(U,Y,S_2,V_2) \Big\}\\
\end{align*}
The probability of error $P_e^{(n)}$ is upper bounded by $P_e^{n} \leq P(E_1)+P(E_2|E_1^c)+P(E_3|E_1^c,E_2^c)+P(E_4|E_1^c,E_2^c,E_3^c)+P(E_5|E_1^c,\dots,E_4^c)+P(E_6|E_1^c,\dots,E_5^c)$. Using standard arguments, and assuming that $(S_1^n,S_2^n) \in \T(S_1,S_2)$ and that $n$ is large enough, we can state that
\begin{enumerate}
\item   \begin{align}
    P(E_1) =& \Pr \big\{ \bigcap_{v_2^n(k) \in \mathcal{C}_v} \big(v_2^n(k),S_2^n\big) \notin \T(V_2,S_2) \big\}\nonumber \\
    =& \prod_{k=1}^{2^{n(I(V_2;S_2) + 2\epsilon)}} \Pr\big\{ \big(v_2^n(k), S_2^n\big) \notin \T(V_2,S_2) \big\}\nonumber\\
    =& \prod_{k=1}^{2^{n(I(V_2;S_2) + 2\epsilon)}} \Big(1-\Pr\big\{ \big(v_2^n(k), S_2^n\big) \in \T(V_2,S_2) \big\}\Big)\nonumber\\
    \leq& \Big(1-2^{-n(I(V_2;S_2)+\epsilon)}\Big)^{2^{n(I(V_2;S_2) + 2\epsilon)}} \nonumber\\
    \leq& e^{-2^{-n(I(V_2;S_2)+\epsilon)} 2^{n(I(V_2;S_2)+2\epsilon)}}\nonumber\\
  =& e^{-2^{n\epsilon}}.
  \end{align}
  The probability that there is no $v_2^n(k)$ in $\mathcal{C}_v$ such that $\big(v_2^n(k),S_2^n \big)$ is strongly jointly typical is exponentially small provided that $|\mathcal{C}_v| \geq 2^{n(I(V_2;S_2)+\epsilon)}$. This follows from the standard rate-distortion argument that $2^{nI(V_2;S_2)}$ $v_2^n$'s ``cover'' ${\cal S}_2^n$, therefore $P(E_1) \mapsto 0$.
\item
  By the Markov lemma~\cite{Berger}, since $(S_1^n,S_2^n)$ are strongly jointly typical, $\big(S_2^n, v_2^n(1)\big)$ are strongly jointly typical and the Markov chain $S_1-S_2-V_2$ holds, then $\big(S_1^n, S_2^n, v_2^n(1)\big)$ are strongly jointly typical with high probability. Therefore, $P(E_2|E_1^c) \to 0$.
\item
  \begin{align}
    P(E_3|E_1^c,E_2^c) &= \Pr \big\{ \bigcup_{\substack{ v_2^n(k'\neq 1) \in \mathcal{C}_v \\ b_v\big(v_2^n(k')\big) = 1}} \big(v_2^n(k'), S_1^n\big) \in \T(V_2,S_1)\big\}\\
    &\leq \sum_{\substack{ v_2^n(k'\neq 1) \in \mathcal{C}_v \\ b_v\big(v_2^n(k')\big) = 1}} \Pr\big\{ \big(v_2^n(k'), S_1^n\big) \in \T(V_2,S_1)\big\}\\
    & \leq \sum_{\substack{ v_2^n(k'\neq 1) \in \mathcal{C}_v \\ b_v\big(v_2^n(k')\big) = 1}}  2^{n(I(V_2;S_1)+\epsilon)}\\
    & = 2^{n(I(V_2;S_2)+2\epsilon - R')} 2^{-n(I(V_2;S_1)-\epsilon)}\\
    & = 2^{n( I(V_2;S_2) - I(V_2;S_1) + 3\epsilon - R')}.
  \end{align}
  The probability that there is another index $k'$, $k'\neq 1$, such that $v_2^n(k')$ is in bin number $1$ and that is strongly jointly typical with $S_1^n$ is bounded by the number of $v_2^n(k')$'s in the bin times the probability of joint typicality. Therefore, if the number of bins $R' > I(V_2;S_2)-I(V_2;S_1) + 3\epsilon$ then $P(E_3|E_1^c,E_2^c) \to 0$.
\item
  We  use here the same argument we used for $P(E_1)$; by the covering lemma, we can state that the probability that there is no $u^n(l)$ in bin number $1$ that is strongly jointly typical with $\big(S_1^n, v_2^n(1)\big)$ tends to zero for large enough $n$ if the average number of $u^n(l)$'s in each bin is greater than $2^{n(I(U;S_1, V_2) + \epsilon)}$;  i.e., $|\mathcal{C}_u|{/} 2^{nR} > 2^{n(I(U;S_1, V_2) + \epsilon)}$. This also implies that in order to avoid an error the number of words one should use is $R < I(U;Y,S_2,V_2) - I(U;S_1,V_2) - 3\epsilon$, where the last expression also equals $I(U;Y,S_2|V_2)-I(U;S_1|V_2)-3\epsilon$.
\item
As we argued for $P(E_2|E_1^c)$, since $\big(X^n,u^n(1),S_1^n,v_2^n(1)\big)$ is strongly jointly typical, $\big(Y^n,X^n,S_1^n,S_2^n\big)$ is strongly jointly typical and the Markov chain $(U,V_2) - (X,S_1,S_2) - Y$ holds, then, by the Markov lemma~\cite{Berger}, $\big(u^n(1), Y^n,S_2^n,v_2^n(1)\big)$ is strongly jointly typical with high probability, i.e.,  $P(E_5|E_1^c,\dots,E_4^c) \to 0$.
\item
  \begin{align}
    P(E_6|E_1^c,\dots,E_5^c) &= \Pr\big\{ \bigcup_{u^n(l'\neq 1) \in \mathcal{C}_u} \big(u^n(l'), Y^n, S_2^n,v_2^n(1)\big) \in \T(U,Y,S_2,V_2) \big\} \nonumber\\
  &\leq \sum_{l' = 2}^{2^{n(I(U;Y,S_2,V_2) + 2\epsilon)}} \Pr \big\{ \big(u^n(l'), Y^n,S_2^n, V_2^n\big) \in \T(U,Y,S_2,V_2)\big\} \nonumber\\
  & \leq \sum_{l' = 2}^{2^{n(I(U;Y,S_2,V_2) + 2\epsilon)}} 2^{-n(I(U;Y,S_2,V_2)-\epsilon)} \nonumber\\
  & \leq 2^{n(I(U;Y,S_2,V_2)-2\epsilon)}2^{-n(I(U;Y,S_2,V_2)-\epsilon)} \nonumber\\
  & = 2^{-n\epsilon}.
  \end{align}
  The probability that there is another index $l'$, $l' \neq 1$, such that $u^n(l')$ is strongly jointly typical with $\big(Y^n,S_2^n,v_2^n(1)\big)$ is bounded by the total number of $u^n$'s times the probability of joint typicality. Therefore, taking $|\mathcal{C}_u| < 2^{n(I(U;Y,S_2,V_2) - \epsilon)}$ assures us that $P(E_6|E_1^c,\dots,E_5^c) \to 0$.  This follows the standard channel capacity argument that one can distinguish at most $2^{nI(U;Y,S_2,V_2)}$ different $u^n(l)$'s given any typical member of ${\cal Y}^n \times {\cal S}_2^n \times {\cal V}_2^n$.
\end{enumerate}

This shows that for rates $R$ and $R'$ as described and for large enough $n$, the error events are of arbitrarily small probability. This concludes the proof of the achievability and the lower bound on the capacity of Case~2.\\
{\mbox{\quad}\bf Converse:} ({\it Channel capacity Case 2 - Upper bound}). We first prove that it is possible to bound the capacity from above by using two random variables, $U$ and $V$, that maintain the Markov chain $U - (S_1,V_2) - S_2$ (that is $C_2^{ub1}$). Then, we prove that it is also possible to upper-bound the capacity by using $U$ and $V$ that maintain the Markov relation $V_2 - S_2 - S_1$ (that is $C_2^{ub2}$).

Fix the rates $R$ and $R'$ and a sequence of codes $(2^{nR}, 2^{nR'}, n)$ that achieve the capacity. By Fano's inequality, $H(W|Y^n,S_2^n) \leq n\epsilon_n$, where $\epsilon_n \rightarrow 0$ as $n \rightarrow \infty$. Let $T_2=f_v(S_2^n)$, and define $V_{2,i} = (T_2,Y^{i-1}, S_{1,i+1}^n, S_2^{i-1}),\ U_i = W$; hence, the Markov chain $U_i - (S_{1,i}, V_{2,i}) - S_{2,i}$ is maintained. The proof for this follows.
\begin{align}
  \label{eq:13}
  p(u_i|s_{1,i},v_{2,i},s_{2,i}) = & p(w|s_{1,i},t_2, y^{i-1}, s_{1,i+1}^n, s_2^{i-1}, s_{2,i}) \nonumber\\
  = & \sum_{x^{i-1},s_1^{i-1}} p(w,x^{i-1},s_1^{i-1} | s_{1,i},t_2, y^{i-1}, s_{1,i+1}^n, s_2^{i-1}, s_{2,i}) \nonumber\\
  \stackrel{(a)}{=} & \sum_{x^{i-1},s_1^{i-1}} p(s_1^{i-1} | t_2, y^{i-1}, s_{1,i}^n, s_2^{i-1}) p(x^{i-1} | t_2, y^{i-1}, s_1^n, s_2^{i-1}) p(w| x^{i-1},t_2,y^{i-1},s_1^n,s_2^{i-1}) \nonumber\\
  = & p(w|t_2,y^{i-1},s_{1,i+1}^n,s_2^{i-1},s_{1,i}).
\end{align}
Next, consider
%
\begin{align}
  nR' \geq& H(T_2) \nonumber \\
  \geq& H(T_2|S_1^n) - H(T_2|S_1^n, S_2^n) \nonumber \\
  =& I(T_2; S_2^n| S_1^n) \nonumber \\
  =& H(S_2^n|S_1^n)-H(S_2^n|T_2,S_1^n) \nonumber \\
  =& \sum_{i=1}^n \Big[H(S_{2,i}|S_1^n,S_2^{i-1})-H(S_{2,i}|T_2,S_1^n,S_2^{i-1})\Big] \nonumber \\
  \stackrel{(a)}{=}& \sum_{i=1}^n \Big[H(S_{2,i}|S_{1,i})  -H(S_{2,i}|T_2,S_1^n,S_2^{i-1},Y^{i-1})\Big] \nonumber \\
  \stackrel{(b)}{=}& \sum_{i=1}^n \Big[H(S_{2,i}|S_{1,i})  -H(S_{2,i}|T_2,S_{1,i+1}^n,S_2^{i-1},Y^{i-1}, S_{1,i})\Big] \nonumber \\
  =& \sum_{i=1}^n \Big[H(S_{2,i}|S_{1,i})-H(S_{2,i}|V_{2,i}, S_{1,i})\Big] \nonumber \\
  =& \sum_{i=1}^n I(S_{2,i} ; V_{2,i}|S_{1,i}),
\end{align}
where (a) follows from the fact that $S_{2,i}$ is independent of $(S_1^{i-1}, S_{1,i+1}^n, S_2^{i-1})$ given $S_{1,i}$, and the fact that $Y^{i-1}$ is independent of $S_{2,i}$ given $(T_2, S_1^n, S_2^{i-1})$ (the proof for this follows) and (b) follows from the fact that conditioning reduces entropy.
\begin{align}
  p(y^{i-1}|t_2, s_1^n,s_2^{i-1}, s_{2,i}) = & \sum_{x^n, w} p(y^{i-1}, x^n,w|t_2,s_1^n,s_2^{i-1}, s_{2,i}) \nonumber\\
  =& \sum_{x^n, w} p(w) p(x^n|w,t_2,s_1^n) p(y^{i-1}| x^{i-1},s_1^{i-1},s_2^{i-1}) \nonumber\\
  =& p(y^{i-1}|t_2, s_1^n,s_2^{i-1}),
\end{align}
where we used the facts that $W$ is independent of $(T_2,S_1^n, S_{2,i}^n)$, $X^n$ is a function of $(W,T_2,S_1^n)$ and that the channel is memoryless; i.e., $Y^{i-1}$ is independent of $(W,T_2,S_{1,i}^n,S_{2,i}^n)$ given $(X^{i-1}, S_1^{i-1}, S_2^{i-1})$. We continue the proof of the converse by considering the following set of inequalities:
\begin{align}
  nR =& H(W) \nonumber \\
  \leq& H(W|T_2) - H(W|T_2,Y^n,S_2^n) +n\epsilon_n \nonumber \\
  =& I(W;Y^n,S_2^n|T_2) +n\epsilon_n \nonumber \\
  =& \sum_{i=1}^n I(W;Y_i,S_{2,i}|T_2, Y^{i-1}, S_2^{i-1}) +n\epsilon_n\nonumber \\
  \stackrel{(b)}{=}& \sum_{i=1}^n \Big[ I(W,S_{1,i+1}^n; Y_i, S_{2,i}| T_2,Y^{i-1}, S_2^{i-1}) \nonumber \\
  & \qquad -I(S_{1,i+1}^n ; Y_i,S_{2,i}|W,T_2,Y^{i-1},S_2^{i-1})\Big] +n\epsilon_n\nonumber \\
  \stackrel{(c)}{=}& \sum_{i=1}^n \Big[I(W,S_{1,i+1}^n; Y_i, S_{2,i}|T_2, Y^{i-1}, S_2^{i-1}) \nonumber \\
  & \qquad -I(S_{1,i} ; Y^{i-1}, S_2^{i-1}|W,T_2,S_{1,i+1}^n)\Big] +n\epsilon_n\nonumber \\
  =& \sum_{i=1}^n \Big[ I(W;Y_i, S_{2,i}|T_2,Y^{i-1}, S_{1,i+1}^n, S_2^{i-1}) \nonumber \\
  & \qquad -I(S_{1,i};W|T_2, Y^{i-1}, S_{1,i+1}^n, S_2^{i-1})\Big]\nonumber \\
  & \qquad+ \Delta - \Delta^* +n\epsilon_n, \label{eq:3}
\end{align}
where
\begin{align}
  \Delta =& \sum_{i=1}^nI(S_{1,i+1}^n; Y_i, S_{2,i} | T_2, Y^{i-1}, S_2^{i-1}), \label{eq:1} \\
  \Delta^* =& \sum_{i=1}^nI(S_{1,i} ; Y^{i-1}, S_2^{i-1} | T_2, S_{1,i+1}^n), \label{eq:2}
\end{align}
(b) follows from the mutual information properties and (c) follows from the Csisz\'ar sum identity.\\*
By using the Csisz\'ar sum on (\ref{eq:1}) and (\ref{eq:2}), we get
\begin{align}
  \Delta = \Delta^*,
\end{align}
and, therefore, from  (\ref{eq:4}) and (\ref{eq:3})
\begin{align}
  R' \geq & \frac{1}{n} \sum_{i=1}^n I(S_{2,i} ; V_{2,i}|S_{1,i})\\
  R - \epsilon_n\leq& \frac{1}{n} \sum_{i=1}^n \Big[ I(U_i;Y_i, S_{2,i}|V_{2,i}) -I(U_i;S_{1,i}|V_{2,i})\Big].
\end{align}
Using the convexity of $R'$ and Jansen's inequality, the standard time sharing argument for $R$ and the fact that $\epsilon_n \to 0$ as $n \to \infty$, we can conclude that
\begin{align}
  R' \geq& I(V_2;S_2|S_1),\\
  R \leq& I(U;Y,S_2|V_2)- I(U;S_1| V_2),
\end{align}
where $U$ and $V$ maintain the Markov chain $U - (S_1,V_2) - S_2$.

We now proceed to prove that it is possible to upper-bound the capacity of Case 2 by using two random variables, $U$ and $V$, that maintain the Markov chain $V_2 - S_2 - S_1$.
Fix the rates $R$ and $R'$ and a sequence of codes $(2^{nR}, 2^{nR'}, n)$ that achieve the capacity. By Fano's inequality, $H(W|Y^n,S_2^n) \leq n\epsilon_n$, where $\epsilon_n \rightarrow 0$ as $n \rightarrow \infty$. Let $T_2=f_v(S_2^n)$ and define $V_{2,i} = (T_2,S_2^{i-1}),\ U_i = (W,Y^{i-1},S_{1,i+1}^n)$. The Markov chain $V_{2,i} - S_{2,i} - S_{1,i}$ is maintained. Then,
%
\begin{align}
  nR' \geq& H(T_2) \nonumber \\
  \geq& H(T_2|S_1^n) - H(T_2|S_1^n, S_2^n) \nonumber \\
  =& I(T_2; S_2^n| S_1^n) \nonumber \\
  =& H(S_2^n|S_1^n)-H(S_2^n|T_2,S_1^n) \nonumber \\
  =& \sum_{i=1}^n \Big[H(S_{2,i}|S_1^n,S_2^{i-1})-H(S_{2,i}|T_2,S_1^n,S_2^{i-1})\Big] \nonumber \\
  \stackrel{(a)}{=}& \sum_{i=1}^n \Big[H(S_{2,i}|S_{1,i})  -H(S_{2,i}|T_2,S_{1,i},S_{1,i+1}^n,S_2^{i-1})\Big] \nonumber \\
  \geq&\sum_{i=1}^n \Big[H(S_{2,i}|S_{1,i})  -H(S_{2,i}|T_2,S_{1,i},S_2^{i-1})\Big] \nonumber \\
  =& \sum_{i=1}^n \Big[H(S_{2,i}|S_{1,i})-H(S_{2,i}|V_{2,i}, S_{1,i})\Big] \nonumber \\
  =& \sum_{i=1}^n I(S_{2,i} ; V_{2,i}|S_{1,i}), \label{eq:4}
\end{align}
where (a) follows from the fact that $S_{2,i}$ is independent of $(S_1^{i-1}, S_{1,i+1}^n, S_2^{i-1})$ given $S_{1,i}$, and the fact that $(Y^{i-1}, S_1^{i-1})$ is independent of $S_{2,i}$ given $(T_2, S_{1,i}^n, S_2^{i-1})$; the proof for this follows.
\begin{align}
  \label{eq:34}
  p(y^{i-1},s_1^{i-1}|t_2, s_{1,i}^n,s_2^{i-1}, s_{2,i}) = & \sum_{x^n, w} p(y^{i-1},s_1^{i-1}, x^n,w|t_2,s_{1,i}^n,s_2^{i-1}, s_{2,i}) \nonumber\\
  =& \sum_{x^n, w} p(w) p(s_1^{i-1}|s_2^{i-1}) p(x^n|w,t_2,s_1^n) p(y^{i-1}| x^{i-1},s_1^{i-1},s_2^{i-1}) \nonumber\\
  =& p(y^{i-1},s_1^{i-1}|t_2, s_{1,i}^n,s_2^{i-1}),
\end{align}
where we used the facts that $W$ is independent of $(T_2,S_{1,i}^n, S_{2,i}^n)$, $S_1^{i-1}$ is independent of $(T_2,S_{1,i}^n,S_{2,i}^n)$ given $S_2^{i-1}$, $X^n$ is a function of $(W,T_2,S_1^n)$ and that the channel is memoryless; i.e., $Y^{i-1}$ is independent of $(W,T_2,S_{1,i}^n,S_{2,i}^n)$ given $(X^{i-1}, S_1^{i-1}, S_2^{i-1})$.

In order to complete our proof, we need the following lemma.
\begin{lemma} The following inequality holds:
  \begin{align}
    \label{eq:33}
    \sum_{i=1}^n I(S_{1,i} ; W, Y^{i-1},S_{1,i+1}^n |T_2,S_2^{i-1}) \leq \sum_{i=1}^n I(S_{1,i} ; W, Y^{i-1}, S_2^{i-1}|T_2,S_{1,i+1}^n).
  \end{align}
  \begin{proof}
Notice that
    \begin{align}
      \label{eq:36}
     \sum_{i=1}^n  I(S_{1,i} ; W, Y^{i-1},S_{1,i+1}^n |T_2,S_2^{i-1}) = \sum_{i=1}^n I(S_{1,i} ; W, Y^{i-1},S_{1,i+1}^n,S_2^{i-1} |T_2) - I(S_{1,i} ; S_2^{i-1}|T_2)
    \end{align}
and that
\begin{align}
  \label{eq:37}
  \sum_{i=1}^n I(S_{1,i} ; W, Y^{i-1}, S_2^{i-1}|T_2,S_{1,i+1}^n) = \sum_{i=1}^n I(S_{1,i} ; W, Y^{i-1},S_{1,i+1}^n, S_2^{i-1}|T_2) - I(S_{1,i} ; S_{1,i+1}^n|T_2).
\end{align}
Therefore, it is enough to show that $ \sum_{i=1}^n - I(S_{1,i} ; S_2^{i-1}|T_2) \leq \sum_{i=1}^n - I(S_{1,i} ; S_{1,i+1}^n|T_2)$ holds in order to prove the lemma. Therefore, consider
\begin{align}
  \label{eq:38}
  \sum_{i=1}^n - I(S_{1,i} ; S_{1,i+1}^n|T_2) - \big(\sum_{i=1}^n - I(S_{1,i} ; S_2^{i-1}|T_2)\big) &  = \sum_{i=1}^n H(S_{1,i}|T_2,S_{1,i+1}^n) - H(S_{1,i}|T_2,S_2^{i-1}) \nonumber\\
  & = \sum_{i=1}^n H(S_1^n|T_2) - H(S_{1,i}|T_2,S_2^{i-1}) \nonumber\\
  & = \sum_{i=1}^n H(S_{1,i}|T_2,S_1^{i-1}) - H(S_{1,i}|T_2,S_2^{i-1}) \nonumber\\
  & \stackrel{(a)}{\geq} 0,
\end{align}
where (a) follows from the fact that the Markov chain $S_{1,i} - (T_2,S_2^{i-1}) - (T_2,S_1^{i-1})$ holds and from the data processing inequality. This completes the proof of the lemma.
  \end{proof}
\end{lemma}

We continue the proof of the converse by considering the following set of inequalities:
\begin{align}
  nR =& H(W) \nonumber \\
  \leq& H(W|T_2) - H(W|T_2,Y^n,S_2^n) +n\epsilon_n \nonumber \\
  =& I(W;Y^n,S_2^n|T_2) +n\epsilon_n \nonumber \\
  =& \sum_{i=1}^n I(W;Y_i,S_{2,i}|T_2, Y^{i-1}, S_2^{i-1}) +n\epsilon_n\nonumber \\
  \stackrel{(a)}{=}& \sum_{i=1}^n \Big[ I(W,S_{1,i+1}^n; Y_i, S_{2,i}| T_2,Y^{i-1}, S_2^{i-1}) \nonumber \\
  & \qquad -I(S_{1,i+1}^n ; Y_i,S_{2,i}|W,T_2,Y^{i-1},S_2^{i-1})\Big] +n\epsilon_n\nonumber \\
  \stackrel{(b)}{=}& \sum_{i=1}^n \Big[I(W,S_{1,i+1}^n; Y_i, S_{2,i}|T_2, Y^{i-1}, S_2^{i-1}) \nonumber \\
  & \qquad -I(S_{1,i} ; Y^{i-1}, S_2^{i-1}|W,T_2,S_{1,i+1}^n)\Big] +n\epsilon_n\nonumber \\
  & = \sum_{i=1}^n \Big[I(W,S_{1,i+1}^n; Y_i, S_{2,i}|T_2, Y^{i-1}, S_2^{i-1}) \nonumber \\
  & \qquad -I(S_{1,i} ; W, Y^{i-1}, S_2^{i-1}|T_2,S_{1,i+1}^n)\Big] +n\epsilon_n\nonumber \\
  & \stackrel{(c)}{\leq} \sum_{i=1}^n \Big[I(W,S_{1,i+1}^n; Y_i, S_{2,i}|T_2, Y^{i-1}, S_2^{i-1}) \nonumber \\
  & \qquad -I(S_{1,i} ; W, Y^{i-1},S_{1,i+1}^n |T_2,S_2^{i-1})\Big] +n\epsilon_n\nonumber \\
  =& \sum_{i=1}^n I(U_i ; Y_i, S_{1,i+1}^n | T_2,S_2^{i-1}) - I(U_i;S_{1,i}|V_{2,i}),
\end{align}
where (a) follows from the mutual information properties, (b) follows from the Csisz\'ar sum identity and (c) follows from Lemma 3. Therefore,
\begin{align}
  R' \geq & \frac{1}{n} \sum_{i=1}^n I(S_{2,i} ; V_{2,i}|S_{1,i})\\
  R - \epsilon_n\leq& \frac{1}{n} \sum_{i=1}^n \Big[ I(U_i;Y_i, S_{2,i}|V_{2,i}) -I(U_i;S_{1,i}|V_{2,i})\Big].
\end{align}
Using the convexity of $R'$ and Jansen's inequality, the standard time sharing argument for $R$ and the fact that $\epsilon_n \to 0$ as $n \to \infty$, we can conclude that
\begin{align}
  R' \geq& I(V_2;S_2|S_1),\\
  R \leq& I(U;Y,S_2|V_2)- I(U;S_1| V_2),
\end{align}
where the Markov chain $V_2 - S_2 - S_1$ holds.
Therefore, we can conclude that the expression given in~(\ref{eq:CC_case2}) is an upper-bound to any achievable rate. This concludes the proof of the upper-bound and the proof of Theorem \ref{theorem:CC} Case 2.

\subsection{Proof of Theorem \ref{theorem:CC}, Case 2$_C$}
\label{sec:proof:CC:causal}
For describing the DSI, $S_2$, with a rate $R'$ we use the standard  rate-distortion coding scheme. Then, for the channel coding we use the Shannon strategy~\cite{1662378} coding scheme where the channel's causal state information at the encoder is $S_1$, $S_2$ is a part of the channel's output and the rate-limited description of $S_2$ is the side information at both the encoder and the decoder.

\begin{figure}[h!]
  \centering
  \small
  \psfrag{A}{\ \ Encoder}
    \psfrag{B}{\ \ Channel}
    \psfrag{D}{\ \ Decoder}
    \psfrag{W}{$W$}
    \psfrag{w}{$\hat{W}$}
    \psfrag{X}{$X^n$}
    \psfrag{x}{$X_i$}
    \psfrag{Y}{$Y^n$}
    \psfrag{y}{$Y_i$}
    \psfrag{T}{$S_1^n$}
    \psfrag{t}{$S_{1,i}$}
    \psfrag{R}{$S_2^n$}
    \psfrag{r}{$S_{2,i}$}
    \psfrag{J}{$R'$}
    \psfrag{K}{$R'$}
    \includegraphics[scale = 0.5]{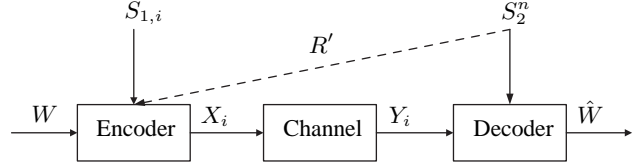}
  \caption{Channel capacity: Case 2 with causal ESI. $C_{2C} = \max I(U;Y,S_2|V_2)$, where the maximization is over all PMFs $p(v_2|s_2) p(u|v_2) p(x|u,s_1,v_2)$ such that $R' \geq I(V_2;S_2)$.}
  \label{fig:channel2c}
\end{figure}
\mbox{\quad}{\bf Achievability:} ({\it Channel capacity Case~2$_C$}). Given $(S_{1,i},S_{2,i})\sim$ i.i.d. $p(s_1,s_2)$, where the ESI is known in a causal way ($S_1^i$ at time $i$), and the memoryless channel $p(y|x,s_1,s_2)$, fix $p(s_1, s_2, v_2, u, x, y) = p(s_1,s_2) p(v_2|s_2) p(u|v_2) p(x|u,s_1,v_2) p(y|x,s_1,s_2)$, where $x = f(u,s_1,v_2)$ (i.e., $p(x|u,s_1,v_2)$ can get the values $0$ or $1$).\\*

{\it Codebook generation and random binning}
\begin{enumerate}
\item Generate a codebook $\mathcal{C}_v$ of $2^{n\big(I(V_2;S_2) + 2\epsilon\big)}$ sequences $V_2^n$ independently using $\iid \sim p(v_2)$. Label them $v_2^n(k)$ where $k \in \big\{1,2,\dots, 2^{n(I(V_2;S_2) + 2\epsilon)}\big\}$.
\item For each $v_2^n(k)$ generate a codebook $\mathcal{C}_u(k)$ of $2^{n\big(I(U;Y,S_2|V_2)-2\epsilon\big)}$ sequences $U^n$ distributed independently according to $\iid \sim p(u|v_2)$. Label them $u^n(w,k)$, where $w\in \big\{1,2,\dots, 2^{n(I(U;Y,S_2|V_2) - 2\epsilon)}\big\}$, and associate the sequences $u^n(w,\cdot)$ with the message $W=w$.
\end{enumerate}
Reveal the codebooks and the content of the bins to all encoders and decoders.\\*
{\it Encoding}
\begin{enumerate}
\item {\it State Encoder}: Given the sequence $S_2^n$, search the codebook $\mathcal{C}_v$ and identify an index $k$ such that $\big(v_2^n(k), S_2^n\big) \in \T(V_2,S_2)$. If such a $k$ is found, stop searching and send it. Otherwise, if no such $k$ is found, declare an error.
\item {\it Encoder}: Given the message $W \in \big\{1,2,\dots,2^{n(I(U;Y,S_2|V_2)-2\epsilon)}\big\}$, the index $k$ and $S_1^i$ at time $i$, identify $u^n(W,k)$ in the codebook $\mathcal{C}_u(k)$ and transmit $x_i = f\big(u_i(W,k), S_{1,i}, v_{2,i}(k)\big)$ at any time $i \in \{1,2,\dots,n\}$. The element $x_i$ is the result of a multiplexer with an input signal $\big(u_i(W,k) , v_{2,i}(k)\big)$ and a control signal  $S_{1,i}$.
\end{enumerate}
{\it Decoding}\\
Given $Y^n, S_2^n$ and $k$, look for a unique index $\hat W$, associated with the sequence $u^n(\hat{W},k) \in {\cal C}_u(k)$, such that $\big(Y^n, S_2^n, u^n(\hat{W},k)\big) \in \T(Y,U,S_2|v_2^n(k))$. If a unique such $\hat W$ is found, declare that the sent message was $\hat W$. Otherwise, if no unique index $\hat W$ exists, declare an error.\\*
\\*
{\it Analysis of the probability of error}\\
Without loss of generality, let us assume that the message $W=1$ was sent and the index $k$ that correspond with $S_2^n$ is $k=1$;  i.e., $v_2^n(1)$ corresponds to $S_2^n$ and $u^n(1,1)$ is chosen according to $\big(W=1, v_2^n(1)\big)$.\\*
Define the following events:
\begin{align*}
  E_1 &:= \Big\{\forall v_2^n(k)\in {\cal C}_v,\ \big(S_2^n, v_2^n(k) \big) \notin \T(S_2,V_2) \Big\}\\
  E_2 &:= \Big\{(u^n(1,1), Y^n,S_2^n) \notin \T(U,Y,S_2|v_2^n(1)) \Big\}\\
  E_3 &:= \Big\{\exists w' \neq 1:\ u^n(w',1) \in {\cal C}_u(1) {\rm\ and}\  \big(u^n(w',1),Y^n ,S_2^n\big) \in \T(U,Y,S_2|v_2^n(1)) \Big\}.
\end{align*}
The probability of error $P_e^{(n)}$ is upper bounded by $P_e^{n} \leq P(E_1)+P(E_2|E_1^c) + P(E_3|E_1^c, E_2^c)$. Using standard arguments and assuming that $(S_1^n,S_2^n) \in \T(S_1,S_2)$ and that $n$ is large enough, we can state that
\begin{enumerate}
\item For each sequence $v_2^n \in \mathcal{C}_v$, the probability that $v_2^n$ is not jointly typical with $S_2^n$ is at most $\big(1-2^{-n(I(V_2;S_2)+\epsilon)}\big)$. Therefore, having $2^{n(I(V_2;S_2)+2\epsilon)}$ $\iid$ sequences in $\mathcal{C}_v$, the probability that none of those sequences is jointly typical with $S_2^n$ is bounded by
    \begin{align}
      P(E_1) \leq & 2^{n(I(V_2;S_2)+2\epsilon)} \big(1-2^{-n(I(V_2;S_2)+\epsilon)}\big) \nonumber\\
        & \leq e^{-2^{n(I(V_2;S_2)+2\epsilon)} 2^{-n(I(V_2;S_2)+\epsilon)}} \nonumber\\
        & = e^{-2^{n\epsilon}},
      \end{align}
where, for every $\epsilon>0$, the last line goes to zero as $n$ goes to infinity.
\item The random variable $Y^n$ is distributed according to $p(y|x,s_1,s_2) = p(y|x,s_1,s_2,v_2)$, therefore, having $(S_2^n,v_2^n(1)) \in \T(S_2,V_2)$ implies that $(Y^n,S_2^n,v_2^n(1)) \in \T(Y,S_2,V_2)$. Recall that  $x_i = f\big(u_i(1,1),S_{1,i},v_2(1)\big)$ and that $U^n$ is generated according to $p(u|v_2)$; therefore, $(X^n,S_1^n, u^n(1,1),v_2^n(1))$ is jointly typical. Thus, by the Markov~lemma~\cite{Berger}, we can state that $(Y^n,S_2^n,u^n(1,1),v_2^n(1)) \in \T(Y,S_2,U,V_2)$ with high probability for a large enough $n$.
\item Now, the probability for a random $U^n$, such that $(U^n,v_2^n(1)) \in \T(U,V_2)$, to be also jointly typical with $(Y^n,S_2^n,v_2^n(1))$ is upper bounded by $2^{-n(I(U,Y,S_2|V_2) - \epsilon)}$, hence
  \begin{align}
    \label{eq:30}
    P(E_3|E_1^c,E_2^c) \leq & \sum_{1<w'}^{|\mathcal{C}_u(1)|} \Pr\big\{\big(u^n(w',1), Y^n,S_2^n\big) \in \T(U,Y,S_2|v_2^n(1)) \Big\} \nonumber\\
    \leq & \sum_{1<w'}^{|\mathcal{C}_u(1)|} 2^{-n(I(U,Y,S_2|V_2) - \epsilon)} \nonumber\\
    \leq & 2^{n(I(U,Y,S_2|V_2) - 2\epsilon)} 2^{-n(I(U,Y,S_2|V_2) - \epsilon)} \nonumber\\
    = & 2^{-n\epsilon},
  \end{align}
  which goes to zero exponentially fast with $n$ for every $\epsilon >0$.\\
  Therefore, $P^{(n)}_{\epsilon} = P(\hat W \neq W)$ goes to zero as $n \to \infty$.
\end{enumerate}
{\mbox{\quad}\bf Converse:} ({\it Channel capacity case 2$_c$}). Fix the rates $R$ and $R'$ and a sequence of codes $(2^{nR}, 2^{nR'}, n)$ that achieve capacity. By Fano's inequality, $H(W|Y^n,S_2^n) \leq n\epsilon_n$, where $\epsilon_n \rightarrow 0$ as $n \rightarrow \infty$. Let $T_2=f_v(S_2^n)$, and define $V_{2,i} = (T_2, Y^{i-1}, S_2^{i-1}),\ U_i = W$. Then,
%
\begin{align}
  nR' \geq& H(T_2) \nonumber \\
  \geq& H(T_2) - H(T_2 | S_2^n) \nonumber \\
  =& I(T_2; S_2^n) \nonumber \\
  =& H(S_2^n)-H(S_2^n|T_2) \nonumber \\
  =& \sum_{i=1}^n \Big[H(S_{2,i}|S_2^{i-1})-H(S_{2,i}|T_2,S_2^{i-1})\Big] \nonumber \\
  \stackrel{(a)}{=}& \sum_{i=1}^n \Big[H(S_{2,i}) - H(S_{2,i}|T_2,S_2^{i-1},Y^{i-1})\Big] \nonumber \\
  =& \sum_{i=1}^n I(S_{2,i} ; T_2, Y^{i-1}, S_2^{i-1}) \nonumber \\
  =& \sum_{i=1}^n I(S_{2,i} ; V_{2,i}), \label{eq:11}
\end{align}

where (a) follows from the fact that $S_{2,i}$ is independent of $S_2^{i-1}$ and the fact that  $S_{2,i}$ is independent of $Y^{i-1}$ given $(T_2, S_2^{i-1})$. The proof for this follows.
\begin{align}
  \label{eq:35}
  p(y^{i-1}|t_2,s_2^{i-1}, s_{2,i}) = & \sum_{w,x^{i-1}, s_1^{i-1}} p(y^{i-1}, w,x^{i-1}, s_1^{i-1}|t_2,s_2^{i-1}, s_{2,i}) \nonumber\\
  =& \sum_{w,x^{i-1},s_1^{i-1}} p(w) p(s_1^{i-1}|s_2^{i-1})p(x^{i-1}|w,t_2,s_1^{i-1}) p(y^{i-1}|x^{i-1},s_1^{i-1}, s_2^{i-1}) \nonumber\\
  =& p(y^{i-1}|t_2, s_2^{i-1}),
\end{align}
where we used the fact that $W$ is independent of $(T_2,S_2^{i-1}, S_{2,i})$, $S_1^{i-1}$ is independent of $(T_2,S_{2,i})$ given $S_2^{i-1}$, $X^{i-1}$ is a function of $(W,T_2,S_1^{i-1})$ and that $Y^{i-1}$ is independent of $(W,T_2,S_{2,i})$ given $(X^{i-1},S_1^{i-1},S_2^{i-1})$. We now continue with the proof of the converse.
\begin{align}
  nR \leq& H(W) \nonumber \\
  \leq& H(W|T_2) - H(W|T_2,Y^n,S_2^n) +n\epsilon_n \nonumber \\
  =& I(W;Y^n,S_2^n|T_2) +n\epsilon_n \nonumber \\
  =& \sum_{i=1}^n I(W;Y_i,S_{2,i}|T_2, Y^{i-1}, S_2^{i-1}) +n\epsilon_n\nonumber \\
  =& \sum_{i=1}^n I(U_i ; Y_i, S_{2,i} | V_{2,i}) + n\epsilon_n \label{eq:12}
\end{align}
and therefore, from  (\ref{eq:11}) and (\ref{eq:12})
\begin{align}
  R' \geq & \frac{1}{n} \sum_{i=1}^n I(S_{2,i} ; V_{2,i})\\
  R - \epsilon_n\leq& \frac{1}{n} \sum_{i=1}^n  I(U_i;Y_i, S_{2,i}|V_{2,i}).
\end{align}
Using the convexity of $R'$ and Jansen's inequality, the standard time-sharing argument for $R$ and the fact that $\epsilon_n \to 0$ as $n \to \infty$, we can conclude that
\begin{align}
  R' \geq& I(V_2;S_2),\\
  R \leq& I(U;Y,S_2|V_2).
\end{align}

Notice that the Markov chain $V_{2,i} - S_{2,i} - S_{1,i}$ holds since $(Y^{i-1},S_2^{i-1})$ is independent of $S_{1,i}$ and $T_2(S_2^n)$ is dependent on $S_{1,i}$ only through $S_{2,i}$. Notice also that the Markov chain $U_i - V_{2,i} - (S_{1,i},S_{2,i})$ holds since

\begin{align}
  \label{eq:10}
  p(w|t_2,y^{i-1}, s_2^{i-1}, s_{1,i}, s_{2,i}) =& \sum_{x^{i-1},s_1^{i-1}} p(w,x^{i-1},s_1^{i-1} | t_2,y^{i-1}, s_2^{i-1}, s_{1,i}, s_{2,i}) \nonumber\\
    =& \sum_{x^{i-1},s_1^{i-1}} p(s_1^{i-1}|t_2,y^{i-1},s_2^{i-1}) p(x^{i-1} | t_2, y^{i-1}, s_1^{i-1},s_2^{i-1}) p(w|t_2,x^{i-1},s_1^{i-1}) \nonumber\\
    = & p(w|t_2,y^{i-1},s_2^{i-1}).
\end{align}
This concludes the converse, and the proof of Theorem \ref{theorem:CC} Case 2$_C$.

\section{Proof of Theorem \ref{theorem:SC}}
\label{proof:CS}

In this section we provide the proof of Theorem~\ref{theorem:SC}, Cases 1 and 1$_C$. Case 2, where the encoder is informed with increased ESI and the decoder is informed with DSI is a special case of~\cite{DBLP:journals/tit/Kaspi85} for $K=1$ and, therefore, the proof for this case is omitted. Following Kaspi's scheme (Figure~\ref{fig:kaspi}) for $K=1$, at the first stage, node $W$ sends a description of $W$ with a rate limited to $R_w$, then, after reconstructing $\hat W$ at the $Z$ node, it sends a function of $Z$ and $\hat W$ over to node $W$ with a rate limited to $R_z$. Let $S_2$ be $W$ in Kaspi's scheme and $(X,S_1)$ be $Z$ in Kaspi's scheme. Consider $D_z = d(Z_i,\hat Z_i) = d\big((X,S_{1,i}), (\hat X_i, \hat S_{1,i})\big) = d(X_i,\hat X_i) = D$. 
Then, it is apparent that Case 2 of the rate-distortion problems is a special case of Kaspi's two-way problem for $K=1$.
\begin{figure}[h!]
  \centering
  \scriptsize
  \psfrag{a}{W CODEC}
  \psfrag{b}{Z CODEC}
  \psfrag{c}{Binary data at rate $R_z$}
  \psfrag{d}{Binary data at rate $R_w$}
  \psfrag{w}{$\{W_i\}$}
  \psfrag{z}{$\{Z_i\}$}
  \psfrag{W}{$\{\hat W_i\}$}
  \psfrag{Z}{$\{\hat Z_i\}$}
  \includegraphics[scale = 1]{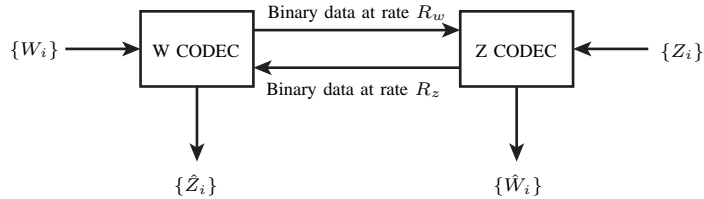}
  \caption{Kaspi's two-way source coding scheme. The total rates are $R_w = \sum_{k=1}^K R_w^k$ and $R_z = \sum_{k=1}^K R_z^k$ and the expected per-letter distortions are $D_w = \mathbb{E}\Big[\frac{1}{n} \sum_{i=1}^n d(W_i, \hat W_i) \Big]$ and $D_z = \mathbb{E}\Big[\frac{1}{n} \sum_{i=1}^n d(Z_i, \hat Z_i) \Big]$.}
  \label{fig:kaspi}
\end{figure}
\subsection{Proof of Theorem \ref{theorem:SC}, Case 1}
\label{proof:SC:case1}
We use the Wyner-Ziv coding scheme for the description of the ESI, $S_1$, at a rate $R'$, where the source is $S_1$ and the side information at the decoder is $S_2$. Then, to describe the main source, $X$, with distortion less than or equal to $D$ we use the Wyner-Ziv coding scheme again, where this time, $S_2$ is the side information at the decoder, $S_1$ is a part of the source and the rate-limited description of $S_1$ is the side information at both the encoder and the decoder.  Notice that $I(U;X,S_1|V_1) - I(U;S_2|V_1) = I(U;X,S_1,V_1) - I(U;S_1,V_1)$ and that since the Markov chain $V_1 - S_1 - S_2$ holds, it is also possible to write $R' \geq I(V_1;S_1) - I(V_1;S_2)$; we use these expressions in the following proof.

\begin{figure}[h!]
  \centering
  \small
    \psfrag{A}{\ \ Encoder}
    \psfrag{B}{\ \ Decoder}
    \psfrag{X}{$X^n$}
    \psfrag{x}{$\hat{X}^n$}
    \psfrag{s}{$\hat{X}_i$}
    \psfrag{E}{$S_1^n$}
    \psfrag{e}{$S_{1,i}$}
    \psfrag{R}{$S_2^n$}
    \psfrag{r}{$S_{2,i}$}
    \psfrag{J}{$R'$}
    \psfrag{K}{$R'$}
    \includegraphics[scale = 0.5]{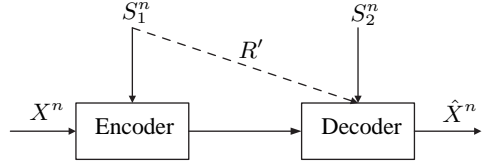}
  \caption{Rate-distortion: Case 1. $R_1(D) = \min I(U;X,S_1|V_1) - I(U;S_2|V_1)$, where the minimization is over all PMFs $p(v_1|s_1) p(u|x,s_1,v_1) p(\hat x|u,s_2,v_1)$ such that $R' \geq I(V_1;S_1|S_2)$ and $\mathbb{E}\Big[d(X,\hat X)\Big]\leq D$.}
  \label{fig:RD1}
\end{figure}
{\bf Achievability:} ({\it Rate-distortion Case 1}). Given $(X_i, S_{1,i}, S_{2,i})\ \iid \sim \ p(x,s_1,s_2)$ and the distortion measure $D$, fix $p(x, s_1, s_2, v_1, u, \hat x) = p(x,s_1,s_2) p(v_1|s_1) p(u|x,s_1,v_1) p(\hat x|u,s_2,v_1)$ that satisfies $\mathbb E \big[ d(X, \hat X) \big] = D$ and $\hat x = f(u, s_2, v_1)$.\\*
{\it Codebook generation and random binning}\\*
\begin{enumerate}
\item Generate a codebook, ${\cal C}_v$, of $2^{n\big(I(V_1;S_1)+2\epsilon\big)}$ sequences, $V_1^n$, independently using $\iid \sim p(v_1)$. Label them $v_1^n(k)$, where $k \in \big\{1,2,\dots,2^{n(I(V_1;S_1) + 2\epsilon)}\big\}$ and randomly  assign each sequence $v_1^n(k)$ a bin number $b_v\big(v_1^n(k)\big)$ in the set $\big\{1,2,\dots,2^{nR'}\big\}$.
\item Generate a codebook ${\cal C}_u$ of $2^{n\big(I(U;X,S_1,V_1)+2\epsilon\big)}$ sequences $U^n$ independently using $\iid \sim p(u)$. Label them $u^n(l)$, where $l \in \big\{1,2,\dots, 2^{n(I(U;X,S_1,V_1)+2\epsilon)}\big\}$, and randomly and assign each $u^n(l)$ a bin number $b_u\big(u^n(l)\big)$ in the set $\big\{1,2,\dots, 2^{nR}\big\}$.
\end{enumerate}
Reveal the codebooks and the content of the bins to all encoders and decoders.\\*
{\it Encoding}
\begin{enumerate}
\item {\it State Encoder}: Given the sequence $S_1^n$, search the codebook ${\cal C}_v$ and identify an index $k$ such that $\big(S_1^n, v_1^n(k)\big) \in \T(S,V_1)$. If such a $k$ is found, stop searching and send the bin number $j = b_v\big(v_1^n(k)\big)$. If no such $k$ is found, declare an error.
\item {\it Encoder}: Given the sequences $X^n$, $S_1^n$ and $v_1^n(k)$, search the codebook ${\cal C}_u$ and identify an index $l$ such that  $\big(X^n, S_1^n, v_1^n(k), u^n(l)\big) \in \T(X,S_1,V_1,U)$. If such an $l$ is found, stop searching and send the bin number $w = b_u\big(u^n(l)\big)$. If no such $l$ is found, declare an error.
\end{enumerate}

{\it Decoding}\\
Given the bins indices $w$ and $j$ and the sequence $S_2^n$, search the codebook ${\cal C}_v$ and identify an index $k$ such that $\big(S_2^n, v_1^n(k)\big) \in \T(S_2,V_1)$ and $b_v\big(v_1^n(k)\big) = j$. If no such $k$ is found or there is more than one such index, declare an error. If a unique $k$, as defined, is found, search the codebook ${\cal C}_u$ and identify an index $l$ such that $\big(S_2^n, v_1^n(k), u^n(l) \big) \in \T(S_2, V_1, U)$ and $b_u\big(u^n(l)\big) = w$. If a unique $l$, as defined, is found, declare $\hat X_i = f_i(u^n_i(l), S_{2,i}, v_{1,i}(k)),\ i = 1,2,\dots ,n$. Otherwise, if there is no such $l$ or there is more than one, declare an error.
\\*
{\it Analysis of the probability of error}\\
Without loss of generality, for the following events $E_2, E_3, E_4, E_5$ and $E_6$, assume that $v_1^n(k=1)$ and $b_v\big(v_1^n(k=1)\big)=1$ correspond to the sequences $(X^n,S_1^n,S_2^n)$ and for the events $E_5$ and $E_6$ assume that $u^n(l=1)$ and $b_u\big(u^n(l=1)\big) = 1$ correspond to the same given sequences.
Define the following events:
\begin{align*}
  E_1 &:= \Big\{\forall v_1^n(k)\in {\cal C}_v,\ \big(S_1^n, v_1^n(k) \big) \notin \T(S_1,V_1) \Big\}\\
  E_2 &:= \Big\{\big(S_1^n, v_1^n(1)\big)\in \T(S_1,V_1) \mbox{ but } \big(S_2^n, v_1^n(1) \big) \notin \T(S_2,V_1) \Big\}\\
  E_3 &:= \Big\{\exists k' \neq 1 \mbox{ such that } b_v\big(v_1^n(k')\big) = 1  \mbox{ and } \big(S_2^n, v_1^n(k') \big) \in \T(S_2,V_1) \Big\}\\
  E_4 &:= \Big\{\forall u^n(l)\in {\cal C}_u,\ \big(X^n, S_1^n, v_1^n(1), u^n(l)\big)  \notin \T(X,S_1,V_1,U\big) \Big\}\\
  E_5 &:= \Big\{\big(X^n, S_1^n, v_1^n(1), u^n(1)\big)  \in \T(X,S_1,V_1,U\big) \mbox{ but } \big(S_2^n, v_1^n(1), u^n(1) \big) \notin \T(S_2,V_1,U\big) \Big\}\\
  E_6 &:= \Big\{\exists l' \neq 1 \mbox{ such that } b_u\big(u^n(l')\big) = 1 \mbox{ and } \big(S_2^n, v_1^n(1), u^n(l')\big) \in \T(S_2,V_1,U)  \Big\}.
\end{align*}
The probability of error $P_e^{(n)}$ is upper bounded by $P_e^{n} \leq P(E_1)+P(E_2|E_1^c)+P(E_3|E_1^c,E_2^c)+P(E_4|E_1^c,E_2^c,E_3^c)+P(E_5|E_1^c,\dots,E_4^c)+P(E_6|E_1^c\dots,E_5^c)$. Using standard arguments and assuming that $(X^n, S_1^n,S_2^n) \in \T(X,S_1,S_2)$ and that $n$ is large enough, we can state that
\begin{enumerate}
\item
  \begin{align}
    P(E_1) =& \Pr \big\{\bigcap_{v_1^n(k)\in{\cal C}_v} \big(S_1^n,v_1^n(k)\big) \notin \T(S_1,V_1) \big\} \nonumber \\
    \leq& \prod_{k=1}^{2^{n\big(I(V_1;S_1)+\epsilon\big)}} \Pr\{\big(S_1^n,V_1^n(k)\big) \notin \T(S_1,V_1)\} \nonumber \\
    \leq& e^{-2^{n\big(I(V_1;S_1)+2\epsilon\big)} 2^{-nI(S_1;V_1)-n\epsilon}} \nonumber\\
    =& e^{-n\epsilon}.
  \end{align}
  The probability that there is no $v_1^n(k)$ in $\mathcal{C}_v$ such that $\big(S_1^n, v_1^n(k) \big)$ is strongly jointly typical is exponentially small provided that $|\mathcal{C}_v| > 2^{n\big(I(S_1;V_1)+\epsilon\big)}$. This follows from the standard rate-distortion argument that $2^{nI(S_1;V_1)}$ $v_1^n(k)$s ``cover'' ${\cal S}_1^n$, therefore $P(E_1) \rightarrow 0$.

\item
  By the Markov lemma, since $(S_1^n,S_2^n)$ are strongly jointly typical and $\big(S_1^n, v_1^n(1)\big)$ are strongly jointly typical and the Markov chain $V_1 - S_1 - S_2$ holds, then $\big(S_1^n, S_2^n, v_1^n(1)\big)$ are also strongly jointly typical. Thus, $P(E_2|E_1^c) \to 0$.
\item
  \begin{align}
    P(E_3) =& \Pr\big\{\bigcup_{\substack{v_1^n(k'\neq1)\\ b_v\big(v_1(k')\big) = 1}} \big(S_2^n,v_1^n(k')\big) \in \T(S_1,V_1) \big\} \nonumber\\
    \leq& \sum_{\substack{v_1^n(k'\neq1)\\ b_v\big(v_1(k')\big) = 1}} \Pr\big\{(S_1^n,v_1^n(k')\big) \in \T(S_1,V_1)\} \nonumber\\
    \leq& 2^{n\big(I(V_1;S_1)+2\epsilon-R')} 2^{-n\big(I(S_2;V_1)-\epsilon\big)}.
  \end{align}
  The probability that there is another index $k',\ k'\neq 1$, such that $v_1^n(k')$ is in bin number $1$ and that it is strongly jointly typical with $S_2^n$  is bounded by the number of $v_1^n(k)$'s in the bin times the probability of joint typicality. Therefore, if $R' > I(V_1;S_1)-I(V_1;S_2) + 3\epsilon$ then $P(E_3|E_1^c,E_2^c) \rightarrow 0$. Furthermore, using the Markov chain $V_1 - S_1 - S_2$, we can see that the inequality can be presented as $R' > I(V_1;S_1|S_2)+3\epsilon$.

\item
  We  use here the same argument we used for $P(E_1)$. By the covering lemma we can state that the probability that there is no $u^n(l)$ in ${\cal C}_u$ that is strongly jointly typical with $\big(X^n,S_1^n, v_1^n(k)\big)$ tends to $0$ as $n \rightarrow \infty$ if $R'_u > I(U; X, S_1, V_1) + \epsilon$.  Hence, $P(E_4|E_1^c,E_2^c,E_3^c) \to 0$. 

\item
  Using the same argument we used for $P(E_2|E_1^c)$, we conclude that $P(E_4|E_1^c,E_2^c,E_3^c) \to 0$.

\item
  We use here the same argument we used for $P(E_2|E_1^c)$. Since $(U, X, S_1V_1)$ are strongly jointly typical,  $(X, S_1, S_2)$ are strongly jointly typical and the Markov chain $( U, V_1) - (X,S_1) - S_2$ holds, then $(U, X, S_1, S_2, V_1)$ are also strongly jointly typical.

\item
The probability that there is another index $l',\ l' \neq 1$ such that $u^n(l')$ is in bin number $1$ and that it is strongly jointly typical with $\big( S_2^n, v_1^n(1) \big)$ is exponentially small provided that $R \geq I(U; X, S_1, V_1) - I(U; S_2, V_1)+3\epsilon = I(U; X,S_1 |V_1) - I(U; S_2 | V_1) + 3\epsilon$. Notice that $2^{n(I(U;X,S_1,V_1)-R)}$ stands for the average number of sequences $u^n(l)$'s in each bin indexed $w$ for $w \in \{1,2,\dots, 2^{nR}\}$.
\end{enumerate}

This shows that for rates $R$ and $R'$ as described, and for large enough $n$, the error events are of arbitrarily small probability. This concludes the proof of the achievability for the source coding Case~1.
\\
{\bf Converse:} ({\it Rate-distortion Case 1}).  Fix a distortion measure $D$, the rates $R'$, $R \geq R(D) = \min I(U; X, S_1| V_1) - I(U; S_2| V_1) = \min I(U;X, S_1| S_2, V_1)$ and a sequence of codes $(2^{nR}, 2^{nR'}, n)$ such that $\mathbb E\Big[ \frac{1}{n} \sum_{i=1}^n d(X_i, \hat X_i) \Big] = D$.   Let $T_1=f_v(S_1^n)$, $T = f(X^n, S_1^n,T)$ and define $V_{1,i} = (T_1, S_{1,i+1}^n, S_2^{i-1}, S_{2,i+1}^n)$ and $\ U_i = T$. Notice that $\hat X_i = \hat X_i(T, T_1, S_2^n)$ and, therefore, $\hat X_i$ is a function of $(U_i, V_{1,i}, S_{2,i})$.
\begin{align}
  n R' \geq& H(T_1) \nonumber \\
  \geq& H(T_1|S_2^n) - H(T_1|S_1^n, S_2^n) \nonumber \\
  =& I(T_1;S_1^n|S_2^n) \nonumber \\
  =& H(S_1^n| S_2^n)- H(S_1^n| T_1, S_2^n) \nonumber \\
  =& \sum_{i=1}^{n} \Big[ H(S_{1,i}| S_{1,i+1}^n, S_2^n) - H(S_{1,i}|T_1, S_{1,i+1}^n , S_2^n) \Big] \nonumber \\
  \stackrel{(a)}{=}& \sum_{i=1}^{n} \Big[ H(S_{1,i}| S_{2,i}) - H(S_{1,i}|T_1, S_{1,i+1}^n , S_2^{i-1}, S_{2,i+1}^n, S_{2,i}) \Big] \nonumber \\
  =& \sum_{i=1}^{n} \Big[ H(S_{1,i}| S_{2,i}) - H(S_{1,i}| V_{1,i}, S_{2,i}) \Big] \nonumber \\
  =& \sum_{i=1}^{n} I(S_{1,i}; V_{1,i}| S_{2,i}),
\end{align}
where $(a)$ follows from the fact that $S_{1,i}$ is independent of $(S_{1,i+1}^n, S_{2}^{i-1}, S_{2,1+i}^n)$ given $S_{2,i}$.
\begin{align}
  nR \geq& H(T) \nonumber \\
  \geq& H(T|T_1,S_2^n) - H(T|T_1, X^n, S_1^n, S_2^n) \nonumber \\
  =& I(T;X^n,S_1^n|T_1,S_2^n) \nonumber \\
  =& H(X^n, S_1^n| T_1, S_2^n) - H(X^n, S_1^n|T, T_1, S_2^n) \nonumber \\
  =& \sum_{i=1}^n \Big[ H(X_i, S_{1,i}| T_1, S_2^n, X_{i+1}^n, S_{1,i+1}^n) - H(X_i, S_{1,i}| T,T_1, S_2^n, X_{i+1}^n, S_{1,i+1}^n) \Big] \nonumber \\
  \stackrel{(b)}{=}& \sum_{i=1}^n \Big[ H(X_i, S_{1,i}| T_1, S_{1,i+1}^n, S_2^n) - H(X_i, S_{1,i}| T,T_1, S_2^n, X_{i+1}^n, S_{1,i+1}^n) \Big] \nonumber  \\
  \stackrel{(c)}{\geq}& \sum_{i=1}^n \Big[ H(X_i, S_{1,i}| T_1, S_{1,i+1}^n, S_2^n) - H(X_i, S_{1,i}|T,T_1, S_{1,i+1}^n, S_2^n) \Big]\nonumber\\
  =& \sum_{i=1}^n  I(X_i, S_{1,i} ; T| T_1, S_{1,i+1}^n, S_2^n) \nonumber \\
  =& \sum_{i=1}^n  I(X_i, S_{1,i} ; U_i|V_{1,i}, S_{2,i}) \nonumber \\
  =& \sum_{i=1}^n R \Big( \E{ d\big(X_i, \hat X_i \big)} \Big) \nonumber \\
  \stackrel{(d)}{\geq}& nR \Big( \E{ \frac{1}{n} \sum_{i=1}^n d\big(X_i, \hat X_i \big)} \Big) \nonumber \\
  =& nR(D),
\end{align}
where $(b)$ follows from the fact that $(X_i,S_{1,i})$ is independent of $X_{i+1}^n$ given $(T_1, S_{1,i+1}^n, S_2^n)$; this is because $X_{i+1}^n$ is independent of $(T_1,X^i,S_1^i)$ given $(S_{1,i+1}^n,S_{2,i+1}^n)$, $(c)$ follows from the fact that conditioning reduces entropy and $(d)$ follows from the convexity of $R(D)$ and Jensen's inequality.

Using also the convexity of $R'$ and Jensen's inequality, we can conclude that
\begin{align}
  R' \geq& I(V_1;S_1|S_2),\\
  R \geq& I(U;X,S_1|V_1, S_2).
\end{align}
 It is easy to verify that $(T_1,S_{1,i+1}^n, S_2^{i-1}, S_{2,i+1}^n) - S_{1,i} - S_{2,i}$ forms a Markov chain, since $T_1(S_1^n)$ depends on $S_{2,i}$ only through $S_{1,i}$. The structure $T - (T_1,S_{1,i+1}^n, S_2^{i-1}, S_{2,i+1}^n, X_i, S_{1,i}) - S_{2,i}$ also forms a Markov chain since $S_{2,i}$ contains no information about $(S_1^{i-1}, X^{i-1}, X_{i+1}^n)$ given $(T_1,S_{1,i}^n, S_2^{i-1}, S_{2,i+1}^n, X_i)$ and, therefore, contains no information about  $T(X^n,S_1^n,T_1)$.

This concludes the converse, and the proof of Theorem \ref{theorem:SC} Case 1.

\subsection{Proof of Theorem \ref{theorem:SC}, Case 1$_C$}
\label{proof:SC:case1c}
For describing the ESI, $S_1$, with a rate $R'$ we use the standard rate-distortion coding scheme. Then, for the main source, $X$, we use a Weissman-El Gamal \cite{DBLP:journals/tit/WeissmanG06} coding scheme where the DSI, $S_2$, is the causal side information at the decoder, $S_1$ is a part of the source and the rate-limited description of $S_1$ is the side information at both the encoder and decoder.
\begin{figure}[h!]
  \centering
  \small
    \psfrag{A}{\ \ Encoder}
    \psfrag{B}{\ \ Decoder}
    \psfrag{X}{$X^n$}
    \psfrag{x}{$\hat{X}^n$}
    \psfrag{s}{$\hat{X}_i$}
    \psfrag{E}{$S_1^n$}
    \psfrag{e}{$S_{1,i}$}
    \psfrag{R}{$S_2^n$}
    \psfrag{r}{$S_{2,i}$}
    \psfrag{J}{$R'$}
    \psfrag{K}{$R'$}
    \includegraphics[scale = 0.5]{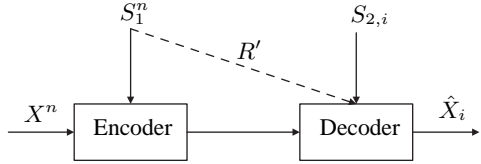}
  \caption{Rate-distortion: Case 1 with causal DSI. $R_{1C}(D) = \min I(U;X,S_1|V_1)$, where the minimization is over all PMFs $p(v_1|s_1) p(u|x,s_1,v_1)p(\hat x|u,s_2,v_1)$ such that $R' \geq I(V_1;S_1)$ and $\mathbb{E}\Big[d(X,\hat X)\Big] \leq D$.}
  \label{fig:RD1c}
\end{figure}

{\bf Achievability:} ({\it Rate-distortion Case 1$_C$}). Given $(X_i, S_{1,i}, S_{2,i}) \sim \iid \ p(x,s_1,s_2)$ where the DSI is known in a causal way ($S_2^i$ in time $i$) and the distortion measure is $D$, fix $p(x, s_1, s_2, v_1, u, \hat x) = p(x,s_1,s_2) p(v_1|s_1) p(u|x,s_1,v_1) p(\hat x|u,s_2,v_1)$ that satisfies $\mathbb E \big[ d(X, \hat X) \big] = D$ and that $\hat x = f(u, s_2, v_1)$.

{\it Codebook generation and random binning}
\begin{enumerate}
\item Generate a codebook $\mathcal{C}_v$ of $2^{n\big(I(V_1;S_1) + 2\epsilon\big)}$ sequences $V_1^n$ independently using $\iid \sim p(v_2)$. Label them $v_1^n(k)$ where $k \in \big\{1,2,\dots, 2^{n(I(V_1;S_1) + 2\epsilon)}\big\}$.
\item For each $v_1^n(k)$ generate a codebook $\mathcal{C}_u(k)$ of $2^{n\big(I(U;X,S_1|V_1)+2\epsilon\big)}$ sequences $U^n$ distributed independently according to $\iid \sim p(u|v_1)$. Label them $u^n(w,k)$, where $w\in \big\{1,2,\dots, 2^{n(I(U;X,S_1|V_1) + 2\epsilon)}\big\}$.
\end{enumerate}
Reveal the codebooks to all encoders and decoders.\\*
{\it Encoding}
\begin{enumerate}
\item {\it State Encoder}: Given the sequence $S_1^n$, search the codebook $\mathcal{C}_v$ and identify an index $k$ such that $\big(v_1^n(k), S_1^n\big) \in \T(V_1,S_1)$. If such a $k$ is found, stop searching and send it. Otherwise, if no such $k$ is found, declare an error.
\item {\it Encoder}: Given $X^n,S_1^n$ and the index $k$, search the codebook $\mathcal{C}_u(k)$ and identify an index $w$ such that $\big(u^n(w,k), X^n,S_1^n\big) \in \T(U,X,S_1|v_1^n(k))$. If such an index $w$ is found, stop searching and send it. Otherwise, declare an error.
\end{enumerate}
{\it Decoding}\\
Given the indices $w,k$ and the sequence $S_1^i$ at time $i$, declare $\hat x_i = f\big(u_i(w,k), S_{2,i}, v_{1,i}(k)\big)$.\\*
\\*
{\it Analysis of the probability of error}\\
Without loss of generality, let us assume that $v_1^n(1)$ corresponds to $S_1^n$ and that $u^n(1,1)$ corresponds to $(X^n,S_1^n,v_1^n(1))$.\\*
Define the following events:
\begin{align*}
  E_1 &:= \Big\{\forall v_1^n(k)\in {\cal C}_v,\ \big(v_1^n(k),S_1^n \big) \notin \T(S_1,V_1) \Big\}\\
  E_2 &:= \Big\{\forall u^n(w,1)\in {\cal C}_u(1),\ \big(X^n,S_1^n,u^n(w,1) \big) \notin \T(X,S_1,U) \Big\}
\end{align*}
The probability of error $P_e^{(n)}$ is upper bounded by $P_e^{n} \leq P(E_1)+P(E_2|E_1^c)$. Assuming that $(S_1^n,S_2^n) \in \T(S_1,S_2)$, we can state that by the standard rate-distortion argument, having more than $2^{n(I(V_1;S_1) + \epsilon)}$ sequences $v^n_1(k)$ in $\mathcal{C}_v$ and a large enough $n$ assures us with probability arbitrarily close to $1$ that we would find an index $k$ such that $\big(v_1^n(k), S_1^n\big) \in \T(V_1,S_1)$. Therefore, $P(E_1) \to 0$ as $n \to \infty$. Now, if $\big(v_1^n(1), S_1^n\big) \in \T(V_1,S_1)$, using the same argument, we can also state that having more than $2^{n(I(U;X,S_1|V_1) +\epsilon)}$ sequences $u^n(w,1)$ in $\mathcal{C}_u(1)$ assures us that $P(E_2|E_1^c) \to 0$ as $n \to \infty$. This concludes the proof of the achievability.

{\bf Converse:} ({\it Rate-distortion Case 1$_C$}). Fix a distortion measure $D$, the rates $R'$, $R \geq R(D) = \min I(U; X, S_1| V_1)$ and a sequence of codes $(2^{nR}, 2^{nR'}, n)$ such that $\mathbb E\Big[ \frac{1}{n} \sum_{i=1}^n d(X_i, \hat X_i) \Big] = D$.   Let $T_1=f_v(S_1^n)$, $T = f(X^n, S_1^n,T_1)$ and define $V_{1,i} = (T_1, S_{1,i+1}^n),\ U_i = T$. Notice that $\hat X_i = \hat X_i(T, T_1, S_2^i)$, and, therefore, $\hat X_i$ is a function of $(U_i,V_{1,i}, S_2^i)$.
\begin{align}
  n R' \geq& H(T_1) \nonumber \\
  \geq& H(V) - H(T_1|S_1^n) \nonumber \\
  =& I(T_1;S_1^n) \nonumber \\
  =& H(S_1^n)- H(S_1^n| T_1) \nonumber \\
  =& \sum_{i=1}^{n} \Big[ H(S_{1,i}| S_{1,i+1}^n) - H(S_{1,i}|T_1, S_{1,i+1}^n ) \Big] \nonumber \\
  \stackrel{(a)}{=}& \sum_{i=1}^{n} \Big[ H(S_{1,i}) - H(S_{1,i}|T_1, S_{1,i+1}^n) \Big] \nonumber \\
  =& \sum_{i=1}^{n} \Big[ H(S_{1,i}) - H(S_{1,i}| V_{1,i}) \Big] \nonumber \\
  =& \sum_{i=1}^{n} I(S_{1,i}; V_{1,i}),
\end{align}
where $(a)$ follows the fact that $S_{1,i}$ is independent of $S_{1,i+1}^n$.
\begin{align}
  nR \geq& H(T) \nonumber \\
  \geq& H(T|T_1) - H(T|T_1, X^n, S_1^n) \nonumber \\
  =& I(T;X^n,S_1^n|T_1) \nonumber \\
  =& H(X^n, S_1^n| T_1) - H(X^n, S_1^n| T, T_1) \nonumber \\
  =& \sum_{i=1}^n \Big[ H(X_i, S_{1,i}| T_1, X_{i+1}^n, S_{1,i+1}^n) - H(X_i, S_{1,i}| T,T_1, X_{i+1}^n, S_{1,i+1}^n) \Big] \nonumber \\
  \stackrel{(b)}{=}& \sum_{i=1}^n \Big[ H(X_i, S_{1,i}| T_1, S_{1,i+1}^n) - H(X_i, S_{1,i}| T,T_1, X_{i+1}^n, S_{1,i+1}^n) \Big] \nonumber  \\
  \stackrel{(c)}{\geq}& \sum_{i=1}^n \Big[ H(X_i, S_{1,i}| T_1, S_{1,i+1}^n) - H(X_i, S_{1,i}|T,T_1, S_{1,i+1}^n) \Big]\nonumber\\
  =& \sum_{i=1}^n  I(X_i, S_{1,i} ; T| T_1, S_{1,i+1}^n) \nonumber \\
  =& \sum_{i=1}^n  I(X_i, S_{1,i} ; U_i|V_{1,i}) \nonumber \\
  =& \sum_{i=1}^n R \Big( \E{ d\big(X_i, \hat X_i \big)} \Big) \nonumber \\
  \stackrel{(d)}{\geq}& nR \Big( \E{ \frac{1}{n} \sum_{i=1}^n d\big(X_i, \hat X_i \big)} \Big) \nonumber \\
  =& nR(D)
\end{align}
where $(b)$ follows from the fact that $(X_i,S_{1,i})$ is independent of $X_{i+1}^n$ given $(T_1, S_{1,i+1}^n)$, $(c)$ follows from the fact that conditioning reduces entropy and $(d)$ follows from the convexity of $R(D)$ and Jensen's inequality.

Using also the convexity of $R'$ and Jensen's inequality, we can conclude that
\begin{align}
  R' \geq& I(V_1;S_1),\\
  R \geq& I(U;X,S_1|V_1).
\end{align}
It is easy to verify that both Markov chains $V_{1,i} - S_{1,i} - (X_i,S_{2,i})$ and $U_i - (X_i, S_{1,i}, V_{1,i}) - S_{2,i}$ hold. This concludes the converse, and the proof of Theorem \ref{theorem:SC} Case 1$_C$.

\subsection{Proof of Theorem \ref{theorem:SC}, Case 2}
\label{proof:SC:case2}
\begin{figure}[h!]
  \centering
  \small
    \psfrag{A}{\ \ Encoder}
    \psfrag{B}{\ \ Decoder}
    \psfrag{X}{$X^n$}
    \psfrag{x}{$\hat{X}^n$}
    \psfrag{s}{$\hat{X}_i$}
    \psfrag{E}{$S_1^n$}
    \psfrag{e}{$S_{1,i}$}
    \psfrag{R}{$S_2^n$}
    \psfrag{r}{$S_{2,i}$}
    \psfrag{J}{$R'$}
    \psfrag{K}{$R'$}
    \includegraphics[scale = 0.5]{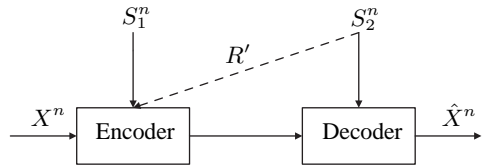}
  \caption{Rate distortion: Case 2. $R_2(D) = \min I(U;X,S_1|V_2) - I(U;S_2|V_2)$, where the minimization is over all PMFs $p(v_2|s_2)p(u|x,s_1,v_2) p(\hat x|u,s_2,v_2)$ such that $R' \geq I(V_2;S_2) - I(V_2;X,S_1)$ and $\mathbb{E}\Big[d(X,\hat X)\Big] \leq D$.}
  \label{fig:RD2}
\end{figure}
This problem is a special case of \cite{DBLP:journals/tit/Kaspi85} for $K=1$, and hence, the proof is omitted.

\section{Proof of Lemma 1}
\label{sec:proof-lemma-1}
We provide here a partial proof of Lemma 1. In the first part we prove the concavity of $C_2^{lb}(R')$ in $R'$ for Case 2, the second part contains the proof that it is enough to take $X$ to be a deterministic function of $(S_1, V_1, U)$ in order to achieve the capacity $C_1(R')$ for Case 1 and in the third part we prove the cardinality bound for Case 1. The proofs of these three parts for the rest of the cases can be derived using the same techniques and therefore are omitted. The proof of Lemma 2 can also be readily concluded using the techniques we use in this appendix and is omitted as well.

{\it Part 1:} We prove here that for Case 2 of the channel capacity problems, the lower bound on the capacity, $C_2^{lb}(R')$, is a concave function of the state information rate, $R'$. Recall that the expression for $C_2^{lb}$  is $C^{lb}_2(R') = \max I(U;Y,S_2|V_2) - I(U;S_1|V_2)$ where the maximization is over all probabilities $p(s_1,s_2) p(v_2|s_2) p(u|s_1,v_2) p(x|u,s_1,v_2) p(y|x,s_1,s_2)$ such that $R' \geq I(V_2;S_2|S_1)$. This means that we want to prove that for any two rates, $R'^{(1)}$ and $R'^{(2)}$, and for any $0 \leq \alpha \leq 1$ and $\bar{\alpha} = 1-\alpha$ the capacity maintains $C^{lb}_2\big(\alpha R'^{(1)} + \bar{\alpha} R'^{(2)} \big) \geq \alpha C^{lb}_2(R'^{(1)}) + \bar{\alpha} C^{lb}_2(R'^{(2)})$.
Let $(U^{(1)}, V_2^{(1)}, X^{(1)}, Y^{(1)})$ and $(U^{(2)}, V_2^{(2)}, X^{(2)}, Y^{(2)})$ be the random variables that meet the conditions on $R'^{(1)}$ and on $R'^{(2)}$ and also achieve $C^{lb}_2(R'^{(1)})$ and $C^{lb}_2(R'^{(2)})$, respectively. Let us introduce the auxiliary random variable $Q \in \{1,2\}$, independent of $S_1,S_2,V_2,U,X$ and $Y$, and distributed according to $\Pr\{Q=1\} = \alpha$ and $\Pr\{Q=2\} = \bar{\alpha}$. Then, consider
\begin{align}
  \label{eq:14}
  \alpha R'^{(1)}+\bar{\alpha} R'^{(2)} & = \alpha\big[I(V_2^{(1)}; S_2)-I(V_2^{(1)};S_1)\big] + \bar{\alpha} \big[I(V_2^{(2)}; S_2)-I(V_2^{(2)};S_1)\big] \nonumber\\
  & \stackrel{(a)}{=} \alpha\big[I(V_2^{(1)}; S_2| Q=1)-I(V_2^{(1)};S_1| Q=1)\big] + \bar{\alpha} \big(I(V_2^{(2)}; S_2| Q=2)-I(V_2^{(2)};S_1| Q=2)\big] \nonumber\\
  & \stackrel{(b)}{=} I(V_2^{(Q)}; S_2|Q)-I(V_2^{(Q)};S_1|Q) \nonumber\\
  & \stackrel{(c)}= I(V_2^{(Q)}, Q; S_2) - I(V_2^{(Q)},Q; S_1),
\end{align}
and
\begin{align}
  \label{eq:15}
  \alpha C^{lb}_2(R'^{(1)}) + \bar{\alpha} C^{lb}_2(R'^{(2)}) =& \alpha\big[I(U^{(1)}; Y^{(1)},S_2|V_2^{(1)})-I(U^{(1)};S_1|V_2^{(1)})\big] \nonumber\\
  &\quad + \bar{\alpha} \big[I(U^{(2)}; Y^{(2)},S_2|V_2^{(2)})-I(U^{(2)};S_1|V_2^{(2)})\big] \nonumber\\
  \stackrel{(d)}{=}& I(U^{(Q)}; Y^{(Q)},S_2|V_2^{(Q)},Q)-I(U^{(Q)};S_1|V_2^{(Q)},Q),
\end{align}
where $(a),(b),(c)$ and $(d)$ all follow from the fact that $Q$ is independent of $(S_1,S_2,V_2,U,X,Y)$ and from $Q$'s probability distribution. Now, let $V_2' = (V_2^{(Q)},Q), U' = U^{(Q)}, Y' = Y^{(Q)}$ and $X' = X^{(Q)}$. Then, following from the equalities above, for any two rates $R'^{(1)}$ and $R'^{(2)}$ and for any $0 \leq \alpha \leq 1$, there exists a set of random variables $(U', V_2', X', Y')$ that maintains
\begin{align}
  \label{eq:24}
  \alpha R'^{(1)} + \bar{\alpha} R'^{(2)} = I(V_2';S_2) - I(V_2';S_1),
\end{align}
and
\begin{align}
  \label{eq:16}
  C^{lb}_2\big(\alpha R'^{(1)} + \bar{\alpha} R'^{(2)} \big)  \geq& I(U';Y',S_2|V_2') - I(U';S_1|V_2') \nonumber \\
  =& \alpha C^{lb}_2(R'^{(1)}) + \bar{\alpha} C^{lb}_2(R'^{(2)}).
\end{align}
This completes the proof of the concavity of $C^{lb}_2(R')$ in $R'$. \qed

{\it Part 2:} We prove here that it is enough to take $X$ to be a deterministic function of $(U,S_1,V_1)$ in order to maximize $I(U;Y,S_2,V_1)-I(U;S_1,V_1)$.  Fix $p(u,v_1|s_1)$. Note that
\begin{align}
  p(y,s_2|u,v_1) &= \sum_{x,s_1} p(s_1|,u,v_1) p(s_2|s_1,v_1,u) p(x|s_1,s_2,v_1,u) p(y|x,s_1,s_2,v_1,u) \nonumber\\
  &= \sum_{x,s_1} p(s_1|u,v_1) p(s_2|s_1) p(x|s_1,v_1,u) p(y|x,s_1,s_2)
\end{align}
is linear in $p(x|u,v_1,s_1)$. This follows from the fact that  fixing $p(u,v_1|s_1)$ also defines $p(s_1|u,v_1)$ and from the following Markov chains  $S_2 - S_1 - (V_1,U)$, $X - (S_1,V_1,U) - S_2$ and $Y - (X,S_1,S_2) - (V_1,U)$. Hence, since $I(U;Y,S_2|V_1)$ is convex in $p(y,s_2|v_1)$ it is also convex in $p(x|u,v_1,s_1)$. Noting also that $I(U;S_1|V_1)$ is constant given a fixed $p(u,v_1|s_1)$, we can conclude that $I(U;Y,S_2|V_1)-I(U;S_1|V_1)$ is convex in $p(x|u,v_1,s_1)$ and, hence, it gets its maximum at the boundaries of $p(x|u,v_1,s_1)$, i.e., when the last is equal $0$ or $1$. This implies that $X$ can be expressed as a deterministic function of $(U,V_1,S_1)$. \qed

{\it Part 3:} We prove now the cardinality bound for Theorem~\ref{theorem:CC}. First, let us recall the support lemma~\cite[p.310]{nla:cat:vn1486008}. Let ${\cal P}({\cal Z})$ be the set of PMFs on the set ${\cal Z}$, and let the set ${\cal P}({\cal Z} | {\cal Q}) \subseteq {\cal P}({\cal Z})$ be a collection of PMFs $p(z|q)$ on ${\cal Z}$ indexed by $q \in \cal Q$. Let $g_j,\ j=1, \dots, k$, be continuous functions on ${\cal P}({\cal Z}|{\cal Q})$. Then, for any $Q \sim F_Q(q)$, there exists  a finite random variable $Q' \sim p(q')$ taking at most $k$ values in $\cal Q$ such that
\begin{eqnarray}
    \mathbb{E} \Big[g_j(p_{Z|Q}(z|Q))\Big] &=& \int_{\cal Q} g_j(p_{Z|Q}(z|q)){\rm d}F(q) \nonumber\\
    &=& \sum_{q'} g_j(p_{Z|q}(z|q'))p(q').
\end{eqnarray}
We first reduce the alphabet size of $V_1$ while considering the alphabet size of $U$ to be constant and then we calculate the cardinality of $U$. Consider the following continuous functions of $p(x,s_1,s_2,u|v_1)$
\begin{align}
  g_j = \left\{
    \begin{array}[c]{l l}
      P_{XS_1S_2|V}(j|v_1), & j \in \big\{1,2,\dots, |\mathcal{X}||\mathcal{S}_1||\mathcal{S}_2|-1\big\},\\
      I(V_1;S_1) - I(V_1;Y,S_2) & j = |\mathcal{X}||\mathcal{S}_1||\mathcal{S}_2|,\\
      I(U;Y,S_2|V_1=v_1) - I(U;S_1|V_1=v_1) & j = |\mathcal{X}||\mathcal{S}_1||\mathcal{S}_2| +1.
\end{array} \right.
\end{align}
Then, by the support lemma, there exists a random variable $V_1'$ with $|\mathcal{V}_1'| \leq |\mathcal{X}||\mathcal{S}_1||\mathcal{S}_2|+1$ such that $p(x,s_1,s_2),\ I(V_1;S_1)-I(V_1;Y,S_2)$ and $I(U;Y,S_2|V_1)-I(U;S_1|V_1)$ are preserved. Notice that the probability of $U$ might have changed due to changing $V_1$; we denote the corresponding $U$ as $U'$. Next, for $v_1' \in \mathcal{V}_1'$ and the corresponding probability $p(v_1')$ that we found in the previous step, we consider $|\mathcal{X}||\mathcal{S}_1||\mathcal{S}_2||\mathcal{V}'_1|$ continuous functions of $p(x,s_1,s_2,v_1'|u')$
\begin{align}
  f_j = \left\{
    \begin{array}[c]{l l}
      P_{XS_1S_2V_1'|U'}(j|u')  & j =  \big\{1,2,\dots,|\mathcal{X}||\mathcal{S}_1||\mathcal{S}_2||\mathcal{V}_1'|-1\big\},\\
      I(U';Y,S_2|V_1')-I(U';S_1|V_1') & j = |\mathcal{X}||\mathcal{S}_1||\mathcal{S}_2||\mathcal{V}_1'|.
    \end{array} \right.
\end{align}
Thus, there exists a random variable $U''$ with $|\mathcal{U}''| \leq |\mathcal{X}||\mathcal{S}_1||\mathcal{S}_2||\mathcal{V}_1'|$ such that the mutual information expressions above and all the desired Markov conditions are preserved. Notice that the expression $I(V_1;S_1)-I(V_1;Y,S_2)$ is being preserved since $p(x,s_1,s_2,v'_1)$ is being preserved.

To conclude, we can bound the cardinality of the auxiliary random variables of Theorem~\ref{theorem:CC} Case 1 by $|\mathcal{V}_1| \leq |\mathcal{X}||\mathcal{S}_1||\mathcal{S}_2|+1$ and $|\mathcal{U}| \leq |\mathcal{X}||\mathcal{S}_1||\mathcal{S}_2||\mathcal{V}_1| \leq |\mathcal{X}||\mathcal{S}_1||\mathcal{S}_2|\big(|\mathcal{X}||\mathcal{S}_1||\mathcal{S}_2|+1\big)$ without limiting the generality of the solution. \qed

\section{Proof of Theorem~\ref{theoremGP}}
\label{sec:proof-GP}

\begin{proof}
  First, let us formulate the Lagrangian for the primal optimization problem defined in~(\ref{eq:21}):
  \begin{align}
    \label{eq:23}
    L\big(\bq, \bmu, \gamma, \blambda\big) = & \sum_{x,s,t} p(x,s) q(t|x) \log \frac{q(t|x)}{Q(t|s)} \nonumber\\
    & + \sum_x \mu_x \Big(\sum_t q(t|x) - 1\Big)\nonumber\\
    & + \gamma\Big( \sum_{x,s,t} p(x,s) q(t|x) d\big(x,t(s)\big) - D\Big) \nonumber\\
    & - \sum_{x,t} \lambda_{x,t} q(t|x),
  \end{align}
with Lagrange multipliers $\bmu, \gamma \geq 0$ and $\blambda \succeq 0$. Recall that $Q(t|s)$ is a marginal distribution that corresponds with $q(t|x)$. i.e.,
\begin{align}
  \label{eq:40}
  Q(t|s) &= \frac{\sum_x p(x,s) q(t|x) }{\sum_s p(x,s) }.
\end{align}

In addition, recall the definition of the Lagrange dual function,
\begin{align}
  g\big(\bmu, \gamma, \blambda \big) = \inf_{\bq} L\big(\bq, \bmu, \gamma, \blambda \big).
\end{align}

In the following proof, we use $\bq_{\bmu,\gamma,\blambda}^*$ to denote the optimal minimizer of the Lagrangian, $L\big(\bq, \bmu, \gamma, \blambda \big)$, for any fixed $\bmu, \gamma,$ and $\blambda$. We also use the notation $g\big(\bmu, \gamma, \blambda \big| \bq_{\mu,\gamma,\lambda}^* \big)$ to denote the Lagrange dual function with $\bq_{\bmu, \gamma, \blambda}^*$ as a constant parameter.

The outline of the proof is as follows:
we first find the PMF $\bq_{\bmu,\gamma,\blambda}^*$, which is the minimizer of the Lagrangian, $L\big(\bq, \bmu, \gamma, \blambda\big)$. We then formulate the Lagrange dual function, $g\big(\bmu, \gamma, \blambda \big| \bq_{\bmu,\gamma,\blambda}^* \big)$, and the Lagrange dual problem, which is to maximize $g$ over $\bmu, \gamma \geq 0$ and $\blambda \succeq 0$. Next, we argue that we can maximize $g$ over  $\bmu, \gamma \geq 0, \blambda \succeq 0$ and, in addition, over any $\bq$ that nullifies the derivative of the Lagrangian (i.e., maintains equation (\ref{eq:25})) without increasing the solution of the Lagrange dual problem. We then note that it is possible to write the Lagrange dual problem with the variable $p(x|s,t)$ instead of $q(t|x)$, where $p(x|s,t)$ is a marginal distribution associated with $q(t|x)$. i.e., $p(x|s,t) = \frac{ p(x,s) q(t|x)}{ \sum_{s,t} p(x,s) q(t|x)}$ is constrained to maintains the Markov chain $T - X - S$. Our next key step is to prove that we can omit the Markov chain constraint without increasing the maximal value of the Lagrange dual problem.
We then conclude our proof by formulating the Lagrange dual problem that we obtained in a geometric programming convex form.

In order to formulate $g\big(\bmu, \gamma, \blambda \big)$, we first find the PMF $\bq_{\mu,\gamma,\lambda}^*$ that minimizes the Lagrangian, $L\big(\bq, \bmu, \gamma, \blambda\big)$, which is a convex function of $\bq$. First, notice that
\begin{align}
  \label{eq:41}
  \frac{\partial}{\partial q(t|x)} &\sum_{x',s',t'} p(x',s') q(t'|x') \log \frac{q(t'|x')}{Q(t'|s')} \nonumber\\
 & \stackrel{(a)}{=} \sum_{s'} p(x,s') \log \frac{q(t|x)}{Q(t|s')} + \sum_{s'} p(x,s') - \sum_{x', s'} p(x',s')q(t|x') \frac{p(x,s')}{p(s')}\frac{1}{Q(t|s')} \nonumber\\
 & = \sum_{s'} p(x,s') \log \frac{q(t|x)}{Q(t|s')} + p(x) - \sum_{s'} p(x,s') \sum_{x'} p(x',s') q(t|x')\frac{1}{p(s')} \frac{1}{Q(t|s')} \nonumber\\
 & \stackrel{(b)}{=} \sum_{s'} p(x,s') \log \frac{q(t|x)}{Q(t|s')} + p(x) - \sum_{s'} p(x,s') \nonumber\\
 & = \sum_{s'} p(x,s') \log \frac{q(t|x)}{Q(t|s')},
\end{align}
where $(a)$ follows from the fact that
\begin{align}
  \label{eq:42}
   \frac{\partial Q(t'|s')}{\partial q(t|x)} &= \frac{\partial}{\partial q(t|x)} \frac{\sum_{x''} p(x'',s') q(t'|x'')}{p(s')} \nonumber\\
   & = \left\{
     \begin{array}[l]{l l}
       \frac{p(x,s')}{p(s')}, & t' = t\\
       0, & t' \neq t
     \end{array} \right.,
\end{align}
and $(b)$ follows from the fact that $p(x,s')$ is independent of $x'$ and the fact that $\sum_{x'} p(x',s') q(t|x')\frac{1}{p(s')} = Q(t|s')$.

Next, we formulate the derivative of the Lagrangian with respect to $q(t|x)$ and we constrain it to be equal to 0.
\begin{align}
  \label{eq:25}
  \frac{\partial L}{\partial q(t|x)} = \sum_s p(x,s) \log \frac{q(t|x)}{Q(t|s)} + \mu_x + \gamma \sum_s p(x,s) d\big(x, t(s)\big) - \lambda_{x,t} = 0.
\end{align}
Using elementary mathematical manipulations we get
\begin{align}
  \label{eq:43}
  \log q(t|x) = \sum_s p(s|x) \Big[ \log Q(t|s) - \frac{\mu_x}{p(x)} - \gamma d\big(x, t(x)\big) - \frac{\lambda_{x,t}}{p(x)} \Big].
\end{align}
Hence,
\begin{align}
  \label{eq:27}
  q_{\bmu,\gamma,\blambda}^*(t|x) = \prod_s \bigg[Q_{\bmu,\gamma,\blambda}^*(t|s) \exp \Big\{-\frac{\mu_x}{p(x)} - \gamma d\big( x,t(s) \big) + \frac{\lambda_{x,t}}{p(x)} \Big\} \bigg]^{p(s|x)}
\end{align}
is an optimal minimizer of the Lagrangian. We get the Lagrange dual function by substituting $\bq$ in the Lagrangian with $\bq_{\bmu,\gamma,\blambda}^*$ that we got in (\ref{eq:27}) and by using constraint (\ref{eq:25}).
\begin{align}
  \label{eq:45}
  g\big(\bmu, \gamma, \blambda \big| \bq_{\bmu,\gamma,\blambda}^*\big) &= \inf_{\bq} L\big(\bq, \bmu, \gamma, \blambda \big) \nonumber\\
  &=L\big(\bq_{\bmu,\gamma,\blambda}^*, \bmu, \gamma, \blambda \big) \nonumber\\
  &= \left\{
    \begin{array}[l]{l l}
-\sum_x \mu_x - \gamma D, & \sum_s p(x,s) \log \frac{q_{\bmu,\gamma,\blambda}^*(t|x)}{Q_{\bmu,\gamma,\blambda}^*(t|s)} + \mu_x + \gamma \sum_s p(x,s) d\big(x, t(s)\big) - \lambda_{x,t} = 0\\
& \qquad \forall x,t\\
-\infty, & {\rm otherwhise}
\end{array} \right.
\end{align}
We get the Lagrange dual problem by making the constraints explicit:
\begin{align}
  \label{eq:46}
  \begin{array}[l]{l l}
    {\rm maximize} & -\sum_x \mu_x - \gamma D\\
    \mbox{subject to} & \sum_s p(x,s) \log \frac{q_{\bmu,\gamma,\blambda}^*(t|x)}{Q_{\bmu,\gamma,\blambda}^*(t|s)} + \mu_x + \gamma \sum_s p(x,s) d\big(x, t(s)\big) - \lambda_{x,t} = 0,\ \forall x,t,\\
                                 & \gamma \geq 0,\\
                                 & \lambda_{x,t} \geq 0,\ \forall x,t,
  \end{array}
\end{align}
where the maximization variables are $\bmu, \gamma$ and $\blambda$ and the constant parameters are the PMFs $\bq_{\bmu,\gamma,\blambda}^*$ and $p(x,s)$, the distortion measure $d\big(x,t(s)\big)$ and the distortion constraint $D$. Notice that since the primal problem, (\ref{eq:21}), is a convex problem with an optimal value of $R(D)$, then the solution of (\ref{eq:46}) is a lower bound on $R(D)$ \cite[Chapter 5.2.2]{Boyd:2004:CO:993483}, and, if Slater's condition holds, then strong duality holds and the optimal value of (\ref{eq:46}) is $R(D)$.

Now, notice that any $\bq$ that maintains the first inequality constraint in~(\ref{eq:46}) nullifies the derivative of the Lagrangian and, hence, results in the same value when placed in the Lagrangian; this value is exactly the Lagrange dual function. Therefore, since $g$ gets the same value for any $\bq$ that maintains the constraint (\ref{eq:25}), we can maximize $g$ over all PMFs $\bq$ that maintain constraint (\ref{eq:25}) without changing $g$'s value. Consequently, the Lagrange dual problem in (\ref{eq:46}) becomes:
\begin{align}
  \label{eq:47}
  \begin{array}[l]{l l}
    {\rm maximize} & -\sum_x \mu_x - \gamma D\\
    \mbox{subject to} & \sum_s p(x,s) \log \frac{q(t|x)}{Q(t|s)} + \mu_x + \gamma \sum_s p(x,s) d\big(x, t(s)\big) - \lambda_{x,t} = 0,\ \forall x,t,\\
    & \gamma \geq 0,\\
    & \lambda_{x,t} \geq 0,\ \forall x,t,\\
    & \sum_t q(t|x) = 1,\ \forall x,
  \end{array}
\end{align}
where the maximization variables are $\bmu, \gamma, \blambda$ and $\bq$ and the constant parameters are $p(x,s)$, $d\big(x,t(s)\big)$ and $D$.

Next, combining (\ref{eq:27}) and the fact that $Q(t|s) \geq 0$, we get that we can replace the first constraint in (\ref{eq:47}) with
\begin{align}
  \label{eq:26}
  q(t|x)  = \prod_s \bigg[Q(t|s) \exp \Big\{-\frac{\mu_x}{p(x)} - \gamma d\big( x,t(s) \big) + \frac{\lambda_{x,t}}{p(x)}\Big\} \bigg]^{p(s|x)},\ \forall x,t.
\end{align}

Since $q(t|x)$ is independent of $s$, we can state that
\begin{align}
  \label{eq:128}
  1  = \prod_s \bigg[\frac{Q(t|s)}{q(t|x)} \exp \Big\{-\frac{\mu_x}{p(x)} - \gamma d\big( x,t(s) \big) + \frac{\lambda_{x,t}}{p(x)} \Big\} \bigg]^{p(s|x)}.
\end{align}
Let us denote $\alpha_x = -\frac{\mu_x}{p(x)}$ and note that $\frac{Q(t|s)}{q(t|x)} = \frac{p(x|s)Q(t|s)}{p(t,x|s)} = \frac{p(x|s)}{p(x|s,t)}$, where $p(x|s,t)$ maintains the Markov chain $T - X - S$. Therefore, equation (\ref{eq:128}) becomes
\begin{align}
  \label{eq:28}
  1  = \prod_s \bigg[p(x|s) \exp \Big\{\alpha_x - \gamma d\big( x,t(s) \big) + \frac{\lambda_{x,t}}{p(x)} - \log p(x|s,t)\Big\} \bigg]^{p(s|x)},
\end{align}
for all $x,t$, and the Lagrange dual problem can be reformulated as
\begin{align}
  \label{eq:39}
  \begin{array}[l]{l l}
    {\rm maximize} & \sum_x \alpha_x p(x) - \gamma D\\
    \mbox{subject to} & 1  = \prod_s \bigg[p(x|s) \exp \Big\{\alpha_x - \gamma d\big( x,t(s) \big) + \frac{\lambda_{x,t}}{p(x)} - \log p(x|s,t)\Big\} \bigg]^{p(s|x)},\ \forall x,t\\
                                 & \gamma \geq 0,\\
                                 & \sum_t p(x|s,t) = 1,\ \forall x,\\
                                 & p(x|s,t) \mbox{ maintain the Markov chain } T-X-S,
  \end{array}
\end{align}
where the variables of the maximization are $\balpha, \gamma, \blambda$ and $\bp \in \mathbb{R}^{|\mathcal{X}| |\mathcal{S}| |\mathcal{T}|}$, which is the set of all $p(x|s,t)$ for all $x \in \mathcal{X}, s \in \mathcal{S}$ and $t \in \mathcal{T}$, and the constant variables are $p(x,s)$, $d\big(x,t(s)\big)$ and $D$. Notice that (\ref{eq:39}) is not a convex problem anymore, since the constraint functions are not convex. We deal with this problem in the following steps by using geometric programming principles.

Next, we want to prove that it is possible to maximize (\ref{eq:39}) over any PMF, $\bp$. i.e., we want to prove that dropping the last constraint in (\ref{eq:39}) does not change the validity of the solution.

First, since (\ref{eq:39}) is an equivalent Lagrange dual problem, then, according to~\cite[Chapter 5.2.2]{Boyd:2004:CO:993483}, we can state that for any choice of $\balpha, \gamma$ and $\blambda$ it yields a lower bound on $R(D)$. Furthermore, according to~\cite[Chapter 5.2.3]{Boyd:2004:CO:993483}, if Slater's condition holds, then the solution of (\ref{eq:39}) coincides with $R(D)$, which is the optimal solution of the primal problem. Now, dropping the constraint that the Markov chain $T - X - S$ must hold, necessarily allows the optimal solution of (\ref{eq:39}) to be greater than or equal to the solution where $T - X - S$ holds. We are left to prove that maximizing over any PMF, $\bp$, cannot exceed $R(D)$. Let us place $p(x|s,t) = \frac{p(t|x,s)p(x|s)}{p(t|s)}$ in (\ref{eq:28}) and look at the following inequalities:
\begin{align}
  1  =& \prod_s \bigg[p(x|s) \exp \Big\{\alpha_x - \gamma d\big( x,t(s) \big) +\frac{\lambda_{x,t}}{p(x)} - \log \frac{p(t|x,s)p(x|s)}{p(t|s)}\Big\} \bigg]^{p(s|x)} \nonumber\\
  =& \prod_s \bigg[\exp \Big\{\log p(x|s) + \alpha_x - \gamma d\big( x,t(s) \big) +\frac{\lambda_{x,t}}{p(x)} - \log \frac{p(t|x,s)p(x|s)}{p(t|s)}\Big\} \bigg]^{p(s|x)} \nonumber\\
  =& \exp\Big\{ \alpha_x - \gamma \sum_s p(s|x) d\big(x, t(s)\big) +\frac{\lambda_{x,t}}{p(x)} - \sum_s p(s|x) \log p(t|x,s) + \sum_s p(s|x) \log p(t|s) \Big\} \nonumber\\
  \stackrel{(a)}{\geq}&  \exp\Big\{ \alpha_x - \gamma \sum_s p(s|x) d\big(x, t(s)\big) +\frac{\lambda_{x,t}}{p(x)} - \log \Big(\sum_s p(s|x)p(t|x,s)\Big) + \sum_s p(s|x) \log p(t|s) \Big\} \nonumber\\
  =& \exp\Big\{ \alpha_x - \gamma \sum_s p(s|x) d\big(x, t(s)\big) +\frac{\lambda_{x,t}}{p(x)} - \log p(t|x) + \sum_s p(s|x) \log p(t|s) \Big\} \nonumber\\
  \stackrel{(b)}{=}& \exp\Big\{ \alpha_x - \gamma \sum_s p(s|x) d\big(x, t(s)\big) + \frac{\lambda_{x,t}}{p(x)} - \sum_s p(s|x) \log \frac{p(t|x)p(x|s)}{p(t|s)} + \sum_s p(s|x) \log p(x|s) \Big\} \nonumber\\
 =& \prod_s \bigg[p(x|s) \exp \Big\{\alpha_x - \gamma d\big( x,t(s) \big) +\frac{\lambda_{x,t}}{p(x)} - \log \frac{p(t|x)p(x|s)}{p(t|s)}\Big\} \bigg]^{p(s|x)},
\end{align}
where $(a)$ follows from Jensen's inequality and $(b)$ follows from the fact that $p(t|x)$ is independent of $s$. Notice that by reducing the value of $\sum_s p(s|x) \log p(t|x,s)$, we allow $\alpha_x - \gamma \sum_s p(s|x) d\big(x,t(s)\big)$ to be greater and, hence, we improve our maximum. Therefore, for any $p(x|s,t) = \frac{p(t|s)}{p(t|x,s)p(x|s)}$, we can take $p'(x|s,t) = \frac{p(t|s)}{p(x|s)\sum_{s'} p(s'|x)p(t|x,s')}$, which satisfies the Markov chain $T -X -S$, and that the maximum over $p(t|x) = \sum_s p(s|x)p(t|x,s)$ would be equal to or greater than the maximum over $p(x|s,t)$. This, and the fact that maximizing over $p(t|x)$ cannot exceed $R(D)$ and that $R(D)$ can be achieved by using $p^*(x|s,t)$ that corresponds to $q^*(t|x)$, prove that, indeed, we can maximize over $p(x|s,t)$ without changing the result of the maximization. Therefore, our dual problem now becomes
\begin{align}
  \label{eq:31}
  \begin{array}[l]{l l}
    \mbox{maximize} &  \sum_x \alpha_x p(x) - \gamma D\\
    \mbox{subject to} & \prod_s \bigg[p(x|s) \exp \Big\{\alpha_x - \gamma d\big( x,t(s) \big) + \frac{\lambda_{x,t}}{p(x)} - \log p(x|s,t)\Big\} \bigg]^{p(s|x)} = 1 \quad \forall x,t,\\
    & \sum_x p(x|s,t) = 1 \quad \forall s,t\\
    & \gamma \geq 0.
  \end{array}
\end{align}

In order to make the problem convex, we need to convert the equality constraints that are not affine into inequality constraints. Let us go back to (\ref{eq:28}); since $\lambda_{x,t} \geq 0$ for all $x$ and $t$ and since $p(x,s) \geq 0$, the constraint (\ref{eq:28}) can be replaced by
\begin{align}
  \label{eq:57}
  1  \geq \prod_s \bigg[p(x|s) \exp \Big\{\alpha_x - \gamma d\big( x,t(s) \big) - \log p(x|s,t)\Big\} \bigg]^{p(s|x)}
\end{align}
without changing the solution of~(\ref{eq:39}). Next, notice that there is a tradeoff between $-\log p(x|s,t)$ and $\alpha_x - \gamma d\big( x,t(s) \big)$. Therefore, we expect $-\log p(x|s,t)$ to be as small as possible to allow $\alpha_x - \gamma d\big( x,t(s) \big)$ to be as large as possible. Hence, we can replace the constraint
\begin{align}
  \label{eq:58}
  \sum_x p(x|s,t) &= 1 \quad \forall s,t,
\end{align}
which is equivalent to
\begin{align}
  \label{eq:59}
  \sum_x \exp\big\{\log p(x|s,t)\big\} &= 1 \quad \forall s,t,
\end{align}
with the weaker constraint
\begin{align}
  \label{eq:60}
  \sum_x \exp\big\{\log p(x|s,t)\big\} &\leq 1 \quad \forall s,t,
\end{align}
without changing the result of the maximization. We denote $y_{x,t,s} = \log p(x|s,t)$ and rewrite the dual problem as
\begin{align}
  \label{eq:32}
  \begin{array}[l]{l l}
    \mbox{maximize} & \sum_x \alpha_x p(x) - \gamma D\\
    \mbox{subject to} & \prod_s \bigg[p(x|s) \exp \Big\{\alpha_x - \gamma d\big( x,t(s) \big) - y_{x,s,t}\Big\} \bigg]^{p(s|x)} \leq 1\quad \forall x,t,\\
    & \sum_x \exp\big\{y_{x,s,t}\big\} \leq 1 \quad \forall s,t,\\
    & \gamma \geq 0,
  \end{array}
\end{align}
where the variables of the maximization are ${\bm \alpha}, \gamma$ and ${\bf y}$ and the constant parameters are the PMF, $p(x,s)$, the distortion measure, $d\big(x, t(s)\big)$, and the distortion constraint, $D$.

Lastly, we present the dual problem in a geometric programming convex form by taking $\log(\cdot)$ on the first two inequality constraints:
\begin{align}
  \label{eq:49}
  \begin{array}[l]{l l}
    \mbox{maximize} &  \sum_x \alpha_x p(x) - \gamma D\\
    \mbox{subject to} &  \alpha_x + \sum_s p(s|x) \bigg[ \log p(x|s) - \gamma d\big( x,t(s) \big) - y_{x,s,t}\bigg] \leq 0\quad \forall x,t,\\
    & \log \left( \sum_x \exp\big\{y_{x,s,t}\big\} \right) \leq 0 \quad \forall s,t,\\
    & \gamma \geq 0,
  \end{array}
\end{align}
where the variables of the maximization are $\balpha, \gamma$ and ${\bf y}$ and the constant parameters are $p(x,s), d\big(x,t(s)\big)$ and $D$.
\end{proof}

\section{Proofs for Section~\ref{sec:algorithm}}
\label{sec:proofs-algorithm}

\subsection{Proof of Lemma~\ref{concavityOfJ}}
\label{sec:proof-concavityOfJ}
\begin{proof}
  For $0 \leq \alpha \leq 1$ and $\bar \alpha = 1-\alpha$
  \begin{align}
    \label{eq:108}
    J_w(\alpha q_1 + \bar \alpha q_2,& \alpha Q_1 + \bar \alpha Q_2) = \sum_{s_1,s_2,v_2,t,y} p(s_1,s_2) w(v_2|s_2) p(y|t,s_1,s_2,v_2)  \Big(\alpha q_1 + \bar \alpha q_2\Big)  \log \frac{\alpha Q_1+\bar \alpha Q_2}{\alpha q_1+ \bar \alpha q_2} \nonumber\\
    &\stackrel{(a)}{\leq} \sum_{s_1,s_2,v_2,t,y} p(s_1,s_2) w(v_2|s_2)p(y|t,s_1,s_2,v_2)  \Big( \alpha q_1  \log \frac{Q_1}{q_1} + \bar \alpha q_2 \log \frac{Q_2}{q_2} \Big) \nonumber\\
    &= \alpha J_w(q_1,Q_1) + \bar \alpha J_w(q_2,Q_2),
  \end{align}
  where $(a)$ follows from the log-sum inequality:
  \begin{align}
    \label{eq:109}
    \sum_i a_i \log \frac{a_i}{b_i} &\geq a \log \frac{a}{b},
  \end{align}
  for $\sum_i a_i = a$ and  $\sum_i b_i = b$.
\end{proof}

\subsection{Proof of Lemma~\ref{maximizing_q}}
\label{sec:proof-maximizing_q}

 \begin{proof}
   Let us calculate $q^*$ using the KKT conditions. We want to maximize $J_w(q^*,Q)$ over $q^*$, where for all $t,s_1$ and $v_2$, $0 \leq q^*(t|s_1,v_2) \leq 1$ and $\sum_{t'} q^*(t'|s_1,v_2) = 1$.

   For fixed $s_1$ and $v_2$,
   \begin{align}
     \label{eq:101}
     0 &= \frac{\partial}{\partial q^*} \Big(  J_w(q^*,Q) + \big( 1-\sum_t q^*(t|s_1,v_2) \big) \nu_{s_1,v_2} \Big)\\
     &= \sum_{s_2,y} p(s_1,s_2) w(v_2|s_2) p(y|t,s_1,s_2,v_2) \Big( \log \frac{Q(t|y,s_2,v_2)}{q^*(t|s_1,v_2)} - 1\Big) - \nu_{s_1,v_2} \label{eq:nu},
   \end{align}
   divide by $p(s_1,v_2)$,
   \begin{align}
   0 &= -\log q^*(t|s_1,v_2) + \frac{\sum_{s_2,y} p(s_1,s_2) w(v_2|s_2) p(y|t,s_1,s_2,v_2)}{p(s_1,v_2)} \log Q(t|y,s_2,v_2) - 1+\frac{\nu_{s_1v_2}}{p(s_1,v_2)},
   \end{align}
   define $-1+\frac{\nu_{s_1v_2}}{p(s_1,v_2)} = \log \nu'_{s_1,v_2}$, hence
   \begin{align}
     q^*(t|s_1,v_2) &= \nu'_{s_1,v_2}\prod_{s_2,y} Q(t|y,s_2,v_2)^{p(s_2|s_1,v_2) p(y|t,s_1,s_2,v_2)},
   \end{align}
   and from the constraint $\sum_{t'} q^*(t'|s_1,v_2) = 1$ we get that
   \begin{align}
     \label{eq:110}
     q^*(t|s_1,v_2) = \frac{\prod_{s_2,y} Q(t|y,s_2,v_2)^{p(s_2|s_1,v_2) p(y|t,s_1,s_2,v_2)}}{ \sum_{t'} \prod_{s_2,y} Q(t'|y,s_2,v_2)^{p(s_2|s_1,v_2) p(y|t',s_1,s_2,v_2)}}.
   \end{align}
 \end{proof}

\subsection{Proof of Lemma~\ref{U_w}}
\label{sec:proof-U_w}

The proof for this lemma is done in three steps: first, we prove that $U_w(q_1)$ is greater than or equal to $J_w(q_0,Q_0^*)$ for any two PMFs $q_0(t|s_1,v_2)$ and $q_1(t|s_1,v_2)$, then, we use Lemma~3 and Lemma~5 to state that for the optimal PMF, $q_c(t|s_1,v_2)$, $C^{lb}_{2,w} = J_w(q_c,Q_c^*)$, and, therefore, $U_w(q)$ is an upper bound of $C^{lb}_{2,w}$ for every $q(t|s_1,v_2)$. Thirdly, we prove that $U_w(q)$ converges to $C_{2,w}^{lb}$.
\begin{proof}
  Consider any two PMFs, $q_0(t|s_1,v_2)$ and $q_1(t|s_1,v_2)$, their corresponding $\{p_0(s_1,s_2,v_2,t,y), Q_0^*(t|y,s_2,v_2)\}$ and $\{p_1(s_1,s_2,v_2,t,y), Q_1^*(t|y,s_2,v_2)\}$, respectively, according to (\ref{eq:p}) and (\ref{eq:Q}) and consider also the following inequalities:
  \begin{align}
    \label{eq:113}
     & \sum_{s_1,s_2,v_2,t,y} p_0(s_1,s_2,v_2,t,y) \log \frac{Q_1^*(t|y,s_2,v_2)}{q_1(t|s_1,v_2)} - J_w(q_0,Q_0^*) \nonumber\\
    = & \sum_{s_1,s_2,v_2,t,y} p_0(s_1,s_2,v_2,t,y) \Big(\log \frac{Q_1^*(t|y,s_2,v_2)}{q_1(t|s_1,v_2)} - \log \frac{Q_0^*(t|y,s_2,v_2)}{q_0(t|s_1,v_2)} \Big) \nonumber\\
    = & \sum_{s_1,s_2,v_2,t,y} p_0(s_1,s_2,v_2,t,y) \log \Big( \frac{Q_1^*(t|y,s_2,v_2)}{Q_0^*(t|y,s_2,v_2)} \frac{q_0(t|s_1,v_2)}{q_1(t|s_1,v_2)} \Big) \nonumber\\
    = & \D{q_0(t|s_1,v_2) \big\| q_1(t|s_1,v_2)} - \D{Q_0^*(t|y,s_2,v_2) \big\| Q_1^*(t|y,s_2,v_2)} \nonumber\\
    \stackrel{(a)}{=} & \D{q_0(t|s_1,s_2,v_2) p(y|t,s_1,s_2,v_2) p(s_1,s_2) w(v_2|s_2) \big\| q_1(t|s_1,s_2,v_2) p(y|t,s_1,s_2,v_2) p(s_1,s_2) w(v_2|s_2)} \nonumber \\
    & \quad -  \D{Q_0^*(t|y,s_2,v_2) \big\| Q_1^*(t|y,s_2,v_2)} \nonumber\\
    = & \D{p_0(s_1,s_2,v_2,t,y) \big\| p_1(s_1,s_2,v_2,t,y)} - \D{Q_0^*(t|y,s_2,v_2) \big\| Q_1^*(t|y,s_2,v_2)} \nonumber\\
    \stackrel{(b)}{=} & \D{p_0(s_2,v_2,y) Q_0^*(t|y,s_2,v_2) p_0(s_1|s_2,v_2,t,y) \big\| p_1(s_2,v_2,y) Q_1^*(t|y,s_2,v_2) p_1(s_1|s_2,v_2,t,y)} \nonumber\\ &\quad - \D{Q_0^*(t|y,s_2,v_2) \big\| Q_1^*(t|y,s_2,v_2)} \nonumber\\
    = & \D{p_0(s_2,v_2,y) \big\| p_1(s_2,v_2,y)} + \D{p_0(s_1|s_2,v_2,t,y) \big\| p_1(s_1|s_2,v_2,t,y)} \nonumber\\
    \stackrel{(c)}{=} & \geq 0,
  \end{align}
where $\D{ \cdot \big\| \cdot}$ is the K-L divergence, $p_j(s_2,v_2,y)$ and $p_j(s_1|s_2,v_2,t,y)$ are marginal distributions of $p_j(s_1,s_2,v_2,t,y)$ for $j=0,1$, $(a)$ follows from the fact that $T$ is independent of $S_2$ given $(S_1,V_2)$ and from the K-L divergence properties, $(b)$ follows from the fact that $Q_j^*(t|y,s_2,v_2)$ is a marginal distribution of $p_j(s_1,s_2,v_2,t,y)$ for $j=0,1$ and $(c)$ follows from the fact that $\D{\cdot \big\| \cdot } \geq 0$ always.

Thus,
\begin{align}
  \label{eq:911}
  J(q_0,Q_0^*) \leq & \sum_{s_1,s_2,v_2,t,y} p_0(s_1,s_2,v_2,t,y) \log \frac{Q_1^*(t|y,s_2,v_2)}{q_1(t|s_1,v_2)} \nonumber\\
  = & \sum_{s_1,s_2,v_2,t,y} p(s_1,s_2)w(v_2|s_2)q_0(t|s_1,v_2)p(y|t, s_1,s_2,v_2) \log \frac{Q_1^*(t|y,s_2,v_2)}{q_1(t|s_1,v_2)} \nonumber\\
  = & \sum_{s_1,v_2} p(s_1,v_2) \sum_t q_0(t|s_1,v_2) \sum_{s_2} p(s_2|s_1,v_2) \sum_y p(y|t,s_1,s_2,v_2) \log \frac{Q_1^*(t|y,s_2,v_2)}{q_1(t|s_1,v_2)} \nonumber\\
  \leq & \sum_{s_1,v_2} p(s_1,v_2) \max_{t'} \sum_{s_2} p(s_2|s_1,v_2) \sum_y p(y|t',s_1,s_2,v_2) \log \frac{Q_1^*(t'|y,s_2,v_2)}{q_1(t'|s_1,v_2)} \nonumber\\
  = & U_w(q_1).
\end{align}
We proved that $U_w(q_1)$ is greater than or equal to $J_w(q_0,Q_0^*)$ for any choice of $q_0(t|s_2,v_2)$ and $q_1(t|s_1,v_2)$. Therefore, by taking $q_0(t|s_1,v_2)$ to be the distribution that achieves $C_{2,w}^{lb}$ and by considering Lemma~3 and Lemma~5, we conclude that $U_w(q) \geq C_{w,2}$ for any choice of $q(t|s_1,v_2)$.

In order to prove that $U_w(q)$ converges to $C_{2,w}^{lb}$ let us rewrite equation (\ref{eq:nu}) as
  \begin{align}
    \label{eq:112}
    \sum_{s_2,y} p(s_2|s_1,v_2) p(y|t,s_1,s_2,v_2) \log \frac{Q(t|y,s_2,v_2)}{q^*(t|s_1,v_2)} = \nu'_{s_1,v_2}.
  \end{align}
  We can see that for a fixed $Q$, the right hand side of the equation is independent of $t$. Considering also
  \begin{align}
    \label{eq:8}
    J_w(q,Q) =& \sum_{s_1,s_2,v_2,t,y} p(s_1,s_2)w(v_2|s_2)q(t|s_1,v_2)p(y|t,s_1,s_2,v_2) \log \frac{Q^(t|y,s_2,v_2)}{q(t|s_1,v_2)} \nonumber\\
    \leq& \sum_{s_1,v_2}p(s_1,v_2) \max_{t'} \sum_{s_2} p(s_2|s_1,v_2) \sum_y p(y|t',s_1,s_2,v_2) \log \frac{Q^*(t'|y,s_2,v_2)}{q(t'|s_1,v_2)},
  \end{align}
   we can conclude that the equation holds when the PMF $q$ is the PMF that achieves $C_{2,w}^{lb}$.
\end{proof}

 \bibliographystyle{IEEEtran}
 \bibliography{Article-AvihayShirazi}

\begin{thebibliography}{10}
\providecommand{\url}[1]{#1}
\csname url@samestyle\endcsname
\providecommand{\newblock}{\relax}
\providecommand{\bibinfo}[2]{#2}
\providecommand{\BIBentrySTDinterwordspacing}{\spaceskip=0pt\relax}
\providecommand{\BIBentryALTinterwordstretchfactor}{4}
\providecommand{\BIBentryALTinterwordspacing}{\spaceskip=\fontdimen2\font plus
\BIBentryALTinterwordstretchfactor\fontdimen3\font minus
  \fontdimen4\font\relax}
\providecommand{\BIBforeignlanguage}[2]{{%
\expandafter\ifx\csname l@#1\endcsname\relax
\typeout{** WARNING: IEEEtran.bst: No hyphenation pattern has been}%
\typeout{** loaded for the language `#1'. Using the pattern for}%
\typeout{** the default language instead.}%
\else
\language=\csname l@#1\endcsname
\fi
#2}}
\providecommand{\BIBdecl}{\relax}
\BIBdecl

\bibitem{1055508}
A.~Wyner and J.~Ziv, ``The rate-distortion function for source coding with side
  information at the decoder,'' \emph{Information Theory, IEEE Transactions
  on}, vol.~22, no.~1, pp. 1 -- 10, jan 1976.

\bibitem{4608994}
Y.~Steinberg, ``Coding for channels with rate-limited side information at the
  decoder, with applications,'' \emph{Information Theory, IEEE Transactions
  on}, vol.~54, no.~9, pp. 4283 --4295, sept. 2008.

\bibitem{citeulike:437050}
S.~I. Gel'fand and M.~S. Pinsker, ``Coding for channel with random
  parameters,'' \emph{Problems of Control Theory}, vol.~9, no.~1, pp. 19--31,
  1980.

\bibitem{1662378}
C.~E. Shannon, ``Channels with side information at the transmitter,'' \emph{IBM
  J. Res. Dev.}, vol.~2, no.~4, pp. 289--293, 1958.

\bibitem{DBLP:journals/tit/HeegardG83}
C.~Heegard and A.~A.~E. Gamal, ``On the capacity of computer memory with
  defects,'' \emph{IEEE Transactions on Information Theory}, vol.~29, no.~5,
  pp. 731--739, 1983.

\bibitem{Cover:2006:DCC:2263239.2266915}
\BIBentryALTinterwordspacing
T.~M. Cover and M.~Chiang, ``Duality between channel capacity and rate
  distortion with two-sided state information,'' \emph{IEEE Trans. Inf.
  Theor.}, vol.~48, no.~6, pp. 1629--1638, Sep. 2006. [Online]. Available:
  \url{http://dx.doi.org/10.1109/TIT.2002.1003843}
\BIBentrySTDinterwordspacing

\bibitem{1424319}
A.~Rosenzweig, Y.~Steinberg, and S.~Shamai, ``On channels with partial channel
  state information at the transmitter,'' \emph{Information Theory, IEEE
  Transactions on}, vol.~51, no.~5, pp. 1817 -- 1830, may 2005.

\bibitem{4385766}
Y.~Cemal and Y.~Steinberg, ``Coding problems for channels with partial state
  information at the transmitter,'' \emph{Information Theory, IEEE Transactions
  on}, vol.~53, no.~12, pp. 4521 --4536, dec. 2007.

\bibitem{1454716}
G.~Keshet, Y.~Steinberg, and N.~Merhav, ``Channel coding in the presence of
  side information,'' \emph{Found. Trends Commun. Inf. Theory}, vol.~4, no.~6,
  pp. 445--586, 2007.

\bibitem{DBLP:journals/tit/Kaspi85}
A.~H. Kaspi, ``Two-way source coding with a fidelity criterion,'' \emph{IEEE
  Transactions on Information Theory}, vol.~31, no.~6, pp. 735--740, 1985.

\bibitem{5466529}
H.~Permuter, Y.~Steinberg, and T.~Weissman, ``Two-way source coding with a
  helper,'' \emph{Information Theory, IEEE Transactions on}, vol.~56, no.~6,
  pp. 2905 --2919, june 2010.

\bibitem{DBLP:journals/tit/WeissmanG06}
T.~Weissman and A.~E. Gamal, ``Source coding with limited-look-ahead side
  information at the decoder,'' \emph{IEEE Transactions on Information Theory},
  vol.~52, no.~12, pp. 5218--5239, 2006.

\bibitem{DBLP:journals/tit/WeissmanM05}
T.~Weissman and N.~Merhav, ``On causal source codes with side information,''
  \emph{IEEE Transactions on Information Theory}, vol.~51, no.~11, pp.
  4003--4013, 2005.

\bibitem{Shannon59}
C.~E. Shannon, ``Coding theorems for a discrete source with a fidelity
  criterion,'' vol. 7, part 4, pp. 142--163, Mar. 1959.

\bibitem{DBLP:journals/tit/PradhanCR03}
S.~S. Pradhan, J.~Chou, and K.~Ramchandran, ``Duality between source coding and
  channel coding and its extension to the side information case,'' \emph{IEEE
  Transactions on Information Theory}, vol.~49, no.~5, pp. 1181--1203, 2003.

\bibitem{Zamir:2006:NLC:2263239.2266893}
\BIBentryALTinterwordspacing
R.~Zamir, S.~Shamai, and U.~Erez, ``Nested linear/lattice codes for structured
  multiterminal binning,'' \emph{IEEE Trans. Inf. Theor.}, vol.~48, no.~6, pp.
  1250--1276, Sep. 2006. [Online]. Available:
  \url{http://dx.doi.org/10.1109/TIT.2002.1003821}
\BIBentrySTDinterwordspacing

\bibitem{911306}
J.~Su, J.~Eggers, and B.~Girod, ``Illustration of the duality between channel
  coding and rate distortion with side information,'' in \emph{Signals, Systems
  and Computers, 2000. Conference Record of the Thirty-Fourth Asilomar
  Conference on}, vol.~2, 29 2000-nov. 1 2000, pp. 1841 --1845 vol.2.

\bibitem{Blahut72computationof}
R.~E. Blahut, ``Computation of channel capacity and rate-distortion
  functions,'' \emph{IEEE Trans. Inform. Theory}, vol.~18, pp. 460--473, 1972.

\bibitem{Arimoto1972a}
S.~Arimoto, ``An algorithm for computing the capacity of arbitrary discrete
  memorylesschannels,'' \emph{IEEE Trans. Inform. Theory}, vol.~18, pp. 14--20,
  1972.

\bibitem{Willems1983}
F.~M.~J. Willems, ``Computation of the wyner-ziv rate-distortion function,''
  Research Report, July 1983.

\bibitem{1365218}
F.~Dupuis, W.~Yu, and F.~Willems, ``Blahut-arimoto algorithms for computing
  channel capacity and rate-distortion with side information,'' in
  \emph{Information Theory, 2004. ISIT 2004. Proceedings. International
  Symposium on}, june-2 july 2004, p. 179.

\bibitem{DBLP:journals/tit/ChengSX05}
S.~Cheng, V.~Stankovic, and Z.~Xiong, ``Computing the channel capacity and
  rate-distortion function with two-sided state information,'' \emph{IEEE
  Transactions on Information Theory}, vol.~51, no.~12, pp. 4418--4425, 2005.

\bibitem{5205855}
O.~Sumszyk and Y.~Steinberg, ``Information embedding with reversible
  stegotext,'' in \emph{Information Theory, 2009. ISIT 2009. IEEE International
  Symposium on}, 28 2009-july 3 2009, pp. 2728 --2732.

\bibitem{DBLP:journals/tit/NaissP13}
I.~Naiss and H.~H. Permuter, ``Extension of the blahut-arimoto algorithm for
  maximizing directed information,'' \emph{IEEE Transactions on Information
  Theory}, vol.~59, no.~1, pp. 204--222, 2013.

\bibitem{Chiang04geometricprogramming}
M.~Chiang, S.~Boyd, and A.~Overview, ``Geometric programming duals of channel
  capacity and rate distortion,'' \emph{IEEE Trans. Inform. Theory}, vol.~50,
  pp. 245--258, 2004.

\bibitem{DBLP:journals/tit/NaissP13a}
I.~Naiss and H.~H. Permuter, ``Computable bounds for rate distortion with feed
  forward for stationary and ergodic sources,'' \emph{IEEE Transactions on
  Information Theory}, vol.~59, no.~2, pp. 760--781, 2013.

\bibitem{cove_thom_91}
T.~M. Cover and J.~A. Thomas, \emph{Elements of Information Theory}.\hskip 1em
  plus 0.5em minus 0.4em\relax John Wiley \& sons, 1991.

\bibitem{Boyd:2004:CO:993483}
S.~Boyd and L.~Vandenberghe, \emph{Convex Optimization}.\hskip 1em plus 0.5em
  minus 0.4em\relax New York, NY, USA: Cambridge University Press, 2004.

\bibitem{Yeung:2008:ITN:1457455}
R.~W. Yeung, \emph{Information Theory and Network Coding}, 1st~ed.\hskip 1em
  plus 0.5em minus 0.4em\relax Springer Publishing Company, Incorporated, 2008.

\bibitem{Berger}
T.~Berger, ``Multiterminal source coding,'' in \emph{Information Theory
  Approach to Communications}, G.~Longo, Ed.\hskip 1em plus 0.5em minus
  0.4em\relax CSIM Course and Lectures, 1978, pp. 171--231.

\bibitem{nla:cat:vn1486008}
I.~Csiszar and J.~Korner, \emph{Information theory : Coding theorems for
  discrete memoryless systems}.\hskip 1em plus 0.5em minus 0.4em\relax Academic
  Press ; Akademiai Kiado, New York : Budapest, 1981.

\end{thebibliography}

\end{document}